\documentclass[11pt]{article}
\usepackage{palatino}
\usepackage{amssymb,amsfonts,amsmath,amsthm,bm}
\usepackage{algorithm2e,algorithmic,framed,url}
\usepackage{graphicx}
\usepackage{array}

\newcommand{\poly}{{\mathrm{poly}}}
\newcommand{\eps}{\varepsilon}
\newcommand{\1}{\uppercase\expandafter{\romannumeral1}}
\newcommand{\2}{\uppercase\expandafter{\romannumeral2}}

\long\def\symbolfootnote[#1]#2{\begingroup%
\def\thefootnote{\fnsymbol{footnote}}\footnote[#1]{#2}\endgroup}

\newcommand{\Expect}[1]{\mbox{}{\mathbb{E}}\left[#1\right]}

\newcommand{\FNorm }[1]{\mbox{}\|#1\|_\mathrm{F}  }
\newcommand{\FNormS}[1]{\mbox{}\|#1\|_\mathrm{F}^2}

\newcommand{\TNorm }[1]{\mbox{}\|#1\|_2  }
\newcommand{\TNormS}[1]{\mbox{}\|#1\|_2^2}

\newcommand{\XNorm }[1]{\mbox{}\|#1\|_{\xi}  }
\newcommand{\XNormS}[1]{\mbox{}\|#1\|_{\xi}^2}

\newcommand{\pinv}[1]{ {#1}^\dagger}

\newtheorem{theorem}{\bf Theorem}[]
\newtheorem{lemma}[theorem]{Lemma}
\newtheorem{fact}[theorem]{Fact}
\newtheorem{remark}[theorem]{Remark}
\newtheorem{definition}[theorem]{Definition}

\newtheorem{corollary}[theorem]{Corollary}

\newcommand{\transp}{^{\textsc{T}}}

\newcommand{\mat}[1]{{\ensuremath{\bm{\mathrm{#1}}}}}

\def\rank{\hbox{\rm rank}}

\def\b{{\mathbf b}}
\def\e{{\mathbf e}}

\def\v{{\mathbf v}}

\def\matA{\mat{A}}
\def\matB{\mat{B}}
\def\matC{\mat{C}}
\def\matD{\mat{D}}
\def\matE{\mat{E}}
\def\matF{\mat{F}}
\def\matG{\mat{G}}
\def\matH{\mat{H}}
\def\matI{\mat{I}}

\def\matK{\mat{K}}
\def\matL{\mat{L}}
\def\matN{\mat{N}}
\def\matM{\mat{M}}

\def\matP{\mat{P}}
\def\matQ{\mat{Q}}
\def\matR{\mat{R}}
\def\matS{\mat{S}}
\def\matT{\mat{T}}
\def\matU{\mat{U}}
\def\matV{\mat{V}}
\def\matW{\mat{W}}
\def\matX{\mat{X}}
\def\matY{\mat{Y}}
\def\matZ{\mat{Z}}
\def\matOmega{\mat{\Omega}}
\def\matSig{\mat{\Sigma}}

\def\matOmega{\mat{\Omega}}
\def\matSig{\mat{\Sigma}}

\def\matPsi{\mat{\Psi}}
\def\matDelta{\mat{\Delta}}
\def\matXi{\mat{\Xi}}

\DeclareMathSymbol{\Prob}{\mathbin}{AMSb}{"50}
\newcommand\remove[1]{}

\def\nnz{{ \rm nnz }}

\def\math#1{$#1$}

\def\frac#1#2{{#1\over #2}}

\def\eqan#1{\begin{eqnarray*}
#1
\end{eqnarray*}}

\DeclareMathSymbol{\R}{\mathbin}{AMSb}{"52}

\def\argmin{\mathop{\hbox{argmin}}\limits}

\def\x{{\mathbf x}}
\def\y{{\mathbf y}}

\def\a{{\mathbf a}}
\def\b{{\mathbf b}}

\def\ceil#1{{\left\lceil\,#1\,\right\rceil}}

\begin{document}

\title{Optimal Principal Component Analysis in Distributed and Streaming Models}

\author{
Christos Boutsidis\\ New York, New York \\ christos.boutsidis@gmail.com
\and
David P. Woodruff\\IBM Research \\ Almaden, California \\ dpwoodru@us.ibm.com
\and
Peilin Zhong \\ Institute for Interdisciplinary Information Sciences \\ Tsinghua University, Beijing, China \\ zpl12@mails.tsinghua.edu.cn
}

\date{}
\maketitle
\begin{abstract}
\noindent
We study the Principal Component Analysis (PCA) problem
in the distributed and streaming models of computation. Given a matrix
$\matA \in \R^{m \times n},$ a rank parameter $k < \rank(\matA)$,
and an accuracy parameter $0 < \varepsilon < 1$, we want to output an $m \times k$ orthonormal matrix $\matU$ for which
$$ \FNormS{\matA -  \matU \matU\transp \matA} \le
\left(1 + \varepsilon \right)
\cdot
\FNormS{\matA - \matA_k},
$$
where $\matA_k \in \R^{m \times n}$ is the best rank-$k$ approximation to $\matA$. We show the following.
\begin{enumerate}
\item In the arbitrary partition model of Kannan et al. (COLT 2014), each
of $s$ machines holds a matrix $\matA^i$ and $\matA = \sum_{i=1}^s \matA^i$. Each machine
should output $\matU$. Kannan et al. achieve $O(skm/\eps) + \poly(sk/\eps)$ words
(of $O(\log(nm))$ bits)
communication.
We obtain the improved bound of $O(skm) + \poly(sk/\eps)$ words,
and show an optimal $\Omega(skm)$ lower bound up to low order terms.
This resolves an open question for high precision PCA.
A $\poly(\eps^{-1})$ dependence is known to be required, but we separate this
dependence from $m$.

\item We bypass the above lower bound when $\matA$ is $\phi$-{\it sparse} in each column and
each server receives a subset of columns. Here we obtain
an $O(sk \phi/\eps)+\poly(sk/\eps)$ word protocol.
Our communication is independent of the matrix dimensions, and achieves the guarantee
that each server, in addition to outputting $\matU$,
outputs a subset of $O(k/\eps)$ columns of $\matA$ containing a $\matU$ in its span
(that is, we solve distributed column subset selection).
We show a matching $\Omega(sk\phi/\eps)$ lower bound for
distributed column subset selection.
Achieving our communication bound when $\matA$ is sparse but not
sparse in each column, is impossible.

\item In the streaming model
in which the columns arrive one at a time, an algorithm of Liberty (KDD, 2013) with
an improved analysis by Ghashami and Phillips (SODA, 2014) shows
$O(km/\eps)$ ``real numbers'' of space is achievable
in a single pass, which we first
improve to an $O(km/\eps) + \poly(k/\eps)$ word space upper bound. This
almost matches a known $\Omega(km/\eps)$ bit lower bound of Woodruff (NIPS, 2014).
We show with two passes one can achieve $O(km) + \poly(k/\eps)$ words of space and
(up to the $\poly(k/\eps)$ term and
the distinction between words versus bits) this is optimal for any constant number of passes.

\item In turnstile streams, in which we receive entries of $\matA$ one at a time in an
arbitrary order, we show how to obtain a factorization of a $(1+\eps)$-approximate
rank-$k$ matrix using $O((m+n)k \eps^{-1})$ words of space. This improves the
$O((m+n \eps^{-2}) k \eps^{-2})$ bound of Clarkson and Woodruff (STOC 2009), and
matches their $\Omega((m+n) k \eps^{-1})$ word lower bound.
\end{enumerate}
Notably, our results do not depend on the condition number of $\matA$.

\end{abstract}
\thispagestyle{empty}
\newpage
\setcounter{page}{1}
%
%
%

\section{Introduction}
In distributed-memory computing systems, such as,
Hadoop~\cite{hadoop} or Spark~\cite{spark},
Principal Component Analysis (PCA) and the related Singular Value Decomposition (SVD)
of large matrices is becoming very challenging. Machine learning libraries implemented on top of such systems,
for example mahout~\cite{mahout} or mllib~\cite{mlllib},
provide distributed PCA implementations since PCA is often used as a building block for a learning algorithm.
PCA is useful for dimension reduction, noise removal, visualization, etc.
In all of these implementations, the bottleneck is the communication;
hence, the focus has been on minimizing the communication cost of the related algorithms,
and not the computational cost, which is the bottleneck in more traditional batch systems.

The data matrix corresponding to the dataset in hand, e.g., a term-document matrix representing a text collection,
or the Netflix matrix
representing user's ratings for different movies, could be distributed in many different ways~\cite{poulson2013elemental}.
In this paper, we focus on the following so-called arbitrary partition and column partition models.
In the arbitrary partition model, a matrix $\matA\in\R^{m\times n}$ is arbitrarily distributed among $s$ machines. Specifically,
$\matA=\sum_{i=1}^s \matA_i$
where $\matA_i\in\R^{m\times n}$ is held by machine $i$. Unless otherwise stated, we always assume $m\leq n$.
In the column partition model, a matrix $\matA\in \R^{m\times n}$ is distributed column-wise among $s < n$ machines:
$
\matA =
\begin{pmatrix}
\matA_1 & \matA_2 & \dots & \matA_s
 \end{pmatrix};
$
here, for $i=1:s,$ $\matA_i$ is an $m \times w_i$ column
submatrix of $\matA$ with $\sum_i w_i=n$.
Note that the column partition model is a special case of the arbitrary partition model.
Thus, it is desirable to prove upper bounds in the arbitrary partition model,
and lower bounds in the column partition model.
Both models have been adapted by traditional numerical linear algebra software libraries for distributed memory matrix computations~\cite{benson2013direct}.

%
A recent body of work~\cite{FSS13,BKLW14,Lib13, GP13,KVW14,SPS14}~(see Section~\ref{sec:related} for a comparison) has focused on designing algorithms which minimize the communication needed for each
machine to output an $m \times k$ matrix $\matU$ for which an $m \times k$ orthonormal matrix $\matU$ for which
$ \FNormS{\matA -  \matU \matU\transp \matA} \le  \left(1 + \varepsilon \right)
\cdot
\FNormS{\matA - \matA_k},$
where $\matA_k \in \R^{m \times n}$ is the best rank-$k$ approximation to $\matA$. Each machine should
output the same matrix $\matU$ which can then be used for downstream applications, such as clustering,
where one first projects the data to a lower dimensional subspace (see, e.g., \cite{BKLW14}). The protocol
should succeed with large constant probability. The model is such that each of the $s$ machines
communicates with one machine, called the ``coordinator'' (which we sometimes refer to as Server), but does not communicate with other machines.
This is known as the coordinator model and one can simulate arbitrary
point-to-point communication with a constant factor overhead in communication
together with an additive $O(\log (mn))$ bit overhead per message (see, e.g., \cite{pvz12}).

\subsection{Our Results}
The best upper bound
in the arbitrary partition model is $O((skm \eps^{-1}) + s \cdot \poly(k \eps^{-1}))$ words (of $O(\log(mn))$ bits)
communication \cite{KVW14}, while the only known lower bound for these problems is in the arbitrary partition
model and is $\Omega(skm)$ \cite{KVW14}. In high precision applications one may want to set $\eps$ to be
as small as possible and the leading order term of $O(skm \eps^{-1})$ is undesirable. A natural question is
whether there is an $O((skm) + s\cdot\poly(k \eps^{-1}))$ word protocol. We note that there is a fairly
simple known $\Omega(\eps^{-2})$ bit lower bound \cite{WZ16} (we present this in Lemma~\ref{lem:epslower} for
completeness),
so although one cannot achieve an $O(\log(1/\eps))$ dependence in
the communication as one would maybe expect given iterative algorithms
for regression with running time $O(\log(1/\eps))$ (see section 2.6 of \cite{w14} for an overview),
an $O((skm) + s\cdot\poly(k \eps^{-1}))$ would at least
separate the dependence of $\eps$ and $m$ and allowing for much smaller $\eps$ with
the same amount of communication.

There are many existing protocols using $O(skm \eps^{-1})$ words of communication~\cite{FSS13,BKLW14,Lib13, GP13,KVW14,SPS14}, and the protocols work in very different ways: that of \cite{FSS13,BKLW14} is based on coresets, while that of \cite{Lib13,GP13} is based on adapting a streaming algorithm to the communication setting, while that of \cite{KVW14} is based on sketching, and that of \cite{SPS14} is based on alternating minimization and is useful only under some assumptions on the condition number (see Section~\ref{sec:related} for a more detailed discussion of these protocols). It was
thus natural to assume there should be a lower bound of $\Omega(skm \eps^{-1})$.

%
%
Instead, we obtain a new upper bound in the arbitrary partition model
and a matching lower bound (up to lower order terms) in the column partition model.
\begin{theorem}[Restatement of Theorem~\ref{thmd2} and Theorem~\ref{thm:main}]\label{thm:arbitrary}

 Suppose matrix $\matA\in\R^{m\times n}$ is partitioned in the arbitrary-partition model (See Definition~\ref{def:model2}). For any $1\geq\eps>0$, there is an algorithm which on termination leaves the same orthonormal matrix $\matU\in\R^{m\times k}$ on each machine such that
$\FNormS{\matA-\matU\matU\transp\matA}\leq (1+\eps)\cdot \FNormS{\matA-\matA_k}$
holds with arbitrarily large constant probability. Further, the algorithm runs in polynomial time with total communication complexity $O(skm+s\cdot poly(k\eps^{-1}))$ words each containing $O(\log(smn\eps^{-1}))$ bits.

 If matrix $\matA\in\R^{m\times n}$ is distributed in the column-partition model (See Definition~\ref{def:model}), then for any positive constant $C\leq O(\poly(skm))$, the algorithm which can leave a orthonormal matrix $\matU\in\R^{m\times k}$ on each machine such that
    $\FNormS{\matA-\matU\matU\transp\matA}\leq C\cdot \FNormS{\matA-\matA_k}$
    holds with constant probability has at least $\Omega(skm \log(skm))$ bits of communication.

\end{theorem}
In some applications even an $O(skm)$ word protocol may be too costly, as $m$ could be very large. One could
hope to do for communication what recent work on input sparsity algorithms \cite{CW13,MM13,NN13,BDN15,C16B}
has done for computation, i.e.,
obtain protocols sensitive to the number of non-zero entries of $\matA$.
%
Many matrices are sparse, e.g., Netflix provides a training data set of $100,480,507$ ratings that $480,189$ users give to $17,770$ movies. If users correspond to columns in the matrix, then the average column sparsity is $\approx 200$ non-zero elements (out of the $17,770$ possible coordinates).

Our second contribution is the first protocol depending on the number of non-zero entries of $\matA$ in the column
partition model. Our communication is independent of the matrix dimensions.
Denote by $\phi$ the maximum number of non-zero elements of a column in $\matA$.
When we say that $\matA$ is $\phi$-sparse, we mean that every column of $\matA$ has at most $\phi$
non-zero elements and $\nnz(\matA) \le \phi \cdot n$. Our protocol has the additional feature of leaving on each
machine the same subset $\matC$ of $O(k/\eps)$ columns of $\matA$ for which there exists an $m \times k$
orthonormal matrix $\matU$ in the column span of $\matC$ for which
$\FNormS{\matA - \matC \pinv{\matC}\matA} \le
\FNormS{\matA - \matU \matU\transp \matA}  \le
(1 + \varepsilon) \cdot \FNormS{\matA -  \matA_k}$. This is known as the {\it column subset selection problem}, which
is useful since $\matC$ may be sparse if $\matA$ is, and also it can lead to better data intepretability. To
partially complement our protocol, we show our protocol is optimal for any protocol
solving the column subset selection problem. We summarize these results as follows.

\begin{theorem}[Restatement of Theorem~\ref{thm1}]\label{thm:sparse}

Suppose a $\phi$-sparse matrix $\matA\in\R^{m\times n}$ is partitioned in the column-partition model (See Definition~\ref{def:model}). For any $1\geq \eps >0$, there is an algorithm which on termination leaves $\matC \in \R^{m \times c}$ with $c=O(k/\eps)$ columns of $\matA$ and an orthonormal matrix $\matU$ on each machine such that
$
\FNormS{\matA - \matC \pinv{\matC}\matA} \le
\FNormS{\matA - \matU \matU\transp \matA}  \le
(1 + \varepsilon) \cdot \FNormS{\matA -  \matA_k}
$
holds with arbitrarily large constant probability. Further, the algorithm runs in polynomial time with total communication complexity $O\left(s k \phi \varepsilon^{-1} + s k^2 \varepsilon^{-4}  \right)$ words each containing of $O(\log(smn\eps^{-1}))$ bits.

If a $\phi$-sparse matrix $\matA\in\R^{m\times n}$ is distributed in the column-partition model (See Definition~\ref{def:model}), then for any positive constant $C\leq O(\poly(sk\phi))$, the algorithm which can leave a orthonormal matrix $\matU\in\R^{m\times k}$ on each machine such that
    $\FNormS{\matA-\matU\matU\transp\matA}\leq C\cdot \FNormS{\matA-\matA_k}$
    holds with constant probability has at least $\Omega(sk\phi \log(sk\phi))$ bits of communication.

\end{theorem}

\begin{remark} Our protocol has no dependence on condition number and works for arbitrary matrices. We make a serious effort as shown in Section~\ref{sec:precisionDistributed} to achieve this. Note that even if the entries of $\matA$
are specified using $O(\log(mn))$ bits, the condition number can be $2^{-\Omega(mn \log(mn))}$ \cite{av97}.
\end{remark}
%

\begin{table}
\tiny
\begin{center}
\begin{tabular}{m{1.75in}|c@{\hspace*{0.05in}}l|c@{\hspace*{0.05in}}l}
  & \multicolumn{2}{c|}{Upper bounds} & \multicolumn{2}{c}{Lower bounds}\\
&&&&\\
\hline
&&&&\\
Definition~\ref{def:dpca2} (arbitrary partition model)
& $O(s k m + s \cdot \poly(k/\varepsilon) )$  & (Theorem~\ref{thmd2})
&  $\Omega(s k m)$ & (Theorem 1.2 in~\cite{KVW14})    \\
&&&&\\
\hline
&&&&\\
Definition~\ref{def:dpca} (column partition model)
& $O(s k m + s \cdot \poly(k/\varepsilon) )$  & (Theorem~\ref{thmd2})
&  $\Omega(s k m)$ &(Theorem~\ref{thm:main})    \\
&&&&\\
\hline
&&&&\\
Definition~\ref{def:dpca} with sparsity $\phi = o(\varepsilon \cdot m) $
& $O(s k \phi  \varepsilon^{-1} +   s \cdot \poly(k/\varepsilon) )$  &(Theorem~\ref{thm1}) & $\Omega(s k \phi )$ & (Corollary~\ref{thm:main2} )\\
&&&&\\
\hline
\end{tabular}
\medskip
\caption{\noindent Communication upper/lower bounds for the Distributed PCA problems.
\label{table:results}}
\end{center}
\end{table}
\begin{remark} We did not discuss the running time of the algorithms, but we are interested in the fastest possible algorithms with the minimal communication cost, and provide fast algorithms as well.
\end{remark}

\begin{remark} The hard instance in Section~\ref{sec:hard} implies an $\Omega(skm)$ lower bound even if the input matrix is sparse overall but has $skm$ non-zero positions in arbitrarily locations. Therefore, to obtain our bounds we discuss the sparsity on each column instead of the overall sparsity of the matrix.
\end{remark}

%
%
A model closely related to the distributed model of computation is the streaming model of computation.
The model we focused on is the so-called {\it turnstile streaming model}. In this model, there is a \emph{stream} of update operations and each operation indicates that the corresponding entry of $\matA$ should be incremented by a specific number. We present novel PCA algorithms in this model.

Our first one-pass algorithm improves upon the best existing streaming PCA algorithm~\cite{Lib13,GP13} in two respects.
First, the space of~\cite{Lib13,GP13} is described in ``real numbers'' while our space bound ($O(m k / \varepsilon+\poly(k/\varepsilon))$ - see Theorem~\ref{thmonepass}) is in terms of words (we also bound the word size). This matches an $\Omega(km /\varepsilon)$ bit lower bound for one-pass algorithms in~\cite{Woo14}, up
to the distinction between words versus bits and a low order term $\poly(k/\varepsilon)$. Second, our algorithm can be applied in the turnstile streaming model which is stronger than column update streaming model in~\cite{Lib13,GP13}.

\begin{theorem}[Restatement of Theorem \ref{thmonepass}]\label{thm:1passsubspace}
Suppose $\matA\in\R^{m\times n}$ is given by a stream of update operations in the turnstile streaming model (See Definition~\ref{def:modelstreaming}). For any $1\geq\varepsilon>0$, there is an algorithm which uses a single pass over the stream and on termination outputs an orthonormal matrix $\matU\in\R^{m\times k}$  such that
$\FNormS{\matA-\matU\matU\transp\matA}\leq (1+\eps)\cdot \FNormS{\matA-\matA_k}$
holds with arbitrarily large constant probability. Further, the algorithm runs in polynomial time with space of total $O\left(m k / \varepsilon + \poly(k\varepsilon^{-1}) \right)$ words each containing $O(\log(smn\varepsilon^{-1}))$ bits.
\end{theorem}

A slight modification of the previous algorithm leads to a one-pass algorithm which can compute a factorization of a $(1+\eps)$-approximate rank-$k$ matrix.
The modified algorithm only needs $O((n+m)k/\varepsilon+poly(k/\varepsilon))$ words of space which improves the
$O((m+n \eps^{-2}) k \eps^{-2})$ upper bound in~\cite{CW09} and matches the $\Omega((n+m)k/\varepsilon)$ bit lower bound given by~\cite{CW09}, up to the low order term $\poly(k/\varepsilon)$.

\begin{theorem}[Restatement of Theorem \ref{thmonepassvar}]\label{thm:1passfrac}
Suppose $\matA\in\R^{m\times n}$ is given by a stream of update operations in the turnstile streaming model (See Definition~\ref{def:modelstreaming}). For any $1\geq \varepsilon>0$, there is an algorithm which uses a single pass over the stream and on termination outputs a matrix $\matA_k^*\in\R^{m\times n}$ with $rank(\matA_k^*)\le k$ such that
$\FNormS{\matA-\matA^*_k}\leq (1+\eps)\cdot \FNormS{\matA-\matA_k}$
holds with arbitrarily large constant probability. Further, the algorithm runs in polynomial time with space of total $O((m+n)k/\varepsilon+\poly(k\varepsilon^{-1}))$ words each containing $O(\log(smn\varepsilon^{-1}))$ bits.
\end{theorem}

We also show a two-pass algorithm which is an implementation of our distributed PCA protocol. It uses $O((km) + \poly(k/\varepsilon))$ words of space, which up to the $\poly(k/\varepsilon)$ term and the distinction between words versus bits, is optimal for any constant number of passes. A ``next natural goal'' in \cite{Woo14} was to improve the lower bound of $\Omega(km)$ to
a bound closer to the $1$-pass $\Omega(km/\varepsilon)$ lower bound established in that paper; our upper bound
shows this is not possible.

\begin{theorem}[Restatement of Theorem \ref{thmtwopass2}]
Suppose $\matA\in\R^{m\times n}$ is given by a stream of update operations in the turnstile streaming model (See Definition~\ref{def:modelstreaming}). For any $1\geq \eps >0$, there is an algorithm which uses two passes over the stream and on termination outputs an orthonormal matrix $\matU\in\R^{m\times k}$ such that
$\FNormS{\matA-\matU\matU\transp\matA}\leq (1+\eps)\cdot \FNormS{\matA-\matA_k}$
holds with arbitrarily large constant probability. Further, the algorithm runs in polynomial time with total space $O\left(m k + \poly(k\varepsilon^{-1}) \right)$ words each containing $O(\log(smn\varepsilon^{-1}))$ bits.
\end{theorem}

\begin{table}
\tiny
\begin{center}
\begin{tabular}{l|c@{\hspace*{0.05in}}l|c@{\hspace*{0.05in}}l}
  & \multicolumn{2}{c|}{Upper bounds} & \multicolumn{2}{c}{Lower bounds}\\
&&&&\\
\hline
&&&&\\
One-pass turnstile, Def.~\ref{def:spca} & $O(m k  \varepsilon^{-1}  + \poly(k,\varepsilon^{-1}) )$  &(Theorem~\ref{thmonepass}) & $\Omega(m k \varepsilon^{-1})$ bits & \cite{Woo14}\\
\hline
&&&&\\
One-pass turnstile, factorization, Def.~\ref{def:spca2} & $O((n+m) k  \varepsilon^{-1}  + \poly(k,\varepsilon^{-1}) )$  &(Theorem~\ref{thmonepassvar}) & $\Omega((n+m) k \varepsilon^{-1})$ words & \cite{CW09}\\
&&&&\\
\hline
&&&&\\
Two-pass turnstile, Def.~\ref{def:spca}
& $O(m k + \poly(k,\varepsilon^{-1}) )$  & (Theorem~\ref{thmtwopass2})
&  $\Omega(m k )$ bits &   \cite{Woo14}   \\
&&&&\\
\hline
\end{tabular}
\medskip
\caption{\noindent Space upper/lower bounds for the Streaming PCA problem.
\label{table:results3}}
\end{center}
\end{table}

\subsection{Technical Overview}\label{sec:subtmp}

\subsubsection{Upper Bounds}

\paragraph{The Arbitrary Partition Model.}
%
%
%
Recall that the input matrix $\matA\in\R^{m\times n}$ is arbitrarily partitioned into
$s$ matrices
$\matA_i \in \R^{m \times n},$ i.e., for $i=1,2,\dots,s$:
$
\matA = \sum_{i=1}^s \matA_i.
$
By sketching on the left and right by two small random sign matrices $\matS\in\R^{O(k/\eps^2)\times m}$ and $\matT\in \R^{n\times O(k/\eps^2)}$, with $O(sk^2/\eps^4)$ words of communication, the server can learn $\tilde{\matA}=\matS\matA\matT$ which is a ``sketch'' of $\matA$. It suffices to compute the best rank-$k$ approximation $\tilde{\matA}_k$ to $\tilde{\matA}$. Suppose the SVD of $\tilde{\matA}_k=\matU_{\tilde{\matA}_k}\matSig_{\tilde{\matA}_k}\matV_{\tilde{\matA}_k}\transp$. Then the server can learn $\matA\matT\matV_{\tilde{\matA}_k}$ by the second round of communication which needs $O(skm)$ words. We then prove~(see Lemma~\ref{lembatch0}):
$
 \|\matA - \matU\matU\transp\matA\|_{\mathrm{F}}^2 \leq (1+\eps)\|\matA-\matA_k\|_{\mathrm{F}}^2
$
where $\matU$ is an orthonormal basis of the column space of $\matA\matT\matV_{\tilde{\matA}_k}$.
Notice that $rank(\matU)\le k$. Thus $\matU$ is exactly what we want.
Section~\ref{sec:algorithm3} shows the details.

%

A technical obstacle in the analysis above is that it may require a large number of machine words to specify the entries of $\matV_{\tilde{\matA}_k}$, even if each entry in $\matA$ is only a single word. In fact, one can show~(we omit the details) that the singular values of $\tilde{\matA}_k$ can be exponentially large in $k/\eps$, which means one would need to round the entries of $\matV_{\tilde{\matA}_k}$ to an additive exponentially small precision, which would translate to an undesirable $sm \cdot \poly(k/\eps)$ bits of communication.

To counter this, we use the {\it smoothed analysis} of Tao and Vu~\cite{tv07} to argue that if we add small random Bernoulli noise to $\matA$, then its minimum singular value becomes inverse polynomial in $n$. Moreover, if the rank of $\matA$ is at least $2k$ and under the assumption that the entries of $\matA$ are representable with a single machine word (so we can identify them with integers in magnitude at most $\poly(nms/\eps)$ by scaling), then we can show the additional noise preserves relative error approximation. To our knowledge, this is the first application of this smoothing technique to distributed linear algebra algorithms. Section~\ref{sec:precisionDistributed} discusses the details of this approach.

On the other hand, if the rank of $\matA$ is smaller than $2k$, the smoothing technique need not preserve relative error. In this case though, we can learn a basis $\matC \in \mathbb{R}^{m \times O(k)}$ for the column span of $\matA$ by multiplying by a pseudorandom matrix based on Vandermonde matrices. Since $\matC$ has at most $2k$ columns, we can efficiently communicate it to all servers using $O(skm)$ communication. At this point we set up the
optimization problem $\min_{\textrm{rank}(\matX) \leq k}\|\matC\matX\matC\transp\matA-\matA\|_{\mathrm{F}}$. The server cannot learn $\matC\transp\matA$ and $\matA$ directly, as this would be $\Omega(nm)$ words, but we can choose additional ``sketching matrices'' $\matT_{left}$ and $\matT_{right}$ to instead solve
$ \min_{\textrm{rank}(\matX) \leq k}\|\matT_{left}(\matC\matX\matC\transp\matA-\matA)\matT_{right}\|_{\mathrm{F}},$ which is a generalization of the subspace embedding technique to ``affine embeddings''~\cite{Woo14}.
the server learns $\matT_{left} \matC \matX \matC\transp \matA \matT_{right}$ and $\matT_{left} \matA \matT_{right}$, which are small matrices, and can be sent to each machine.
Each machine locally solves for the rank-$k$ solution $\matX_*$ minimizing $\|\matT_{left}\matC \matX \matC\transp \matA \matT_{right} - \matT_{left} \matA \matT_{right}\|_{\mathrm{F}}$. Now each machine can find the same $m \times k$ orthonormal basis and output it without further communication.

\paragraph{Sparse Matrices in the Column Partition Model.}
Our algorithm for sparse matrices is technically simpler than that in the arbitrary partition model.
Section~\ref{sec:algorithm1} presents a distributed PCA algorithm that has
$O((\phi k s \varepsilon^{-1})+\poly(sk/\eps))$ words of communication cost.
Our idea in order to take advantage of the sparsity in the input matrix $\matA$ is to select and transfer to the coordinator a small set of ``good'' columns from each sub-matrix $\matA_i$.
Specifically, we design an algorithm that first computes, in a distributed way, a matrix
$\tilde\matC \in \R^{m \times c}$ with $c =O(k \varepsilon^{-1})$ columns of $\matA$,
and then finds $\matU \in span(\tilde\matC)$ using a distributed, communication-efficient algorithm
developed in~\cite{KVW14}. The matrix  $\tilde\matC$ is constructed
using optimal algorithms for column sampling~\cite{BDM11a,DRVW06,cohen2014dimensionality},
extended properly to the distributed case.

As mentioned, our algorithm actually solves the distributed column subset selection problem.
It builds on results for column subset selection in the batch setting~\cite{BDM11a,DRVW06,cohen2014dimensionality},
but gives a novel
analysis showing that in the distributed setting one can find columns as good as in the
batch setting.
To the best of our knowledge, only heuristics for distributed column subset selection
were known~\cite{FEGK13}.
%
%
%
We also optimize the time complexity of our algorithm, see Table~\ref{table:summary}.

\paragraph{Turnstile streaming model.}
Our one-pass streaming PCA algorithm starts by generating two random sign matrices $\matS\in\R^{O(k/\varepsilon)\times n}$ and $\matR\in\R^{m\times O(k/\varepsilon)}$. As shown in Lemma~\ref{lembatch2},
$$\min_{\rank(\matX) \le k} \FNormS{\matA\matR\matX\matS\matA - \matA} \leq (1+\varepsilon) \cdot \FNormS{\matA - \matA_k}$$
The intuition of the above is as the following.
Suppose the SVD of $\matA_k$ is $\matU_{\matA_k}\matSig_{\matA_k}\matV_{\matA_k}\transp$ where $\matA_k$ is the best rank-$k$ approximation matrix to $\matA$.
Since $\matX$ can be chosen as $\matSig_{\matA_k}\matV_{\matA_k}$, we have
$\min_{rank(\matX)\le k} \FNormS{\matU_{\matA_k}\matX - \matA} = \FNormS{\matA_k-\matA}$.
As shown in~\cite{CW09}, $\matS$ can provide a sketch for the regression problem. Thus,
$\FNormS{\matU_{\matA_k}\tilde{\matX} - \matA}\le (1+\varepsilon)\cdot\min_{rank(\matX)\le k} \FNormS{\matU_{\matA_k}\matX - \mat\matA}$
where
$$\tilde{\matX}=\arg \min_{rank(\matX)\le k} \FNormS{\matS\matU_{\matA_k}\matX - \matS\matA}$$
Notice that $\tilde{\matX}$ is in the row space of $span(\matS\matA)$. Then we focus on the regression problem:
$$\min_{\rank(\matX) \le k} \FNormS{\matX\matS\matA - \matA}.$$
Similarly, we can use $\matR$ to sketch the problem. Thus finally we only need to optimize
$$\min_{\rank(\matX) \le k} \FNormS{\matA\matR\matX\matS\matA - \matA}.$$
However, it is too expensive to store the entire matrix $\matA$. Thus, we sketch on the left and right by $\matT_{left}$ and $\matT_{right}$. Since all of $\matT_{left}\matA\matR,$ $\matS\matA\matT_{right}$ and $\matT_{left}\matA\matT_{right}$ can be maintained in small space, $\matX_*=\arg \min_{\rank(\matX) \le k} \FNormS{\matT_{left}(\matA\matR\matX\matS\matA - \matA)\matT_{right}}$ can be constructed after one pass over the stream.
 Additionally, if $\matA\matR$ is also maintained in the algorithm, we can get $\matU\in\R^{m\times k}$ of which columns are an orthonormal basis of $span(\matA\matR\matX_*)$ satisfying
$
\FNormS{\matA - \matU \matU\transp \matA}  \le
(1 + \varepsilon) \cdot  \FNormS{\matA -  \matA_k}.
$
Furthermore, if $\matS\matA$ is also maintained, since it suffices to compute the SVD of $\matX_*=\matU_{\matX_*}\matSig_{\matX_*}\matV_{\matX_*}\transp$, $\matT=\matT_{left}\matA\matR\matU_{\matX_*}$ and $\matK=\matV_{\matX_*}\transp\matS\matA\matT_{right}$ can be constructed in $O((n+m)k)$ words of space. Therefore, the algorithm can compute $\matA_k^*=\matT\matSig_{\matX_*}\matK=\matA\matR\matX_*\matS\matA$ in $O((n+m)k)$ words of space. The algorithm then can output $\matA_k^*$ with total space $O((n+m)k/\varepsilon+\poly(k/\varepsilon))$ words (see Theorem~\ref{thmonepassvar}).

Our two-pass streaming PCA algorithm is just an implemententation of our distributed PCA algorithm in the streaming model.

\begin{table}
\tiny
\begin{center}
\begin{tabular}{l|c@{\hspace*{0.05in}}l|c@{\hspace*{0.05in}}l}
  & \multicolumn{2}{c|}{Upper bounds} & \multicolumn{2}{c}{Lower bounds}\\
&&&&\\
\hline
&&&&\\
Definition~\ref{def:dcssp1} & $O(s k \phi  \varepsilon^{-1})$  &(Theorem~\ref{thm1s}) & $\Omega(s k \phi  \varepsilon^{-1})$ & (Theorem~\ref{lower:cssp1})\\
&&&&\\
\hline
&&&&\\
Definition~\ref{def:dcssp2}
& $O\left(s k \phi  \varepsilon^{-1} + s \cdot \poly(k/\varepsilon)  \right)$  & (Theorem~\ref{thm1s})
&  $\Omega(s k \phi  \varepsilon^{-1})$ &(Corollary~\ref{lower:cssp2})    \\
&&&&\\
\hline
\end{tabular}
\medskip
\caption{\noindent Communication upper/lower bounds for the CSSP problems of Definitions~\ref{def:dcssp1} and~\ref{def:dcssp2}.
\label{table:results2}}
\end{center}
\end{table}

\subsubsection{Lower Bounds}
Table~\ref{table:results2} summarizes the matching communication lower bounds that we obtain for distributed column-based matrix reconstruction. Theorem~\ref{lower:cssp1} proves a lower bound for the problem in Definition~\ref{def:dcssp1}; then, a lower bound for the problem of Definition~\ref{def:dcssp2} follows immediately since this is a harder problem~(see Corollalry~\ref{lower:cssp2}).

\paragraph{Distributed Column Subset Selection.}
Our most involved lower bound argument is in showing the tightness for the distributed column subset selection problem,
and is given in Theorem \ref{lower:cssp1}.
To illustrate the ideas, suppose $k = 1$.
We start with the canonical hard matrix for column subset selection, namely, an $m \times m$ matrix $\matA$ whose first row is all ones, and remaining rows are a subset
of the identity matrix. Intuitively the best rank-$1$ approximation is very aligned with the first row of $\matA$. However, given only $o(1/\eps)$ columns
of the matrix, there is not a vector in the span which puts a $(1-\eps)$-fraction of its mass on the first coordinate, so $\Omega(1/\eps)$ columns
are needed. Such a matrix has $\phi = 2$.

Obtaining a lower bound for general $\phi$ involves creating a large net of such instances.
Suppose we set $m = \phi$ and
consider $\tilde{\matL} \matA$, where $\tilde{\matL}$ is formed by choosing a random
$\phi \times \phi$ orthonormal matrix, and then rounding each entry to the nearest integer multiple of $1/\poly(n)$. We can show that $o(1/\eps)$ columns
of $\tilde{\matL} \matA$ do not span a rank-$1$ approximation to $\tilde{\matL} \matA$. We also need a lemma which shows that if we take two independently
random orthonormal matrices, and round their entries to integer multiples of $1/\poly(n)$, then these matrices are extremely unlikely to share a constant
fraction of columns. The idea then is that if one server holds $\tilde{\matL} \matA$, and every other server holds the zero matrix, then
every server needs to output the same subset of columns of $\tilde{\matL} \matA$. By construction of $\matA$,
each column of $\tilde{\matL} \matA$ is the sum of the first column of $\tilde{\matL}$ and an arbitrary column of $\tilde{\matL}$, and since we can assume
all servers know the first column of $\tilde{\matL}$ (which involves a negligible $O(s\phi)$ amount of communication), this implies that each server must
learn $\Omega(1/\eps)$ columns of $\tilde{\matL}$. The probability that random discretized orthonormal matrices share $\Omega(1/\eps)$ columns is sufficiently
small that we can choose a large enough net for these columns to identify $\tilde{\matL}$ in that net, requiring $\Omega(\phi/\eps)$ bits of communication
per server, or $\Omega(s\phi /\eps)$ words in total. The argument for general $k$ can be viewed as a ``direct sum theorem'', giving $\Omega(s k \phi/\eps)$
words of communication in total.

\paragraph{Dense Matrices in the Column Partition Model.}
Our optimal lower bound for dense matrices is technically simpler. It gives an $\Omega(skm)$ word communication lower bound in column partition model for dense matrices. Theorem~\ref{thm:main} argues that there exists  an $m \times n$ matrix $\matA$ such that for this $\matA,$ any $k \le 0.99 m,$ and any error parameter $C$ with $1 < C < \poly(s k m),$
if there exists a protocol to construct an $m \times k$ matrix $\matU$ satisfying
$
\FNormS{\matA - \matU \matU\transp \matA} \le C \cdot \FNormS{\matA - \matA_k},
$ with constant probability, then this protocol has communication cost at least $\Omega(s k m)$ words. This lower bound is proven for a matrix $\matA$ that has $k$ fully dense columns. We should note that the lower bound holds even if each machine only outputs the projection of $\matA_i$ onto $\matU$, i.e., $\matU\matU\transp\matA_i,$ rather than $\matU$ itself. Note that this problem is potentially easier, since given $\matU$, a machine could find  $\matU \matU\transp\matA_i$.

The intuition of the proof is that machine $1$ has a random matrix $\matA_1\in\R^{m\times k}$ with orthonormal columns (rounded to integer multiples of $1/\poly(n)$),
while all other machines have a very small multiple of the identity matrix. When concatenating
the columns of these matrices, the best $k$-dimensional subspace to project the columns onto
must be very close to $\matA_1$ in spectral norm. Since machine $i$, $i > 1$,
has tiny identity
matrices, after projection it obtains a projection matrix very close to the projection onto the column
span of $\matA_1$. By choosing a net of $m \times k$ orthonormal matrices, all of which pairwise
have high distance in spectral norm, each machine can reconstruct a random element in this net, which requires lots of information, and hence communication.

Our lower bound of Theorem~\ref{thm:main} implies a lower bound of $\Omega(s k \phi)$ words for matrices in which the column sparsity is $\phi$~(see Corollary~\ref{thm:main2}) as a simple corollary.


\subsection{Road Map}
In the first ten pages, we have one additional section,
Section~\ref{sec:dense_short}, which gives an overview of our algorithm
in the arbitrary partition model.

In the Appendix, we then present more detail on prior results on distributed PCA algorithms in Section~\ref{sec:related}.
Section~\ref{sec:background} introduces the notation and basic results from linear algebra that we use in our algorithms.
Section~\ref{sec:dense} presents our distributed PCA algorithm for arbitrary matrices for Definition~\ref{def:dpca2}. The communication of the algorithm in this section can be stated only in terms of ``real numbers''. We resolve this issue in Section~\ref{sec:precisionDistributed} where a modified algorithm has communication cost bounded in terms of machine words.
Section~\ref{sec:streaming} discusses space-optimal PCA methods in the streaming model of computation. Section~\ref{sec:algorithm} presents a distributed PCA algorithm for sparse matrices in column partition model, while Section~\ref{sec:faster} extends this algorithm to a faster distributed PCA algorithm for sparse matrices.
Section~\ref{sec:lower} presents communication lower bounds.


\section{Outline of our result in the arbitrary partition model}\label{sec:dense_short}
Here we give a brief outline of our result in the arbitrary patition
model, deferring the full details to the appendix.

We describe a fast distributed PCA algorithm with total communication $O(m s k)$ words plus low order terms, which is optimal in the arbitrary partition model in the sense that an $\tilde{\Omega}( m s k)$ bit lower bound was given by~\cite{KVW14}. 
The algorithm employs, in a novel way, the notion of {\it projection-cost preserving sketches} from~\cite{cohen2014dimensionality}.
In particular, whereas
all previous~\cite{Sar06,CW13arxiv} dimension-reduction-based SVD methods reduce one dimension of the input matrix to compute some approximation to the SVD, our method reduces {\it both} dimensions and computes an approximation to the SVD from a small almost square matrix (we note that other work, such as work
on estimating eigenvalues \cite{an13}, uses sketches in both dimensions, but it is not
clear how to compute {\it singular vectors} using such sketches).
Unlike~\cite{KVW14} which reduces only one dimension in the first communication round, we do the reduction on {\it both} dimensions in the same round.

We first present the batch version of the algorithm which offers a new low-rank matrix approximation technique;
a specific implementation of this algorithm offers a communication-optimal distributed PCA algorithm (we also discuss
in Section~\ref{sec:streaming} a variant of this algorithm that offers a two-pass space-optimal PCA method in the
turnstile streaming model). Before presenting all these new algorithms in detail, we present the relevant results from the
previous literature that we employ in the analysis.

\subsection{Projection-cost preserving sketching matrices}
In this section, we recap a notion of sketching matrices which we call ``projection-cost preserving sketching matrices''. A sketching matrix from this family is a linear matrix transformation and it has the property that for all projections it preserves, up to some error, the difference between the matrix in hand and its projection in Frobenius norm.

\begin{definition}[Projection-cost preserving sketching matrices] \label{def:pcpsm}
We say that $\matW \in \mathbb{R}^{n \times \xi}$ is an $(\varepsilon,k)$-projection-cost preserving sketching matrix of $\matA\in \mathbb{R}^{m\times n}$, if for all rank-$k$ orthogonal projection matrices $\matP\in \mathbb{R}^{m\times m}$, it satisfies
$$ (1-\varepsilon)\FNormS{\matA-\matP\matA} \leq \FNormS{\matA\matW-\matP\matA\matW}+c \leq (1+\varepsilon)\FNormS{\matA-\matP\matA}$$
where $c$ is a non-negative constant which only depends on $\matA$ and $\matW$. We also call $\matA\matW$ an $(\varepsilon,k)$-projection-cost preserving sketch of $\matA$.
\end{definition}

Due to the following lemma, we know that a good rank-$k$ approximation projection matrix of $(\varepsilon,k)$-projection-cost preserving sketch $\matA\matW$ also provides a good rank-$k$ approximation to $\matA$.

\begin{lemma}[PCA via Projection-Cost Preserving Sketches - Lemma 3 in~\cite{cohen2014dimensionality}] \label{fact:pcavpcps}
Suppose $\matW\in \mathbb{R}^{n \times \xi}$ is an $(\varepsilon,k)$-projection-cost preserving sketching matrix of $\matA\in \mathbb{R}^{m\times n}$. Let
$\hat{\matP}^{*} = \arg \min_{rank(\matP)\leq k} \FNormS{\matA\matW-\matP\matA\matW}$.
For all $\hat{\matP},\varepsilon'$ satisfying $rank(\hat{\matP})\leq k,\varepsilon'\geq 0$, if $\FNormS{\matA\matW-\hat{\matP}\matA\matW}\leq (1+\varepsilon')\FNormS{\matA\matW-\hat{\matP}^*\matA\matW}$,
$$\FNormS{\matA-\hat{\matP}\matA}\leq \frac{1+\varepsilon}{1-\varepsilon}\cdot(1+\varepsilon')\FNormS{\matA-\matA_k}$$
\end{lemma}

\cite{cohen2014dimensionality} also provides several ways to construct projection-cost preserving sketching matrices. Because we mainly consider the communication, we just choose one which can reduce the dimension as much as possible. Furthermore, it is also an oblivious projection-cost preserving sketching matrix.

\begin{lemma}[Dense Johnson-Lindenstrauss matrix - part of Theorem 12 in~\cite{cohen2014dimensionality}]\label{lem:djlm}
For $\varepsilon<1$, suppose each entry of $\matW\in \mathbb{R}^{n\times \xi}$ is chosen $O(\log(k))$-wise independently and uniformly in $\{1/\sqrt{\xi},-1/\sqrt{\xi}\}$ where $\xi=O(k\varepsilon^{-2})$~\cite{CW09}. For any $\matA\in \mathbb{R}^{m\times n}$, with probability at least $0.99$, $\matW$ is an $(\varepsilon,k)$-projection-cost preserving sketching matrix of $\matA$.
\end{lemma}

\subsection{A batch algorithm for the fast low rank approximation of matrices}\label{sec:algbatch}
In this section, we describe a new method for quickly computing a low-rank approximation to a given matrix.
This method does not offer any specific advantages over previous such techniques~\cite{Sar06,CW13arxiv,boutsidis2013improved}; however,
this new algorithm can be implemented efficiently in the distributed setting (see Section~\ref{sec:algorithm3}) and in
the streaming model of computation (see Section~\ref{sec:streaming}); in fact we are able to obtain communication-optimal and
space-optimal results, respectively.  For completeness as well as ease of presentation, we first present and analyze the simple batch version of the algorithm. The algorithm uses the dense Johnson-Lindenstrauss matrix of Lemma~\ref{lem:djlm} in order to reduce both dimensions of $\matA,$ before computing some sort of SVD to a
$\poly(k/\varepsilon) \times \poly(k/\varepsilon)$ matrix (see Step $2$ in the algorithm below). 

Consider the usual inputs: a matrix $\matA \in \R^{m \times n},$ a rank parameter $k < \rank(\matA),$ and an accuracy parameter $0 < \varepsilon < 1$.
The algorithm below returns an orthonormal matrix $\matU \in \R^{m \times k}$ such that
$$
\FNormS{\matA - \matU \matU\transp \matA} \le (1 + \varepsilon)
\cdot \FNormS{\matA-\matA_k}.
$$

\vspace{0.2in}
\begin{small}
{\bf Algorithm}

\begin{enumerate}

\item Construct two dense Johnson-Lindenstrauss matrices $\matS\in\mathbb{R}^{\xi_1\times m},\matT\in\mathbb{R}^{n \times \xi_2}$ with $\xi_1=O(k\varepsilon^{-2}),\xi_2=O(k\varepsilon^{-2})$~(see Lemma~\ref{lem:djlm}).

\item Construct $\tilde{\matA}=\matS\matA\matT$.

\item Compute the SVD of $\tilde{\matA}_k=\matU_{\tilde{\matA}_k}\matSig_{\tilde{\matA}_k}\matV_{\tilde{\matA}_k}\transp$
~($\matU_{ \tilde{\matA}_k } \in \R^{\xi_1 \times k}$, $\matSig_{ \tilde{\matA}_k } \in \R^{k \times k}$, $\matV_{ \tilde{\matA}_k } \in \R^{\xi_2 \times k})$.

\item Construct $\matX=\matA\matT\matV_{\tilde{\matA}_k}$

\item Compute an orthonormal basis $\matU \in \mathbb{R}^{m\times k}$ for $span(\matX)$ (notice that $rank(\matX) \le k$).

\end{enumerate}

\end{small}

Theorem~\ref{thmbatch} later in this section analyzes the approximation error and the running time of the previous algorithm.
First, we prove the accuracy of the algorithm.

\begin{lemma}\label{lembatch0}
The matrix $\matU \in \R^{m \times k}$ with $k$ orthonormal columns satisfies with probability at least $0.98$:
$$
 \FNormS{\matA -  \matU \matU\transp \matA} \le
\left(1 +\varepsilon \right) \cdot  \FNormS{\matA - \matA_k}.
$$
\end{lemma}
\begin{proof}
$$\FNormS{\tilde{\matA}-\tilde{\matA}_k}=\FNormS{\tilde{\matA}-\tilde{\matA}\matV_{\tilde{\matA}_k}\matV_{\tilde{\matA}_k}\transp}
=\FNormS{\matS\matA\matT-\matS\matA\matT\matV_{\tilde{\matA}_k}\matV_{\tilde{\matA}_k}\transp}=\FNormS{\matT\transp\matA\transp\matS\transp-\matV_{\tilde{\matA}_k}\matV_{\tilde{\matA}_k}\transp\matT\transp\matA\transp\matS\transp}$$

The first equality follows by the SVD of $\tilde{\matA}_k=\matU_{\tilde{\matA}_k}\matSig_{\tilde{\matA}_k}\matV_{\tilde{\matA}_k}\transp$. The second equality is by the construction of $\tilde{\matA}=\matS\matA\matT$. The third equality is due to $\forall \matM,\FNormS{\matM}=\FNormS{\matM\transp}$.

Due to Lemma \ref{lem:djlm}, with probability at least $0.99$, $\matS\transp$ is an $(\varepsilon,k)$-projection-cost preserving sketch matrix of $\matT\transp\matA\transp$. According to Lemma \ref{fact:pcavpcps},
\begin{equation} \label{eqn:1dimreduce}
\FNormS{\matA\matT-\matA\matT\matV_{\tilde{\matA}_k}\matV_{\tilde{\matA}_k}\transp}
=\FNormS{\matT\transp\matA\transp-\matV_{\tilde{\matA}_k}\matV_{\tilde{\matA}_k}\transp\matT\transp\matA\transp}
\leq \frac{1+\varepsilon}{1-\varepsilon}\cdot\FNormS{\matA\matT-(\matA\matT)_k}
\end{equation}

Observe that
\eqan{
\FNormS{\matU\matU\transp\matA\matT-\matA\matT}
= \FNormS{\matX\pinv{\matX}\matA\matT-\matA\matT} \leq \FNormS{\matX\matV_{\tilde{\matA}_k}\transp-\matA\matT} = \FNormS{\matA\matT\matV_{\tilde{\matA}_k}\matV_{\tilde{\matA}_k}\transp-\matA\matT}}
which is at most
$\frac{1+\varepsilon}{1-\varepsilon}\cdot \FNormS{\matA\matT-(\matA\matT)_k}.$

The first equality uses the fact that $\matU$ is an orthonormal basis of $span(\matX)$. The first inequality is followed by $\forall \matX,\matM,\matN,\FNormS{\matX\pinv{\matX}\matM-\matM}\leq\FNormS{\matX\matN-\matM}$. The second equality uses the construction that $\matX=\matA\matT\matV_{\tilde{\matA}_k}$. The second inequality follows by Eqn~(\ref{eqn:1dimreduce}).

Due to Lemma \ref{lem:djlm}, with probability at least $0.99$, $\matT$ is an $(\varepsilon,k)$-projection-cost preserving sketch matrix of $\matA$. Due to Lemma \ref{fact:pcavpcps},
$$\FNormS{\matU\matU\transp\matA-\matA} \leq \frac{(1+\varepsilon)^2}{(1-\varepsilon)^2}\cdot \FNormS{\matA-\matA_k}$$
Due to union bound, the probability that $\matS\transp$ is an $(\varepsilon,k)$-projection-cost preserving sketch matrix of $\matT\transp\matA\transp$ and $\matT$ is an $(\varepsilon,k)$-projection-cost preserving sketch matrix of $\matA$ is at least $0.98$. Note that $\frac{(1+\varepsilon)^2}{(1-\varepsilon)^2}$ is $1+O(\varepsilon)$ when $\varepsilon$ is small enough, so we can adjust $\varepsilon$ here by a constant factor to show the statement.
\end{proof}

Next, we present the main theorem.

\begin{theorem}\label{thmbatch}
The matrix $\matU \in \R^{m \times k}$ with $k$ orthonormal columns satisfies with probability at least $0.98$:
$$
 \FNormS{\matA -  \matU \matU\transp \matA} \le
\left(1 +\varepsilon \right) \cdot  \FNormS{\matA - \matA_k}.
$$
The running time of the algorithm is
$$
O\left(n m k \varepsilon^{-2} + m k^2 \varepsilon^{-4} + \poly(k \varepsilon^{-1})\right).
$$
\end{theorem}
\begin{proof}
The correctness is shown by Lemma~\ref{lembatch0}. The running time is
analyzed in Section \ref{sec:dense}.

\end{proof}

\subsection{The distributed PCA algorithm}\label{sec:algorithm3}
Recall that the input matrix $\matA \in \R^{m \times n}$ is partitioned arbitrarily as:
$\matA = \sum_i^s \matA_i$
for $i=1:s,$ $\matA_i \in \R^{m \times n}$.
The idea in the algorithm below is to implement the algorithm in Section~\ref{sec:algbatch} in the distributed setting.

\vspace{0.2in}
\begin{small}
{\bf Input:}
\begin{enumerate}
\item $\matA \in \R^{m \times n}$ arbitrarily partitioned
$
\matA = \sum_i^s \matA_i
$
for $i=1:s,$ $\matA_i \in \R^{m \times n}$.
\item rank parameter $k < \rank(\matA)$
\item accuracy parameter $\varepsilon > 0$
\end{enumerate}

{\bf Algorithm}

\begin{enumerate}

\item Machines agree upon two dense Johnson-Lindenstrauss matrices $\matS\in\mathbb{R}^{\xi_1\times m},\matT\in\mathbb{R}^{n \times \xi_2}$ with $\xi_1=O(k\varepsilon^{-2}),\xi_2=O(k\varepsilon^{-2})$~(see Lemma~\ref{lem:djlm}).

\item Each machine locally computes $\tilde{\matA_i}=\matS\matA_i\matT$ and sends $\tilde{\matA_i}$ to the server. Server constructs $\tilde{\matA}=\sum_i \tilde{\matA_i}$.

\item Server computes the SVD of $\tilde{\matA}_k=\matU_{\tilde{\matA}_k}\matSig_{\tilde{\matA}_k}\matV_{\tilde{\matA}_k}\transp$
~($\matU_{ \tilde{\matA}_k } \in \R^{\xi_1 \times k}$, $\matSig_{ \tilde{\matA}_k } \in \R^{k \times k}$, $\matV_{ \tilde{\matA}_k } \in \R^{\xi_2 \times k})$.

\item Server sends $\matV_{\tilde{\matA}_k}$ to all machines.

\item Each machine construct $\matX_i=\matA_i\matT\matV_{\tilde{\matA}_k}$ and sends $\matX_i$ to the server. Server constructs $\matX=\sum_i \matX_i$.

\item Server computes an orthonormal basis $\matU \in \mathbb{R}^{m\times k}$ for $span(\matX)$ (notice that $rank(\matX) \le k$).

\item Server sends $\matU$ to each machine.

\end{enumerate}

\end{small}

Notice that in the first step, $\matS$ and $\matT$ can be described using a random seed that is $O(\log(k))$-wise independent due to Lemma \ref{lem:djlm}.

The major remaining challenge here is to obtain small bit complexity, since this protocol could have very large bit complexity. See Section~\ref{sec:dense} in the Appendix for the full details of our algorithm.

\newpage
\bibliographystyle{abbrv}
\bibliography{RPI_bib}
\newpage
\section{Related Work}\label{sec:related}

Distributed PCA (or distributed SVD) algorithms have been investigated for a long time. One line of work, developed primarily within the numerical linear algebra literature, studies such algorithms from the perspective of parallelizing  existing standard SVD algorithms without sacrificing accuracy. This approach aims at high accuracy implementations with the least possible communication cost. The distributed models of computation go typically beyond the column-partition model and arbitrary-partition model that we study in this paper~\cite{poulson2013elemental}. An extensive survey of this line of work is out of the scope of our paper;
we refer the reader to~\cite{jessup1994parallel,tisseur1999parallel,gu1995divide} and references therein for more details as well as to popular software for distributed SVD such ScaLAPACK~\cite{scaLAPACK} and Elemental~\cite{poulson2013elemental}.

Another line of work for distributed PCA algorithms has emerged within the machine learning and
datamining communities. Such algorithms have been motivated by the need to apply SVD or PCA to extremely large matrices encoding enormous amounts of data. The algorithms in this line of research are typically heuristic approaches that work well in practice but come with no rigorous theoretical analysis. We refer the reader to~\cite{qu2002principal,macua2010consensus,bai2005principal,sensors2008} for more details.

Finally, distributed PCA algorithms in the column-partition model have been recently studied within the theoretical computer science community. Perhaps the more intuitive algorithm for PCA in this model appeared in~\cite{FSS13,BKLW14}: first, a number of left singular vectors and singular values are computed in each machine; then, the server collects those singular vectors and concatenates them column-wise in a new matrix and then it computes the top $k$ left singular vectors of this ``aggregate'' matrix. It is shown in~\cite{FSS13,BKLW14} that if the number of singular vectors and singular values in the first step is $O(k \varepsilon^{-1})$, then, the approximation error in Frobenius norm is at most $(1+\varepsilon)$ times the optimal Frobenius norm error; the communication cost is $O(s k m \varepsilon^{-1})$ \emph{real numbers}
because each of the $s$ machines sends $O(k \varepsilon^{-1})$ singular vectors of dimension $m$;
unfortunately, it is unclear how one can obtain a communication cost in terms of words/bits.
A different algorithm with the same communication cost (only in terms of real numbers since a word/bit communication bound remained unexplored) is implicit in~\cite{Lib13, GP13}~(see Theorem 3.1 in~\cite{GP13} and the discussion in Section 2.2 in~\cite{Lib13}).
Kannan, Vempala and Woodruff proposed arbitrary-partition model in~\cite{KVW14}. They developed a $(1+\varepsilon)$ Frobenius norm error algorithm with communication cost
$O\left(s k m \varepsilon^{-1} + s k ^2 \varepsilon^{-4} \right)$ words. Bhojanapalli, Jain, and Sanghavi in~\cite{SPS14}   developed an algorithm that provides a bound with respect to the spectral norm. Their algorithm is based on sampling elements from the input matrix, but the communication cost is prohibitive if $n$ is large. The cost also depends on the condition number of $\matA$.
Moreover, to implement the algorithm, one needs to know  $\FNorm{\matA-\matA_k}.$ Finally,~\cite{KM04} discussed a distributed implementation of the  ``orthogonal iteration'' to compute eigenvectors of graphs. The model they consider is different, but perhaps their protocol could be extended to our model.
\begin{table*}
  \begin{tiny}
\begin{center}
  \begin{tabular}{| l | l | l | l | l | l | | |}
    \hline
     Reference                                &  $\TNorm{\matA-\matU \matU\transp \matA}\le$   & $\FNormS{\matA-\matU \matU\transp \matA}\le$ & $\delta$ & \text{Communication cost}  & \text{Total number of arithmetic operations} \\ \hline
     Implicit in~\cite{FSS13}           &  -  &  $(1+\varepsilon)\FNormS{\matA-\matA_k}$ & $0$ & - & $O\left( m n \min\{m,n\}  + m s k \varepsilon^{-1} \min\{m,s k \varepsilon^{-1} \} \right)$ \\ \hline
     Theorem 2 in~\cite{BKLW14}      &   -  & $(1+\varepsilon)\FNormS{\matA-\matA_k}$& $0$ & - & $O\left( m n \min\{m,n\}  + m s k \varepsilon^{-1} \min\{m,s k \varepsilon^{-1} \} \right)$ \\ \hline
     Theorem 6 in~\cite{BKLW14}      &   -  & $(1+\varepsilon)\FNormS{\matA-\matA_k}$& $>0$& -  & $O\left( \nnz(\matA) + s \left( \frac{m^3k}{\varepsilon^4} + \frac{k^2 m^2}{\varepsilon^6} \right) \log\left(\frac{m}{\varepsilon}\right) \log\left(\frac{s k}{d \varepsilon}\right) \right)$ \\ \hline
     Implicit in~\cite{Lib13, GP13}   &  -  &  $(1+\varepsilon)\FNormS{\matA-\matA_k}$ & $0$ & - & $O(mnk\varepsilon^{-2})$ \\ \hline
     Thm 1.1 in~\cite{KVW14}*     &  -  & $(1+\varepsilon)\FNormS{\matA-\matA_k}$ & $O(1)$  & $O\left(s k m \varepsilon^{-1} + s k ^2 \varepsilon^{-4} \right)$ & $O(poly(m,n,k,s,\varepsilon^{-1}))$ \\ \hline
     Thm 5.1 in~\cite{SPS14}      & $\TNorm{\matA-\matA_k} +\Gamma$  & -  & $>0$ & $O\left( sm + n k^5 \varepsilon^{-2} \Delta \right)$ & $O\left(\nnz(\matA) + \delta^{-1} nk^5\varepsilon^{-2} \sigma_1^2(\matA) \sigma_{k}^{-2}(\matA) \right)$ \\ \hline
     Remark p. 11~\cite{SPS14}          & $\TNorm{\matA-\matA_k} + \Gamma$  & -  & $>0$ & $O\left( sm + n k^3 \varepsilon^{-2} \Delta \right)$ & $O\left(\nnz(\matA) + \delta^{-1}nk^5\varepsilon^{-2} \sigma_1^2(\matA) \sigma_{k}^{-2}(\matA) \right)$ \\ \hline
     Theorem~\ref{thmd2}*         &  -  & $(1+\varepsilon)\FNormS{\matA-\matA_k}$ &$O(1)$ & $O(s k m +   sk^3 \varepsilon^{-5 })$ &
$O\left( m n k \varepsilon^{-2}  + mk^2\varepsilon^{-4} +  poly(k \varepsilon^{-1}) \right).$\\ \hline
     Theorem~\ref{thm1}         &  -  & $(1+\varepsilon)\FNormS{\matA-\matA_k}$ & $O(1)$& $O(s k \phi  \varepsilon^{-1} +   sk^2 \varepsilon^{-4} )$ & $O\left( 
     m n s\cdot poly(k, 1/\varepsilon) \right)$ \\ \hline
     Theorem~\ref{thm1s}         &  -  & $(1+\varepsilon)\FNormS{\matA-\matA_k}$ &$O(1)$ & $O(s k \phi \varepsilon^{-1} +   sk^3 \varepsilon^{-5 })$ & $O\left(\nnz(\matA) \cdot  \log^2(\frac{n s}{\delta})  +
(m+n) \cdot s \cdot poly(\frac{k}{\varepsilon}\log(\frac{ns}{\delta}) \right)$. \\ \hline
  \end{tabular}
\end{center}
\vspace{-0.25in}
\caption{\scriptsize{Distributed PCA Algorithms in the column-partition model with $s$ machines~(see Definition~\ref{def:model}). $*$ indicates that the algorithm can also be applied in the arbitrary partition model~(see Definition~\ref{def:model2}).
$\matA \in \R^{m \times n}$ has rank $\rho,$
$\matU \in \R^{m \times k}$ with $k < \rho$ is orthonormal,  $0 < \varepsilon < 1$, and 
each column of $\matA$ contains at most $\phi \le m$ non-zero elements. $\delta$ is the failure probability.
Finally, for notational convenience let $\Gamma : = \varepsilon \FNorm{\matA-\matA_k}$,
$\Delta := \sigma_1^2(\matA) \sigma_{k}^{-2}(\matA) \log^2\left( \TNorm{\matA} \FNorm{\matA-\matA_k}^{-1} \varepsilon^{-1} \right)$.
}}\label{table:summary}
 \end{tiny}
\end{table*}


\section{Preliminaries}\label{sec:background}

\paragraph{Definitions.} We formally define the problems and the models here.

\begin{definition}[Arbitrary-partition model]\label{def:model2}
An $m \times n$ matrix $\matA$ is arbitrarily partitioned into $s$ matrices $\matA_i \in \R^{m\times n},$ i.e., for $i=1,2,\dots,s$:
$
\matA = \sum_{i=1}^s \matA_i
.$
There are $s$ machines, and the $i$-th machine has $\matA_i$ as input.
There is also another machine, to which we refer to as the ``server'', which acts as the central coordinator.
The model only allows communication between the machines and the server. The communication cost of an algorithm in this model is
defined as the total number of words transferred between the machines and the server, where we assume each word is $O(\log(nms/\eps))$ bits.
\end{definition}
\begin{definition}[Column-partition model]\label{def:model}
An $m \times n$ matrix $\matA$ is partitioned arbitrarily column-wise into
$s$ blocks
$\matA_i \in \R^{m \times w_i},$ i.e., for $i=1,2,\dots,s$:
$
\matA =
\begin{pmatrix}
\matA_1 & \matA_2 & \dots & \matA_s
 \end{pmatrix}.
$
Here, $\sum w_i = n$.
There are $s$ machines, and the $i$-th machine has $\matA_i$ as input.
There is also another machine, to which we refer to as the ``server'', which acts
as the central coordinator. The model only allows communication between the machines and the server. The communication cost of an algorithm in this model is defined as the total number of words transferred between the machines and the server, where we assume each word is $O(\log(nms/\eps))$ bits. 
\end{definition}

\begin{definition}[The Distributed Principal Component Analysis Problem in arbitrary partition model]\label{def:dpca2}
Give an $m \times n$ matrix $\matA$ arbitrarily partitioned into $s$ matrices $\matA_i \in \R^{m \times n}$ ($i:1,2,\dots,s$):
$
\matA=\sum_{i=1}^s \matA_i,
$ a rank parameter $k<\rank(\matA),$ and an accuracy parameter $0 < \varepsilon < 1,$ design an algorithm in the model of Definition~\ref{def:model2} which, upon termination, leaves on each machine a matrix $\matU \in \R^{m \times k}$ with $k$ orthonormal columns such that
$
\FNormS{\matA - \matU \matU\transp \matA}  \le
(1 + \varepsilon) \cdot  \FNormS{\matA - \matA_k},
$
and the communication cost of the algorithm is as small as possible.
\end{definition}

\begin{definition}[The Distributed Principal Component Analysis Problem in column partition model]\label{def:dpca}
Given an $m \times n$ matrix $\matA$ partitioned column-wise into $s$ arbitrary blocks
$\matA_i \in \R^{m \times w_i}$ ($i:1,2,\dots,s$):
$
\matA =
\begin{pmatrix}
\matA_1 & \matA_2 & \dots & \matA_s
 \end{pmatrix},
$ a rank parameter $k < \rank(\matA),$ and an accuracy parameter $0 < \varepsilon < 1,$
design an algorithm in the model of Definition~\ref{def:model} which, upon termination, leaves on each
machine a matrix $\matU \in \R^{m \times k}$ with $k$ orthonormal columns such that
$
\FNormS{\matA - \matU \matU\transp \matA}  \le
(1 + \varepsilon) \cdot  \FNormS{\matA - \matA_k},
$
and the communication cost of the algorithm is as small as possible.
\end{definition}

\begin{definition}[The Distributed Column Subset Selection Problem]\label{def:dcssp1}
Given an $m \times n$ matrix $\matA$ partitioned column-wise into $s$ arbitrary blocks
$\matA_i \in \R^{m \times w_i}$ ($i:1,2,\dots,s$):
$
\matA =
\begin{pmatrix}
\matA_1 & \matA_2 & \dots & \matA_s
 \end{pmatrix},
$ a rank parameter $k < \rank(\matA),$ and an accuracy parameter $0 < \varepsilon < 1,$
design an algorithm in the model of Definition~\ref{def:model} that, upon termination, leaves on each
machine a matrix
$\matC \in \R^{m \times c}$ with $c < n$ columns of $\matA$ such that
$
\FNormS{\matA - \matC \pinv{\matC}\matA}  \le
(1 + \varepsilon) \cdot \FNormS{\matA - \matA_k},
$
and
\begin{enumerate}
\item The number of selected columns $c$ is as small as possible.
\item The communication cost of the algorithm is as small as possible.
\end{enumerate}
\end{definition}
\begin{definition}[The Distributed Column Subset Selection Problem - rank $k$ subspace version]\label{def:dcssp2}
Given an $m \times n$ matrix $\matA$ partitioned column-wise into $s$ arbitrary blocks
$\matA_i \in \R^{m \times w_i}$ ($i:1,2,\dots,s$):
$
\matA =
\begin{pmatrix}
\matA_1 & \matA_2 & \dots & \matA_s
 \end{pmatrix},
$ a rank parameter $k < \rank(\matA),$ and an accuracy parameter $0 < \varepsilon < 1,$
design an algorithm in the model of Definition~\ref{def:model} that, upon termination, leaves on each machine
a matrix $\matC \in \R^{m \times c}$ with $c < n$ columns of $\matA$ and
a matrix $\matU \in \R^{m \times k}$ with $k$ orthonormal columns with $\matU \in span(\matC),$
such that
$
\FNormS{\matA - \matC \pinv{\matC}\matA} \le
\FNormS{\matA - \matU \matU\transp \matA}  \le
(1 + \varepsilon) \cdot \FNormS{\matA -  \matA_k},
$
and
\begin{enumerate}
\item The number of selected columns $c$ is as small as possible.
\item The communication cost of the algorithm is as small as possible.
\end{enumerate}
\end{definition}

\begin{definition}[Streaming model for Principal Component Analysis]\label{def:modelstreaming}
Let all the entries in $\matA\in\R^{m\times n}$ initially be zeros.
In the streaming model of computation, there is a stream of update operations that the $q^{th}$ operation has form $(i_q,j_q,x_q)$ which indicates that $\matA_{i_q,j_q}$ should be incremented by $x_q$ where $i_q\in\{1,...,m\},j_q\in \{1,...,n\},x_q\in\mathbb{R}$. An algorithm is allowed a single pass over the stream.
At the end of the stream the algorithm stores some information regarding $\matA$
which we call a ``sketch'' of $\matA$. The space complexity of an algorithm in this model is defined as the total number of words required to describe the information the algorithm stores during the stream including the sketch.
Each word is $O(\log(nms/\eps))$ bits.
\end{definition}
\begin{definition}[The Streaming Principal Component Analysis Problem]\label{def:spca}
Given an $m \times n$ matrix $\matA,$ a rank parameter $k < \rank(\matA),$ and an accuracy parameter $0 < \varepsilon < 1,$
design an algorithm that, using as little space as possible,
first finds a sketch of $\matA$ in the streaming model (see~Definition~\ref{def:modelstreaming}) and then,
using only this sketch outputs $\matU \in \R^{m \times k}$ with $k$ orthonormal columns such that
$
\FNormS{\matA - \matU \matU\transp \matA}  \le
(1 + \varepsilon) \cdot  \FNormS{\matA -  \matA_k}.
$
\end{definition}
\begin{definition}[The Streaming Principal Component Analysis Problem (factorization)]\label{def:spca2}
Given an $m \times n$ matrix $\matA,$ a rank parameter $k < \rank(\matA),$ and an accuracy parameter $0 < \varepsilon < 1,$
design an algorithm that, using as little space as possible,
first finds a sketch of $\matA$ in the streaming model (see~Definition~\ref{def:modelstreaming}) and then,
using only this sketch outputs $\matA_k^* \in \R^{m \times n}$ with $rank(\matA_k^*)\le k$ such that
$
\FNormS{\matA - \matA_k^*}  \le
(1 + \varepsilon) \cdot  \FNormS{\matA -  \matA_k}.
$
\end{definition}

\paragraph{Notation.}
\math{\matA,\matB,\ldots} are matrices; \math{\a,\b,\ldots} are
column vectors. $\matI_{n}$ is the $n \times n$
identity matrix;  $\bm{0}_{m \times n}$ is the $m \times n$ matrix of zeros; $\bm{1}_n$ is the $n \times 1$ vector of ones; $\bm{e}_i$ is the standard basis (whose dimensionality will be clear from the context): the $i$th element of $\bm{e}_i$ is one and the rest are zeros. $\matA^{(i)}$ and $\matA_{(j)}$ denotes the $i$th column and $j$th row of $\matA$, respectively. $\matA_{ij}$ is the $(i,j)$th entry in $\matA$.

\vspace{-0.1in}

\paragraph{Sampling Matrices.}
Let \math{\matA=[\matA^{(1)},\ldots,\matA^{(n)}] \in \R^{m \times n}}
and let $\matC=[\matA^{(i_1)},\ldots,\matA^{(i_c)}] \in \R^{m \times c}$ consist of \math{c<n}
columns of~\math{\matA}. Note that we can write \math{\matC=\matA\matOmega}, where the \emph{sampling matrix} is
\math{\matOmega=[\e_{i_1},\ldots,\e_{i_c}] \in \R^{n \times c}} (here \math{\e_i} are the standard basis vectors in \math{\R^n}). If $\matD \in \R^{c \times c}$ is a diagonal matrix,
then $\matA \matOmega \matD$ contains $c$ columns of $\matA$ rescaled with the corresponding elements in $\matD$.
We abbreviate $\matS :=  \matOmega \matD,$ hence the matrix $\matS \in \R^{n \times c}$ ``samples'' and ``rescales'' $c$ columns from $\matA$.

\vspace{-0.1in}

\paragraph{Matrix norms.}
We use the Frobenius and the spectral matrix-norms:
$ \FNormS{\matA} = \sum_{i,j} \matA_{ij}^2$;
$\TNorm{\matA} = \max_{\TNorm{\x}=1}\TNorm{\matA\x}$.
$\XNorm{\matA}$ is used if a result holds for
both norms $\xi = 2$ and $\xi = \mathrm{F}$. The standard submultiplicativity
property for matrix norms implies that for any $\matA$ and $\matB$:
$\XNorm{\matA \matB} \le \XNorm{\matA} \cdot \XNorm{\matB}$.
The triangle inequality for matrix norms implies that $\XNorm{\matA + \matB} \le \XNorm{\matA} + \XNorm{\matB}$.
A version of the triangle inequality for the norms squared is:
$\XNormS{\matA + \matB} \le 2 \cdot \XNormS{\matA} +2 \cdot \XNormS{\matB}$.
A version of the matrix pythagorean theorem is: if $\matA\transp\matB$ is the all-zeros matrix,
then $\FNorm{\matA+\matB}^2 = \FNorm{\matA}^2+\FNorm{\matB}^2$. If $\matV$ has orthonormal columns, then $\XNorm{\matA \matV\transp} = \XNorm{\matA},$ for any $\matA$.
If $\matP$ is a symmetric projection matrix (i.e., $\matP = \matP\transp$ and $\matP^2 = \matP$)
then, $\XNorm{\matP \matA} \le \XNorm{\matA},$ for any $\matA$.

\paragraph{Singular Value Decomposition (SVD) and Moore-Penrose Pseudo-inverse.}
The SVD of
$\matA \in \R^{m \times n}$ with $\rank(\matA) = \rho$ is
$$
\label{svdA} \matA
         = \underbrace{\left(\begin{array}{cc}
             \matU_{k} & \matU_{\rho-k}
          \end{array}
    \right)}_{\matU_{\matA} \in \R^{m \times \rho}}
    \underbrace{\left(\begin{array}{cc}
             \matSig_{k} & \bf{0}\\
             \bf{0} & \matSig_{\rho - k}
          \end{array}
    \right)}_{\matSig_\matA \in \R^{\rho \times \rho}}
    \underbrace{\left(\begin{array}{c}
             \matV_{k}\transp\\
             \matV_{\rho-k}\transp
          \end{array}
    \right)}_{\matV_\matA\transp \in \R^{\rho \times n}},
$$
with singular values \math{\sigma_1\ge\ldots\sigma_k\geq\sigma_{k+1}\ge\ldots\ge\sigma_\rho > 0}. Here, $k < \rho$.
We will use $\sigma_i\left(\matA\right)$ to denote the $i$-th
singular value of $\matA$ when the matrix is not clear from the context.
The matrices
$\matU_k \in \R^{m \times k}$ and $\matU_{\rho-k} \in \R^{m \times (\rho-k)}$ contain the left singular vectors of~$\matA$, and, similarly, the matrices $\matV_k \in \R^{n \times k}$ and $\matV_{\rho-k} \in \R^{n \times (\rho-k)}$ contain the right singular vectors of~$\matA$. It is well-known that $\matA_k=\matU_k \matSig_k \matV_k\transp = \matA \matV_k \matV_k\transp = \matU_k\matU_k\transp\matA$ minimizes \math{\XNorm{\matA - \matX}} over all matrices \math{\matX \in \R^{m \times n}} of rank at most $k$. Specifically, $\TNormS{\matA - \matA_k} = \sigma_{k+1}^{2}(\matA)$ and
$\FNormS{\matA - \matA_k} = \sum_{i=k+1}^{\rho} \sigma_i^2(\matA)$ (see~\cite{golub2012matrix}).
$\pinv{\matA} = \matV_\matA \matSig_\matA^{-1} \matU_\matA\transp \in \R^{n \times m}$
denotes the so-called Moore-Penrose pseudo-inverse of $\matA \in \R^{m \times n}$ (here $\matSig_\matA^{-1}$ is the inverse of $\matSig_\matA$).
By the SVD of $\matA$ and $\pinv{\matA}$, for all $i=1,\dots,\rho = \rank(\matA) = \rank(\pinv{\matA})$:
$\sigma_i(\pinv{\matA}) = 1/\sigma_{\rho - i + 1}(\matA).$

\subsection{The best rank \math{k} matrix $\Pi_{\matV,k}^\xi(\matA)$ within a subspace $\matV$} \label{sec:bestrankk}
Our analysis uses theory involving computing the best rank $k$ approximation of a matrix $\matA$ within a given column subspace $\matV$.
Let $\matA \in \mathbb{R}^{m \times n}$, let $k < n$ be an integer, and let $\matV \in \mathbb{R}^{m \times c}$ with $k < c< n $.
$\Pi_{\matV,k}^\mathrm{F}(\matA) \in \mathbb{R}^{m \times n}$ is the best rank \math{k} approximation to \math{\matA} in the column span of \math{\matV}.
Equivalently, we can write
$$\Pi_{\matV, k}^\mathrm{F}(\matA) = \matV \matX_{opt},$$
where
$$
\matX_{opt} = \argmin_{\matX \in {\R}^{c \times n}:\rank(\matX)\leq k}\FNormS{\matA-
\matV \matX}.
$$
In order to compute $\Pi_{\matV,k}^{\mathrm{F}}(\matA)$ given $\matA$,
$\matV$, and $k$, one can use the following algorithm:
\begin{center}
\begin{algorithmic}[1]
\STATE $\matV = \matY \matPsi$ is a $qr$ decomposition of $\matV$ with $\matY \in \R^{m \times c}$ and $\matPsi \in \R^{c \times c}$.
This step requires $O(m c^2)$ arithmetic operations.
\STATE
$\matXi = \matY\transp \matA \in \R^{c \times n}$. This step requires \math{O(mnc)} arithmetic operations.
\STATE
 $\matXi_k = \matDelta \tilde{\matSig}_k \tilde{\matV}_k\transp \in \R^{c \times n}$ is a rank $k$ SVD of $\matXi$
 with $\matDelta \in \R^{c \times k}, \tilde{\matSig}_k \in \R^{k \times k}, $ and $\tilde{\matV}_k \in \R^{n \times k}.$
This step requires \math{O( nc^2)} arithmetic operations.
\STATE Return $ \matY \matDelta \matDelta\transp \matY\transp\matA  \in \mathbb{R}^{m \times n}$ of rank at most $k$.
\end{algorithmic}
\end{center}
\medskip
Notice that $  \matY \matDelta \matDelta\transp \matY\transp\matA  \in \mathbb{R}^{m \times n}$ is a rank $k$ matrix that lies in the column span of $\matV$.
The next lemma is a simple corollary of Lemma 4.3 in~\cite{CW09}.
\begin{lemma}\label{lem:bestF}
Given $\matA \in {\R}^{m \times n}$, $\matV\in\R^{m\times c}$ and an integer $k$, $\matY \matDelta \matDelta\transp \matY\transp$ and  $ \matQ \tilde{\matU}_k \tilde{\matSig}_k \tilde{\matV}_k\transp$ satisfy:
$$
\FNormS{\matA -  \matY \matDelta \matDelta\transp \matY\transp\matA } \le \FNormS{ \matA-  \matY \matDelta \tilde{\matSig}_k \tilde{\matV}_k\transp  } = \FNormS{\matA-\Pi_{\matV,k}^{\mathrm{F}}(\matA)}
$$
The above algorithm requires $O(mnc + nc^2)$ arithmetic operations to construct $\matY, \matPsi,$ and $\matDelta$.
We will denote the above procedure as
$ [ \matY, \matPsi, \matDelta ] = BestColumnSubspaceSVD(\matA, \matV, k). $
\end{lemma}
\begin{proof}
The equality was proven in Lemma 4.3 in~\cite{CW09}. To prove the inequality, notice that for any matrix $\matX:$
$\FNormS{\matA -  \matY \matDelta \pinv{(\matY \matDelta)}\matA } \le \FNormS{ \matA -  \matY \matDelta \matX  }.$
Also,
$\pinv{\left(\matY \matDelta\right)} = \pinv{\matDelta} \pinv{\matY} = \matDelta\transp \matY\transp,$
because both matrices are orthonormal.
\end{proof}

Next, we state a version of the above lemma for the transpose of $\matA$, equivalently for the best rank $k$ approximation within the row space of a given subspace $\matR$.
Let $\matA \in \mathbb{R}^{m \times n}$, let $k < n$ be an integer, and let $\matR \in \mathbb{R}^{c \times n}$ with $k < c < m $.
$\Pi_{\matR,k}^\mathrm{F}(\matA) \in \mathbb{R}^{m \times n}$ is the best rank \math{k} approximation to \math{\matA} in the row span of \math{\matR}.
Equivalently, we can write
$$\Pi_{\matR, k}^\mathrm{F}(\matA) = \matX_{opt} \matR,$$
where
$$
\matX_{opt} = \argmin_{\matX \in {\R}^{m \times c}:\rank(\matX)\leq k}\FNormS{\matA-\matX \matR}.
$$
In order to compute $\Pi_{\matR,k}^{\mathrm{F}}(\matA)$ given $\matA$,
$\matR$, and $k$, one can use the following algorithm:
\begin{center}
\begin{algorithmic}[1]
\STATE $\matR\transp = \matY \matZ$ is a $qr$ decomposition of $\matR\transp$ with $\matY \in \R^{n \times c}$ and $\matZ \in \R^{c \times c}$.
This step requires $O(n c^2)$ arithmetic operations.
\STATE
$\matXi =  \matA \matY\in \R^{m \times c}$. This step requires \math{O(mnc)} arithmetic operations.
\STATE
 $\matXi_k = \tilde{\matU}_k \tilde{\matSig}_k \matDelta \transp \in \R^{m \times c}$ is a rank $k$ SVD of $\matXi$
 with $\tilde{\matU}_k \in \R^{m \times k}, \tilde{\matSig}_k \in \R^{k \times k}, $ and $\matDelta \in \R^{c \times k}.$
This step requires \math{O( m c^2)} arithmetic operations.
\STATE Return $ \matA \matY \matDelta\matDelta\transp \matY\transp  \in \mathbb{R}^{m \times n}$ of rank at most $k$.
\end{algorithmic}
\end{center}
\medskip
Notice that $ \matA \matY \matDelta\matDelta\transp \matY\transp  \in \mathbb{R}^{m \times n}$ is a rank $k$ matrix that lies in the row span of $\matR$.
The next lemma is a corollary of Lemma~\ref{lem:bestF} when applied with $\matA := \matA\transp$ and $\matV := \matR\transp$.
\begin{lemma}\label{lem:bestFrows}
Given $\matA \in {\R}^{m \times n}$, $\matR\in\R^{c \times n}$ and an integer $k$,
$\matY \matDelta\matDelta\transp \matY\transp$ and  $ \tilde{\matU}_k \tilde{\matSig}_k \matDelta\transp$ satisfy:
$$
\FNormS{\matA -  \matA \matY \matDelta\matDelta\transp \matY\transp } \le \FNormS{ \matA-  \tilde{\matU}_k \tilde{\matSig}_k \matDelta\transp \matR  } =
\FNormS{\matA-\Pi_{\matR,k}^{\mathrm{F}}(\matA)}.
$$
The above algorithm requires $O(mnc + mc^2)$ arithmetic operations to construct $\matY, \matZ,$ and $\matDelta$.
We will denote the above procedure as
$ [ \matY, \matZ, \matDelta ] = BestRowSubspaceSVD(\matA, \matR, k). $
\end{lemma}

\section{Distributed PCA in the arbitrary partition model}\label{sec:dense}
This section describes a fast distributed PCA algorithm with total communication $O(m s k)$ words plus low order terms, which is optimal in the arbitrary partition model in the sense that an $\tilde{\Omega}( m s k)$ bits lower bound was given by~\cite{KVW14}. 
The algorithm employs, in a novel way, the notion of {\it projection-cost preserving sketches} from~\cite{cohen2014dimensionality}.
In particular, whereas
all previous~\cite{Sar06,CW13arxiv} dimension-reduction-based SVD methods reduce one dimension of the input matrix to compute some approximation to the SVD, our method reduces {\it both} dimensions and computes an approximation to the SVD from a small almost square matrix. Unlike~\cite{KVW14} which reduces only one dimension in the first communication round, we do the reduction on {\it both} dimensions in the same round.

We first present the batch version of the algorithm which offers a new low-rank matrix approximation technique;
a specific implementation of this algorithm offers a communication-optimal distributed PCA algorithm (we also discuss
in Section~\ref{sec:streaming} a variant of this algorithm that offers a two-pass space-optimal PCA method in the
turnstile streaming model). Before presenting all these new algorithms in detail, we present the relevant results from the
previous literature that we employ in the analysis.

\subsection{Projection-cost preserving sketching matrices}
In this section, we recap a notion of sketching matrices which we call ``projection-cost preserving sketching matrices''. A sketching matrix from this family is a linear matrix transformation and it has the property that for all projections it preserves, up to some error, the difference between the matrix in hand and its projection in Frobenius norm.

\begin{definition}[Projection-cost preserving sketching matrices] \label{def:pcpsm}
We say that $\matW \in \mathbb{R}^{n \times \xi}$ is an $(\varepsilon,k)$-projection-cost preserving sketching matrix of $\matA\in \mathbb{R}^{m\times n}$, if for all rank-$k$ orthogonal projection matrices $\matP\in \mathbb{R}^{m\times m}$, it satisfies
$$ (1-\varepsilon)\FNormS{\matA-\matP\matA} \leq \FNormS{\matA\matW-\matP\matA\matW}+c \leq (1+\varepsilon)\FNormS{\matA-\matP\matA}$$
where $c$ is a non-negative constant which only depends on $\matA$ and $\matW$. We also call $\matA\matW$ an $(\varepsilon,k)$-projection-cost preserving sketch of $\matA$.
\end{definition}

Due to the following lemma, we know that a good rank-$k$ approximation projection matrix of $(\varepsilon,k)$-projection-cost preserving sketch $\matA\matW$ also provides a good rank-$k$ approximation to $\matA$.

\begin{lemma}[PCA via Projection-Cost Preserving Sketches - Lemma 3 in~\cite{cohen2014dimensionality}] \label{fact:pcavpcps}
Suppose $\matW\in \mathbb{R}^{n \times \xi}$ is an $(\varepsilon,k)$-projection-cost preserving sketching matrix of $\matA\in \mathbb{R}^{m\times n}$. Let
$\hat{\matP}^{*} = \arg \min_{rank(\matP)\leq k} \FNormS{\matA\matW-\matP\matA\matW}$.
For all $\hat{\matP},\varepsilon'$ satisfying $rank(\hat{\matP})\leq k,\varepsilon'\geq 0$, if $\FNormS{\matA\matW-\hat{\matP}\matA\matW}\leq (1+\varepsilon')\FNormS{\matA\matW-\hat{\matP}^*\matA\matW}$,
$$\FNormS{\matA-\hat{\matP}\matA}\leq \frac{1+\varepsilon}{1-\varepsilon}\cdot(1+\varepsilon')\FNormS{\matA-\matA_k}$$
\end{lemma}

\cite{cohen2014dimensionality} also provides several ways to construct projection-cost preserving sketching matrices. Because we mainly consider the communication, we just choose one which can reduce the dimension as much as possible. Furthermore, it is also an oblivious projection-cost preserving sketching matrix.

\begin{lemma}[Dense Johnson-Lindenstrauss matrix - part of Theorem 12 in~\cite{cohen2014dimensionality}]\label{lem:djlm}
For $\varepsilon<1$, suppose each entry of $\matW\in \mathbb{R}^{n\times \xi}$ is chosen $O(\log(k))$-wise independently and uniformly in $\{1/\sqrt{\xi},-1/\sqrt{\xi}\}$ where $\xi=O(k\varepsilon^{-2})$~\cite{CW09}. For any $\matA\in \mathbb{R}^{m\times n}$, with probability at least $0.99$, $\matW$ is an $(\varepsilon,k)$-projection-cost preserving sketching matrix of $\matA$.
\end{lemma}

\subsection{A batch algorithm for the fast low rank approximation of matrices}\label{sec:algbatch}
In this section, we describe a new method for quickly computing a low-rank approximation to a given matrix.
This method does not offer any specific advantages over previous such techniques~\cite{Sar06,CW13arxiv,boutsidis2013improved}; however,
this new algorithm can be implemented efficiently in the distributed setting (see Section~\ref{sec:algorithm3}) and in
the streaming model of computation (see Section~\ref{sec:streaming}); in fact we are able to obtain communication-optimal and
space-optimal results, respectively.  For completeness as well as ease of presentation, we first present and analyze the simple batch version of the algorithm. The algorithm uses the dense Johnson-Lindenstrauss matrix of Lemma~\ref{lem:djlm} in order to reduce both dimensions of $\matA,$ before computing some sort of SVD to a
$\poly(k/\varepsilon) \times \poly(k/\varepsilon)$ matrix (see Step $2$ in the algorithm below). 

Consider the usual inputs: a matrix $\matA \in \R^{m \times n},$ a rank parameter $k < \rank(\matA),$ and an accuracy parameter $0 < \varepsilon < 1$.
The algorithm below returns an orthonormal matrix $\matU \in \R^{m \times k}$ such that
$$
\FNormS{\matA - \matU \matU\transp \matA} \le (1 + \varepsilon)
\cdot \FNormS{\matA-\matA_k}.
$$

\vspace{0.2in}
\begin{small}
{\bf Algorithm}

\begin{enumerate}

\item Construct two dense Johnson-Lindenstrauss matrices $\matS\in\mathbb{R}^{\xi_1\times m},\matT\in\mathbb{R}^{n \times \xi_2}$ with $\xi_1=O(k\varepsilon^{-2}),\xi_2=O(k\varepsilon^{-2})$~(see Lemma~\ref{lem:djlm}).

\item Construct $\tilde{\matA}=\matS\matA\matT$.

\item Compute the SVD of $\tilde{\matA}_k=\matU_{\tilde{\matA}_k}\matSig_{\tilde{\matA}_k}\matV_{\tilde{\matA}_k}\transp$
~($\matU_{ \tilde{\matA}_k } \in \R^{\xi_1 \times k}$, $\matSig_{ \tilde{\matA}_k } \in \R^{k \times k}$, $\matV_{ \tilde{\matA}_k } \in \R^{\xi_2 \times k})$.

\item Construct $\matX=\matA\matT\matV_{\tilde{\matA}_k}$

\item Compute an orthonormal basis $\matU \in \mathbb{R}^{m\times k}$ for $span(\matX)$ (notice that $rank(\matX) \le k$).

\end{enumerate}

\end{small}

Theorem~\ref{thmbatch} later in this section analyzes the approximation error and the running time of the previous algorithm.
First, we prove the accuracy of the algorithm.

\begin{lemma}\label{lembatch0}
The matrix $\matU \in \R^{m \times k}$ with $k$ orthonormal columns satisfies with probability at least $0.98$:
$$
 \FNormS{\matA -  \matU \matU\transp \matA} \le
\left(1 +\varepsilon \right) \cdot  \FNormS{\matA - \matA_k}.
$$
\end{lemma}
\begin{proof}
$$\FNormS{\tilde{\matA}-\tilde{\matA}_k}=\FNormS{\tilde{\matA}-\tilde{\matA}\matV_{\tilde{\matA}_k}\matV_{\tilde{\matA}_k}\transp}
=\FNormS{\matS\matA\matT-\matS\matA\matT\matV_{\tilde{\matA}_k}\matV_{\tilde{\matA}_k}\transp}=\FNormS{\matT\transp\matA\transp\matS\transp-\matV_{\tilde{\matA}_k}\matV_{\tilde{\matA}_k}\transp\matT\transp\matA\transp\matS\transp}$$

The first equality follows by the SVD of $\tilde{\matA}_k=\matU_{\tilde{\matA}_k}\matSig_{\tilde{\matA}_k}\matV_{\tilde{\matA}_k}\transp$. The second equality is by the construction of $\tilde{\matA}=\matS\matA\matT$. The third equality is due to $\forall \matM,\FNormS{\matM}=\FNormS{\matM\transp}$.

Due to Lemma \ref{lem:djlm}, with probability at least $0.99$, $\matS\transp$ is an $(\varepsilon,k)$-projection-cost preserving sketch matrix of $\matT\transp\matA\transp$. According to Lemma \ref{fact:pcavpcps},
\begin{equation} \label{eqn:1dimreduce}
\FNormS{\matA\matT-\matA\matT\matV_{\tilde{\matA}_k}\matV_{\tilde{\matA}_k}\transp}
=\FNormS{\matT\transp\matA\transp-\matV_{\tilde{\matA}_k}\matV_{\tilde{\matA}_k}\transp\matT\transp\matA\transp}
\leq \frac{1+\varepsilon}{1-\varepsilon}\cdot\FNormS{\matA\matT-(\matA\matT)_k}
\end{equation}

Observe that
\eqan{
\FNormS{\matU\matU\transp\matA\matT-\matA\matT}
& = & \FNormS{\matX\pinv{\matX}\matA\matT-\matA\matT} \\
& \leq & \FNormS{\matX\matV_{\tilde{\matA}_k}\transp-\matA\matT} \\
& = & \FNormS{\matA\matT\matV_{\tilde{\matA}_k}\matV_{\tilde{\matA}_k}\transp-\matA\matT} \\
& \leq & \frac{1+\varepsilon}{1-\varepsilon}\cdot \FNormS{\matA\matT-(\matA\matT)_k}
}

The first equality uses the fact that $\matU$ is an orthonormal basis of $span(\matX)$. The first inequality is followed by $\forall \matX,\matM,\matN,\FNormS{\matX\pinv{\matX}\matM-\matM}\leq\FNormS{\matX\matN-\matM}$. The second equality uses the construction that $\matX=\matA\matT\matV_{\tilde{\matA}_k}$. The second inequality follows by Eqn~(\ref{eqn:1dimreduce}).

Due to Lemma \ref{lem:djlm}, with probability at least $0.99$, $\matT$ is an $(\varepsilon,k)$-projection-cost preserving sketch matrix of $\matA$. Due to Lemma \ref{fact:pcavpcps},
$$\FNormS{\matU\matU\transp\matA-\matA} \leq \frac{(1+\varepsilon)^2}{(1-\varepsilon)^2}\cdot \FNormS{\matA-\matA_k}$$
Due to union bound, the probability that $\matS\transp$ is an $(\varepsilon,k)$-projection-cost preserving sketch matrix of $\matT\transp\matA\transp$ and $\matT$ is an $(\varepsilon,k)$-projection-cost preserving sketch matrix of $\matA$ is at least $0.98$. Note that $\frac{(1+\varepsilon)^2}{(1-\varepsilon)^2}$ is $1+O(\varepsilon)$ when $\varepsilon$ is small enough, so we can adjust $\varepsilon$ here by a constant factor to show the statement.
\end{proof}

Next, we present the main theorem.

\begin{theorem}\label{thmbatch}
The matrix $\matU \in \R^{m \times k}$ with $k$ orthonormal columns satisfies with probability at least $0.98$:
$$
 \FNormS{\matA -  \matU \matU\transp \matA} \le
\left(1 +\varepsilon \right) \cdot  \FNormS{\matA - \matA_k}.
$$
The running time of the algorithm is
$$
O\left(n m k \varepsilon^{-2} + m k^2 \varepsilon^{-4} + \poly(k \varepsilon^{-1})\right).
$$
\end{theorem}
\begin{proof}
The correctness is shown by Lemma~\ref{lembatch0}.

\paragraph{Running time.}
Next, we analyze the running time of the algorithm:
\begin{enumerate}
\item There are a total of $(\xi_1\times m+\xi_2 \times n)$ entries of $\matS$ and $\matT$. It is enough to generate them in $O((n+m)k\varepsilon^{-2})$ operations.
\item We first compute $\matA\matT$ with $O(m n \xi_2)$ arithmetic operations. Then, we compute $\matS\matA\matT$ with $O(\xi_1 m \xi_2)$ arithmetic operations.
\item This step requires $O(\poly(k/\varepsilon))$ operations since we compute the SVD of a
$O(\poly(k/\varepsilon)) \times O(\poly(k/\varepsilon))$ matrix $\tilde{\matA}$.
\item We already have $\matA\matT$ from the second step. Hence, $O(m \xi_2 k)$ additional arithmetic operations suffice to compute $\matA\matT\matV_{\tilde{\matA}_k}$
\item  $O(m k^2)$ operations suffice to compute an orthonormal basis for $\matX$, e.g., with a QR factorization.
\end{enumerate}
\end{proof}

\subsection{The distributed PCA algorithm}\label{sec:algorithm3}
Recall that the input matrix $\matA \in \R^{m \times n}$ is partitioned arbitrarily as:
$\matA = \sum_i^s \matA_i$
for $i=1:s,$ $\matA_i \in \R^{m \times n}$.
The idea in the algorithm below is to implement the algorithm in Section~\ref{sec:algbatch} in the distributed setting.

\vspace{0.2in}
\begin{small}
{\bf Input:}
\begin{enumerate}
\item $\matA \in \R^{m \times n}$ arbitrarily partitioned
$
\matA = \sum_i^s \matA_i
$
for $i=1:s,$ $\matA_i \in \R^{m \times n}$.
\item rank parameter $k < \rank(\matA)$
\item accuracy parameter $\varepsilon > 0$
\end{enumerate}

{\bf Algorithm}

\begin{enumerate}

\item Machines agree upon two dense Johnson-Lindenstrauss matrices $\matS\in\mathbb{R}^{\xi_1\times m},\matT\in\mathbb{R}^{n \times \xi_2}$ with $\xi_1=O(k\varepsilon^{-2}),\xi_2=O(k\varepsilon^{-2})$~(see Lemma~\ref{lem:djlm}).

\item Each machine locally computes $\tilde{\matA_i}=\matS\matA_i\matT$ and sends $\tilde{\matA_i}$ to the server. Server constructs $\tilde{\matA}=\sum_i \tilde{\matA_i}$.

\item Server computes the SVD of $\tilde{\matA}_k=\matU_{\tilde{\matA}_k}\matSig_{\tilde{\matA}_k}\matV_{\tilde{\matA}_k}\transp$
~($\matU_{ \tilde{\matA}_k } \in \R^{\xi_1 \times k}$, $\matSig_{ \tilde{\matA}_k } \in \R^{k \times k}$, $\matV_{ \tilde{\matA}_k } \in \R^{\xi_2 \times k})$.

\item Server sends $\matV_{\tilde{\matA}_k}$ to all machines.

\item Each machine construct $\matX_i=\matA_i\matT\matV_{\tilde{\matA}_k}$ and sends $\matX_i$ to the server. Server constructs $\matX=\sum_i \matX_i$.

\item Server computes an orthonormal basis $\matU \in \mathbb{R}^{m\times k}$ for $span(\matX)$ (notice that $rank(\matX) \le k$).

\item Server sends $\matU$ to each machine.

\end{enumerate}

\end{small}

Notice that in the first step, $\matS$ and $\matT$ can be described using a random seed that is $O(\log(k))$-wise independent due to Lemma \ref{lem:djlm}.

\subsubsection{Main result}
The theorem below analyzes the approximation error, the communication complexity,
and the running time of the previous algorithm. Notice that the communication cost of this algorithm is only given in terms of ``real numbers''.
The only step where we can not bound the length of a machine word is when the server communicates $\matV_{\tilde{\matA}_k}$ to all machines; and this is because the entries of $\matV_{\tilde{\matA}_k}$ could be unbounded~(see the discussion regarding the upper bounds in Section~\ref{sec:subtmp}). We resolve this issue in the following section.

\begin{theorem}\label{thmd}
The matrix $\matU \in \R^{m \times k}$ with $k$ orthonormal columns satisfies w.p.
$0.98$:
\begin{equation}\label{eqnthm2d}
 \FNormS{\matA -  \matU \matU\transp \matA} \le
\left(1 + \varepsilon \right) \cdot  \FNormS{\matA - \matA_k}.
\end{equation}
The communication cost of the algorithm is
$$
O(m s k + s \cdot \poly(k/\varepsilon))
$$
``real numbers''
and the  running time is of the order
$$
O\left(n m k \varepsilon^{-2} + m k^2 \varepsilon^{-4} + \poly(k \varepsilon^{-1})\right).
$$
\end{theorem}
\begin{proof}

The matrix $\matU$ - up to the randomness in the algorithm - is exactly the same matrix as in the batch algorithm in Section~\ref{sec:algbatch}, hence Theorem~\ref{thmbatch} proves Eqn.~\ref{eqnthm2d}.

The algorithm communicates $O(m s k + s \cdot \poly(k/\varepsilon))$ real numbers in total:  $O(s \cdot \poly(k/\varepsilon))$ in steps 2 and 4, and $O(s m k)$ in steps 5 and 7.

The operations in the algorithm are effectively the same operations as in the batch algorithm in
Section~\ref{sec:algbatch}, hence the analysis of the running time in Theorem~\ref{thmbatch} shows the claim.
\end{proof}

\section{Obtaining bit complexity for the distributed PCA algorithm}\label{sec:precisionDistributed}
For the algorithm in the previous section, we were only able to provide a communication upper bound in terms of ``real numbers''.  In this section, we describe how to obtain a communication upper bound in terms of words for the above protocol, where each word is $O(\log(mnsk/\varepsilon))$ bits.

The basic idea is that we have a case analysis depending on the rank of the matrix $\matA$. If the rank of $\matA$ is less than or equal to $2k,$ we follow one distributed protocol and if the rank is at least $2k$ we follow
a different protocol. In Section~\ref{sec:bit1}, Section~\ref{sec:bit2}, and Section~\ref{sec:bit3} we describe the algorithm that tests the rank of a distributed matrix, and the two PCA protocols, respectively.
Then, in Section~\ref{sec:bit4} we give the details of the overall algorithm
and in Section~\ref{sec:bit5} we give its analysis.

\subsection{Testing the rank of a distributed matrix}\label{sec:bit1}

\begin{lemma}\label{lem:rankTest}
Given $\matA \in \R^{m \times n}$ and a rank parameter $k < \rank(\matA),$
there exists a distributed protocol in the arbitrary partition model to test if the rank of $\matA$ is less than or equal to $2k$ using
$O(sk^2)$ words of communication and succeeding with probability $1-\delta$ for an arbitrarily small constant $\delta > 0$.
\end{lemma}
\begin{proof}
This is an immediate implementation of a streaming algorithm due to \cite{CW09}
for testing if an $n \times n$ matrix $\matA$
has rank at least $2k$ in the streaming model, using $O(k^2)$ words of space.
In that algorithm, there is a fixed $6nk/\delta \times n$ matrix $\matH$ whose entries are integers of magnitude
at most $\poly(n)$,
where $\delta > 0$ is an arbitrarily small constant. The algorithm simply chooses $4k$ random rows
from $\matH$. Letting $\matH'$ be the $2k \times n$ matrix of the first $2k$ random rows, and
$\matH''$ be the $n \times 2k$
matrix whose columns are the next $2k$ randomly chosen rows, the algorithm just declares that $\matA$ has
rank at least $k$ iff the rank of $\matH'\matA \matH''$ is $2k$.

The above streaming algorithm can be implemented in the distributed setting by having the coordinator choose
$4k$ random rows of the fixed, known matrix, and send the row identities to each of the machines.
This only takes $O(sk)$ words of communication. Then machine $i$ computes $\matH' \matA_i \matH''$, and returns
this to the coordinator. The coordinator can then add these up to compute $\matH' \matA \matH''$ and compute
its rank. The total communication is $O(sk^2)$ words and the protocol succeeds with probability at least
$1-\delta$ for an arbitrarily small constant $\delta > 0$. Note that we can assume our input matrix $\matA$,
which is $m \times n$, is a square $n \times n$ matrix by padding with all-zeros rows.
Those rows of course, will never get communicated in the implementation described above.
\end{proof}

\subsection{Distributed PCA protocol when $\rank(\matA) \le 2k$}\label{sec:bit2}

\subsubsection{Subsampled Randomized Hadamard Transform and Affine Embeddings}
Our algorithms use the following tool, known as 	``Subsampled Randomized Hadamard Transform'' or SRHT for short, to implement efficiently fast dimension reduction in large matrices.

\begin{definition}[Normalized Walsh--Hadamard Matrix]
\label{def:walsh}
Fix an integer $m = 2^p$, for $p = 1,2,3, ...$. The (non-normalized) $m \times m$ matrix of the Walsh--Hadamard transform is defined recursively as,
\vspace{-.0751in}
$$ \matH_n = \left[
\begin{array}{cc}
  \matH_{m/2} &  \matH_{m/2} \\
  \matH_{m/2} & -\matH_{m/2}
\end{array}\right],
\qquad \mbox{with} \qquad
\matH_2 = \left[
\begin{array}{cc}
  +1 & +1 \\
  +1 & -1
\end{array}\right].
$$
The $m \times m$ normalized matrix of the Walsh--Hadamard transform is equal to $\matH = m^{-\frac{1}{2}} \matH_m \in \R^{m \times m}.$
\end{definition}

\begin{definition}[Subsampled Randomized Hadamard Transform (SRHT) matrix]
\label{def:srht}
Fix integers $\xi$ and $m = 2^p$ with $\xi < m$ and $p = 1,2,3, ...$. An SRHT matrix is an $\xi \times m$ matrix of the form $$ \matT = \sqrt{\frac{m}{\xi}} \cdot \matR \matH  \matD;$$
\begin{itemize}
\item $\matD \in \R^{m \times m}$ is a random diagonal matrix whose entries are independent random signs, i.e. random variables uniformly distributed on $\{\pm 1\}$.
\item $\matH \in \R^{m \times m}$ is a normalized Walsh--Hadamard matrix (see Definition~\ref{def:walsh}).
\item $\matR \in \R^{\xi \times m}$ is a subset or $r$ rows from the $n \times n$ identity matrix, where the rows are  chosen uniformly at random and without replacement.
\end{itemize}
\end{definition}

The next lemma argues that an SRHT matrix is a so-called ``affine embedding matrix''.
The SRHT is one of the possible choices of Lemma 32 in~\cite{CW13arxiv} that will satisfy the lemma.
\begin{lemma}[Affine embeddings - Theorem 39 in~\cite{CW13arxiv}]\label{lem:affine}
Suppose $\matG$ and $\matH$ are matrices with $m$ rows, and $\matG$ has rank at most $r$.
Suppose $\matT$ is a $\xi \times m$ SRHT matrix~(see Definition~\ref{def:srht})
with $\xi = O(r / \varepsilon^2)$. Then, with probability $0.99$, for all $\matX$ simultaneously:
$$
(1 - \varepsilon) \cdot \FNormS{\matG\matX - \matH}
\le
\FNormS{\matT (\matG\matX - \matH)}
\le
(1 + \varepsilon) \cdot \FNormS{\matG\matX - \matH}.
$$
\end{lemma}
Finally, we note that matrix-vector multiplications with SRHT's are fast.
\begin{lemma} [Fast Matrix-Vector Multiplication, Theorem 2.1 in~\cite{AL08}]
\label{prop:SRHT-compute-time}
Given $\x \in \R^m$ and $\xi < m$, one can construct $\matT \in \R^{\xi \times m}$
and compute $\matT \x$ in at most $2 m \log_2(\xi + 1) )$ operations.
\end{lemma}

\subsubsection{Generalized rank-constrained matrix approximations}\label{sec:Uopt}
Let $\matM \in \R^{m \times n}$,
$\matN \in \R^{m \times c}$,
$\matL \in \R^{r \times n}$,
and $k \le c,r$ be an integer.
Consider the following optimization problem,
$$ \matX_{opt}  \in \argmin_{ \matX \in \R^{c \times r}, \rank(\matX)\le k } \FNormS{ \matM - \matN \matX \matL }.  $$
Then, the solution  $\matX_{opt} \in \R^{c \times r}$ with $\rank(\matX_{opt}) \le k$ that has the minimum $\FNorm{\matX_{opt}}$ out of all possible feasible solutions
is given via the following formula,
$$ \matX_{opt}  =  \pinv{\matN}  \left( \matU_{\matN}\matU_{\matN}\transp \matM \matV_{\matL}\matV_{\matL}\transp \right)_k \pinv{\matL}.$$
$\left( \matU_{\matN}\matU_{\matN}\transp \matM \matV_{\matL}\matV_{\matL}\transp \right)_k \in \R^{m \times n}$ of rank at most $k$ denotes the best rank $k$ matrix to
$\matU_{\matN}\matU_{\matN}\transp \matM \matV_{\matL}\matV_{\matL}\transp \in \R^{m \times n}$.
This result was proven in~\cite{FT07} (see also~\cite{SR12} for the spectral norm version of the problem).

\subsubsection{The PCA protocol}

\begin{lemma}\label{lem:lowRankProtocol}
Suppose the rank $\rho$ of $\matA \in \R^{m \times n}$ satisfies $\rho \leq 2k$, for some rank parameter $k$.
Then, there is a protocol for the Distributed
Principal Component Analysis Problem in the arbitrary partition model
using $O(smk + sk^2/\varepsilon^2)$ words of communication and succeeding with probability $1-\delta$ for
an arbitrarily small constant $\delta > 0$.
\end{lemma}
\begin{proof}
The $n \times 2k$ matrix $\matH''$
chosen in the protocol of Lemma \ref{lem:rankTest} satisfies that with probability $1-\delta$,
for an arbitrarily small constant $\delta > 0$, the rank of $\matA \matH''$ is equal to the rank of $\matA$
if $\rho < 2k$. Indeed, if this were not true, the algorithm could not be correct, as the rank
of $\matH' \matA \matH''$ is at most the rank of $\matA \matH''$ (and the same algorithm can be used for
any $\rho < 2k$). It follows that with probability $1-\delta$, the column span of $\matA \matH''$ is equal
to the column span of $\matA$. Hence, as in the protocol of Lemma \ref{lem:rankTest}, the coordinator
learns the column space of $\matA$, which can be described with $O(km)$ words. The coordinator thus communicates
this to all machines, using $O(skm)$ total words of communication, assuming the entries of $\matA$ are integers
of magnitude at most $\poly(mns/\varepsilon)$.

Let $\matC = \matA \matH''$, which is $m \times 2k$. We can set up the optimization problem:
\begin{equation} \label{opt:opt}
\min_{\rank(\matX)\le k} \|\matC \matX \matC\transp \matA - \matA\|_{\mathrm{F}}.
\end{equation}
Because the size of $\matC$ is only $m \times 2k$, every machine can know $\matC$ by sending a total of $O(smk)$ words. Since the rank of $\matC$ and $\matA$ are small, we can sketch on the left and right using affine embeddings
$\matT_{left}$ and $\matT_{right}$. Then machine $i$ sends $\matC\transp \matA_i \matT_{right}$ to the coordinator,
together with $\matT_{left} \matA_i \matT_{right}$. The coordinator computes $\matC\transp \matA \matT_{right}$
and $\matT_{left} \matA \matT_{right}$ by adding up the sketches, and sends these back to all the machines.
Each machine can then solve the optimization problem
\begin{equation} \label{opt:opt2}
\min_{\rank(\matX)\le k} \|\matT_{left} \matC \matX \matC\transp \matA \matT_{right} - \matT_{left} \matA \matT_{right}\|_{\mathrm{F}},
\end{equation}
each obtaining the same $\matX_*$ which is the optimal solution to Eqn~(\ref{opt:opt2}).
Due to Lemma \ref{lem:affine}, $\matX_*$ is a $(1+O(\varepsilon))$-approximation to the best solution to Eqn~(\ref{opt:opt}).
Finally, every machine outputs the same orthonormal basis $\matU \in \R^{m \times k}$
for $\matC \matX_*$.

We can construct affine embedding matrices
$\matT_{left} \in \R^{\xi_1 \times m}$ and $\matT_{right} \in \R^{n \times \xi_2}$ with
$\xi_1 = O(k/\varepsilon^2),$ $\xi_2 = O(k/\varepsilon^2)$~(see Definition~\ref{def:srht}).
The total
communication of this protocol is $O(skm + sk^2/\varepsilon^2)$
words and the success probability can be made
$1-\delta$ for an arbitrarily small constant $\delta > 0$.
%

\end{proof}

\subsection{Distributed PCA protocol when $\rank(\matA) > 2k$}\label{sec:bit3}
The idea here is more involved than in the previous subsection and in order to describe the algorithm we need several
intermediate results.

\subsubsection{Lower bounds on singular values of matrices with integers entries}
The first lemma gives a lower bound on the singular values of a matrix with integer entries with bounded magnitude.
\begin{lemma}(Lemma 4.1 of \cite{CW09}, restated)\label{lem:cwbits}
If an $m \times n$ matrix $\matA$ has integer entries bounded in magnitude
by $\gamma$, and has rank $\rho = \rank(\matA)$, then the $k$-th largest singular value
$\matA$ satisfies
$$\sigma_k \geq (mn\gamma^2)^{-k/(2(\rho-k))}.$$
\end{lemma}
\begin{proof}
In the proof of Lemma 4.1 of \cite{CW09}, equation (10), it is shown
that if $\lambda_k$ is the $k$-th
largest eigenvalue of $\matA\transp\matA$, then
$$\lambda_k \geq (mn\gamma^2)^{-k/(\rho-k)}.$$
This implies the $k$-th singular value $\sigma_k$ of $\matA$ satisfies $\sigma_k \geq (mn\gamma^2)^{-k/(2(\rho-k))}.$
\end{proof}
Next, we state two immediate corollaries of this lemma for future reference.
\begin{corollary}\label{cor:plusOne}
If an $m \times n$ matrix $\matA$ has integer entries bounded in magnitude
by $\poly(mns/\varepsilon)$ and $\|\matA-\matA_k\|_{\mathrm{F}} > 0$, then
$$\|\matA-\matA_k\|_{\mathrm{F}} > (mns/\varepsilon)^{-O(k)}.$$
\end{corollary}
\begin{proof}
Since $\|\matA-\matA_k\|_{\mathrm{F}} > 0$, the rank $\rho$ of $\matA$ is at least $k+1$,
so by the preceding lemma,
$$\|\matA-\matA_k\|_{\mathrm{F}} \geq \sigma_{k+1}
\geq (\poly(mns/\varepsilon))^{-k/2},
$$
as desired.
\end{proof}
\begin{corollary}\label{cor:plusTwo}
If an $m \times n$ matrix $\matA$ has integer entries bounded in magnitude
by $\poly(mns/\varepsilon)$ and rank$(\matA) \geq 2k$, then
$$\|\matA-\matA_k\|_{\mathrm{F}} > 1/\poly(mns/\varepsilon).$$
\end{corollary}
\begin{proof}
This follows by plugging $\rho = 2k$ into Lemma \ref{lem:cwbits}.
\end{proof}

\subsubsection{Lower bounds on singular values of integer-perturbed matrices}
In this section, we describe a perturbation technique for matrices and provide lower bounds
for the smallest singular value of the perturbed matrix.  We start with a theorem of Tao and Vu.
\begin{theorem}(Theorem 2.5 of \cite{tv07})\label{thm:tv}
Let $\matM$ be an $n \times n$ matrix with integer entries bounded in magnitude by $n^C$ for a constant $C > 0$. Let $\matN_n$ be a matrix with independent entries each chosen to be $1$ with probability $1/2$, and $-1$ with probability $1/2$. Then, there is a constant $B > 0$ depending
on $C$, for which
$$\Pr[\|(\matM + \matN_n)^{-1}\|_2 \geq n^B] \leq 1/n.$$
\end{theorem}
In words, the result indicates that the spectral norm of the inverse of the perturbed matrix is bounded from below with high probability.
We now describe a simple corollary of this result.
\begin{corollary}\label{cor:tv}
Let $\matM$ be an $n \times n$ matrix with integer entries bounded in magnitude by $n^C$ for
a constant $C > 0$. Let $\matN_n$ be a matrix with independent entries each chosen to be
$1/n^D$ with probability $1/2$, and $-1/n^D$ with probability $1/2$, where $D > 0$ is a constant. Then, there is a constant $B > 0$ depending on $C$ and $D$ for which
$$
\Pr[\|(\matM + \matN_n)^{-1}\|_2 \geq n^B] \leq 1/n.
$$
\end{corollary}
\begin{proof}
This follows by Theorem \ref{thm:tv} after replacing the constat $C$ in that theorem with $C + D$, and scaling by $n^D$.
\end{proof}

We need to generalize Corollary \ref{cor:tv} to rectangular matrices since we will eventually apply this perturbation technique to the matrix $\matA$ to which we would like to compute a distributed PCA.
\begin{lemma}\label{lem:rectangular}
Let $\matM$ be an $m \times n$ matrix with integer entries bounded in magnitude by $n^C$ for a constant $C > 0$, and suppose $m \leq n$. Let $\matN_{m,n}$ be a matrix with independent entries each chosen to be $1/n^D$ with probability $1/2$ and $-1/n^D$ with probability $1/2$, where $D > 0$ is a constant. Then, there is a constant $B > 0$ depending on $C$ and $D$ for which
$$\Pr[\sigma_m(\matM + \matN_{m,n}) < 1/n^B] \leq 1/n,$$
where $\sigma_m(\matM + \matN_{m,n})$ denotes the smallest singular value of $\matM + \matN_{m,n}$.
\end{lemma}
\begin{proof}
Suppose we were to pad $\matM$ with $n-m$ zero rows, obtaining a square matrix $\matM$. Now consider the $n \times n$ matrix $\matN_{n,n}$ with independent entries each chosen to be $1/n^D$ with probability $1/2$ and $-1/n^D$ with probability $1/2$. By Corollary \ref{cor:tv}, all singular values of $\matM+\matN_{n,n}$ are at least $1/n^B$. Now consider a unit vector $\x \in \mathbb{R}^n$ which is zero on all but its top $m$ coordinates. Let $\y \in \mathbb{R}^m$ be the unit vector which agrees with $\x$ on its top $m$ coordinates. Then,
\begin{eqnarray*}
1/n^B & \leq & \|\x (\matM + \matN_{n,n})\|_2\\
& = & \|\y(\matM + \matN_{m,n})\|_2,
\end{eqnarray*}
where the inequality uses the lower bound on the singular values of $\matM+\matN_{n,n}$, which occurs with probability at least $1-1/n$, and the equality follows by definition of the matrices and vectors we have defined. As $\y \in \mathbb{R}^m$ can be chosen to be an arbitrary unit vector, it follows that $\sigma_m(\matM + \matN_{m,n}) \geq 1/n^B$, which completes the proof.
\end{proof}

\subsection{Description of algorithm}\label{sec:bit4}

Using the above results, we are now ready to describe a distributed PCA algorithm whose communication
cost can be bounded in terms of machine words and not just in terms of ``real numbers''.
As in the algorithm in Section~\ref{sec:algorithm3}, we denote with
$\matB_i \in \R^{m \times n}$ the matrix that $\matB_1$ arises from $\matA_1$ after applying the Bernoulli perturbation technique discussed above in Lemma~\ref{lem:rectangular} and $\forall i>1,$ $\matB_i$ is equal to $\matA_i$.
Using this notation, we have $\matB := \sum_i\matB_i \in \R^{m \times n}.$
Notice that $\matB$ exactly arises from $\matA$ after applying the such Bernoulli perturbation technique.

\vspace{0.2in}
\begin{small}
{\bf Input:}
\begin{enumerate}
\item $\matA \in \R^{m \times n}$ arbitrarily partitioned
$
\matA = \sum_i^s \matA_i
$
for $i=1:s,$ $\matA_i \in \R^{m \times n}$.
\item rank parameter $k < \rank(\matA)$
\item accuracy parameter $\varepsilon > 0$
\end{enumerate}

{\bf Algorithm}

\begin{enumerate}

\item Use the protocol of Lemma~\ref{lem:rankTest} with $\delta=0.01$ to test if the rank of $\matA$ is less than $2k$.

\item If $\rank(\matA) \le 2k,$ use the protocol of Lemma~\ref{lem:lowRankProtocol} to find some orthonormal $\matU \in \R^{m \times k}$.

\item If $\rank(\matA) > 2k,$

\begin{enumerate}

\item machine $1$ locally and independently adds $1/n^D$ with probability $1/2$, and $-1/n^D$ with probability $1/2$, to each of the entries of $\matA_1$, where $D$ is the constant of Lemma~\ref{lem:rectangular}. Note that this effectively adds the matrix $\matN_{m,n}$ of Lemma~\ref{lem:rectangular} to the entire matrix $\matA$. For notational convenience let $\matB = \matA + \matN_{m,n}$, $\matB_1 \in \R^{m \times n}$ be the local perturbed matrix of machine $1$ and $\forall i>1,$ $\matB_i$ is just equal to $\matA_i$.

\item Machines agree upon two dense Johnson-Lindenstrauss matrices $\matS\in\mathbb{R}^{\xi_1\times m},\matT\in\mathbb{R}^{n \times \xi_2}$ with $\xi_1=O(k\varepsilon^{-2}),\xi_2=O(k\varepsilon^{-2})$~(see Lemma~\ref{lem:djlm}).

\item Each machine locally computes $\tilde{\matB_i}=\matS\matB_i\matT$ and sends $\tilde{\matB_i}$ to the server. Server constructs $\tilde{\matB}=\sum_i \tilde{\matB_i}$.

\item Server computes the SVD of $\tilde{\matB}_k=\matU_{\tilde{\matB}_k}\matSig_{\tilde{\matB}_k}\matV_{\tilde{\matB}_k}\transp$
~($\matU_{ \tilde{\matB}_k } \in \R^{\xi_1 \times k}$, $\matSig_{ \tilde{\matB}_k } \in \R^{k \times k}$, $\matV_{ \tilde{\matB}_k } \in \R^{\xi_2 \times k})$.

\item Now server rounds each of the entries in $\matV_{ \tilde{\matB}_k }$ to the nearest integer multiple of $1/n^{\gamma}$ for a sufficiently large constant $\gamma > 0$. Let the matrix after the rounding be ${\hat{\matV}}_{\tilde{\matB}_k}$

\item Server sends ${\hat{\matV}}_{\tilde{\matB}_k}$ to all machines.

\item Each machine construct $\hat{\matX_i}=\matB_i\matT{\hat{\matV}}_{\tilde{\matB}_k}$ and sends $\hat{\matX_i}$ to the server. Server constructs $\hat{\matX}=\sum_i \hat{\matX_i}$.


\item Server computes an orthonormal basis $\matU \in \mathbb{R}^{m\times k}$ for $span(\hat{\matX})$ (notice that $rank(\hat{\matX}) \le k$), e.g. with a QR factorization.

\item Server sends $\matU$ to each machine.

\end{enumerate}

\end{enumerate}

\end{small}

\subsection{Main result}\label{sec:bit5}

The theorem below analyzes the approximation error, the communication complexity,
and the running time of the previous algorithm. Notice that the communication cost of
this algorithm is given in terms of machine words.

\begin{theorem}\label{thmd2}
The matrix $\matU \in \R^{m \times k}$ with $k$ orthonormal columns satisfies with arbitrarily large constant probability:
\begin{equation}\label{eqnthm2d2}
 \FNormS{\matA -  \matU \matU\transp \matA} \le
\left(1 + \varepsilon \right) \cdot  \FNormS{\matA - \matA_k}.
\end{equation}
The communication cost of the algorithm is
$
O(m s k + s \cdot \poly(k/\varepsilon))
$
words
and the  running time is
$
O\left(n m k \varepsilon^{-2} + m k^2 \varepsilon^{-4} + \poly(k \varepsilon^{-1})\right).
$
\end{theorem}

\subsection{Proof of Theorem~\ref{thmd2}}\label{sec:proof32}
\paragraph{Proof of Eqn.~\ref{eqnthm2d2}}

If  $\rank(\matA)\le 2k,$ then Lemma~\ref{lem:lowRankProtocol} shows the claim.
If $\rank(\matA) > 2k,$ then the situation is more involved and we prove the approximation bound in the following subsection.

\paragraph{Communication Complexity}
Step 1 requires $O(sk^2)$ words (Lemma~\ref{lem:rankTest}).
Step 2 requires $O(smk + sk^2/\varepsilon^2)$ words (Lemma~\ref{lem:lowRankProtocol}).
Step 3 requires $O(m s k + s \cdot \poly(k/\varepsilon))$ words and this follows from the analysis in
Theorem~\ref{thmd}. The only difference with Theorem~\ref{thmd} is that now all matrices communicated
in the algorithm have real numbers which can be represented efficiently with one machine word of at most
$O(\log(m n s / \varepsilon))$ bits. To see this, note that each entry of ${\hat{\matV}}_{\tilde{\matB}_k}$ is bounded in magnitude by $n^{O(B+C)/\gamma}$. Therefore, the entries of ${\hat{\matV}}_{\tilde{\matB}_k}$ and the entries of $\matX_i=\matB_i\matT{\hat{\matV}}_{\tilde{\matB}_k}$ can each be described using $O(\log n)$ bits, i.e., a constant number of machine words.

\paragraph{Running time}
The operations in the third step in algorithm are effectively the same operations as in the batch algorithm in
Section~\ref{sec:algbatch}~(plus some extra operations whose running time is not asymptotically larger),
hence the analysis of the running time in Theorem~\ref{thmbatch} shows the claim
(the number of operations in the first two steps of the algorithm is asymptotically less than the number of operations
in the third step).

\subsection{Proof of Eqn.~\ref{eqnthm2d2} if $\rank(\matA) > 2k$}

\begin{lemma}[Result in~\cite{CW09}]\label{lem:subspcebd}
Suppose $\matM\in \R^{p \times q}$. For $\varepsilon>0$, let $\matP \in \R^{\xi \times p}$ be a matrix of which entries are $O(q)$-wise independently and uniformly chosen from $\{-1/\sqrt{\xi},+1/\sqrt{\xi}\}$ where $\xi=O(q/\varepsilon^2)$. With arbitrarily large constant probability, $\matP$ is an $\varepsilon$-subspace embedding matrix of column space of $\matM$. Specifically, $\forall \x$, it has
$$(1 - \varepsilon)\|\matM\x\|_2\leq \|\matP\matM\x\|_2 \leq(1+ \varepsilon)\|\matM\x\|_2$$
\end{lemma}

\begin{lemma}[Lemma C.2 in~\cite{li2014sketching}]\label{lem:sglvlbd}
Suppose $\matP$ is an $\varepsilon$-subspace embedding matrix of column space of $\matM\in \R^{p\times q}$~($q\leq p$). With arbitrarily large constant probability, $\forall 1\leq i \leq q$
$$(1-\varepsilon)\sigma_i(\matP\matM)\leq\sigma_i(\matM)\leq(1+\varepsilon)\sigma_i(\matP\matM)$$
where $\sigma_i$ means the $i^{th}$ singular value.
\end{lemma}

It is suffice to show the following:
\begin{lemma}\label{lem:bdinv}
With arbitrarily large constant probability, $\|\pinv{(\matB\matT{\matV}_{\tilde{\matB}_k})}\|_2\leq 8n^B$
\end{lemma}
\begin{proof}
 Notice that $\matB\matT{\matV}_{\tilde{\matB}_k}$ is an $m\times k$ matrix. With high probability, $\matS$ is a $\varepsilon$-subspace embedding matrix of column space of $\matB\matT{\matV}_{\tilde{\matB}_k}$ due to Lemma~\ref{lem:subspcebd}. According to Lemma~\ref{lem:sglvlbd}, $\sigma_{min}(\matB\matT{\matV}_{\tilde{\matB}_k}) \geq (1-\varepsilon)\sigma_{min}(\matS\matB\matT{\matV}_{\tilde{\matB}_k})$ where $\sigma_{min}$ means the minimum singular value. Without loss of generality, we assume $\varepsilon<1/2$. Then $\sigma_{min}(\matB\matT{\matV}_{\tilde{\matB}_k}) \geq \frac12\sigma_{min}(\matS\matB\matT{\matV}_{\tilde{\matB}_k})$. Because ${\matV}_{\tilde{\matB}_k}$ are the top $k$ right singular vectors of $\matS\matB\matT$, $\sigma_{min}(\matS\matB\matT{\matV}_{\tilde{\matB}_k})\geq \sigma_{min}(\matS\matB\matT)$. Applying Lemma~\ref{lem:subspcebd} again, $\matT\transp$ is a constant error subspace embedding matrix of column space of $\matB\transp\matS\transp$. Combining with Lemma~\ref{lem:sglvlbd}, we can make $\sigma_{min}(\matS\matB\matT)\geq \frac12\sigma_{min}(\matS\matB)$. Notice that Lemma~\ref{lem:subspcebd} also implies that $\sqrt{\frac{\xi_1}{m}}\matS\transp$ is a constant error subspace embedding matrix of $\R^{\xi_1}$ which means $\forall \x \in \R^{\xi_1}$, we can make $\|\x\sqrt{\frac{\xi_1}{m}}\matS\|_2>\frac12\|\x\|_2$. Thus, $\forall \x \in \R^{\xi_1},$ $\|\x\sqrt{\frac{\xi_1}{m}}\matS\matB\|_2>\frac12\sigma_{min}(\matB)\|\x\|_2$. Due to Lemma~\ref{lem:rectangular}, $\sigma_{min}(\matB)>1/n^B$ which implies $\sigma_{min}(\matS\matB)>\frac12\cdot\sqrt{\frac{m}{\xi_1}}\sigma_{min}(\matB)>\frac12\sqrt{\frac{\xi_1}{m}}\frac1{n^B}>\frac12\frac1{n^B}$. To conclude $\|\pinv{(\matB\matT{\matV}_{\tilde{\matB}_k})}\|_2\leq 1/\sigma_{min}(\matB\matT{\matV}_{\tilde{\matB}_k})\leq 8n^B$.

\end{proof}

Let $\matE={\matV}_{\tilde{\matB}_k}-\hat{\matV}_{\tilde{\matB}_k}$, we have
\eqan{
\FNorm{\matA-\matU\matU\transp\matA}
& = & \FNorm{\matB-\matN_{m,n}-\matU\matU\transp(\matB-\matN_{m,n})} \\
& \leq & \FNorm{\matB-\matU\matU\transp\matB}+\FNorm{\matN_{m,n}}\\
& = & \FNorm{\matB-(\matB\matT{\hat{\matV}}_{\tilde{\matB}_k})\pinv{(\matB\matT{\hat{\matV}}_{\tilde{\matB}_k})}\matB}+\FNorm{\matN_{m,n}} \\
& \leq & \FNorm{\matB-(\matB\matT({\matV}_{\tilde{\matB}_k}+\matE))\pinv{(\matB\matT{\matV}_{\tilde{\matB}_k})}\matB}+\FNorm{\matN_{m,n}} \\
& \leq & \FNorm{\matB-(\matB\matT{\matV}_{\tilde{\matB}_k})\pinv{(\matB\matT{\matV}_{\tilde{\matB}_k})}\matB}+\FNorm{\matB\matT\matE\pinv{(\matB\matT{\matV}_{\tilde{\matB}_k})}\matB}+\FNorm{\matN_{m,n}} \\
& \leq & (1+O(\varepsilon)) \FNorm{\matB-\matB_k}+\FNorm{\matB\matT\matE\pinv{(\matB\matT{\matV}_{\tilde{\matB}_k})}\matB}\\
& \leq & (1+O(\varepsilon)) \FNorm{\matB-\matB_k}
}
The first equality follows by using the relation $\matB=\matA+\mat{N}_{m,n}$. The first inequality uses the triangle inequality for the Frobenius norm and the fact that $\matI-\matU\matU\transp$ is a projector matrix and can be dropped without increasing the Frobenius norm. The second equality uses the fact that both $\matU\matU^T$ and $(\matB\matT{\hat{\matV}}_{\tilde{\matB}_k})\pinv{(\matB\matT{\hat{\matV}}_{\tilde{\matB}_k})}$ are the same projector matrices. The second inequality follows by $\forall \matM,\matC,\matX,$ $\FNorm{\matM-\matC\pinv{\matC}\matM}\leq\FNorm{\matM-\matC\matX}$ and the relation $\matE={\matV}_{\tilde{\matB}_k}-\hat{\matV}_{\tilde{\matB}_k}$. The third inequality uses the triangle inequality. The fourth inequality uses that $\|\matB-\matB_k\|_{\mathrm{F}} \geq 1/\poly(mns/\varepsilon)$ (follows from Lemma~\ref{lem:rectangular}) while $\FNorm{\matN_{m,n}}$ can be made $1/n^t$ for an arbitrarily large integer $t$, and the fact that $\FNorm{\matB-(\matB\matT{\matV}_{\tilde{\matB}_k})\pinv{(\matB\matT{\matV}_{\tilde{\matB}_k})}\matB}\leq (1+O(\varepsilon))\|\matB-\matB_k\|_{\mathrm{F}}$ (implied by replacing all $\matA$ with $\matB$ in Lemma~\ref{lembatch0}). The last inequality uses the fact $\FNorm{\matB\matT\matE\pinv{(\matB\matT{\matV}_{\tilde{\matB}_k})}\matB}\leq\FNorm{\matB}\FNorm{\matT}\FNorm{\matE}\FNorm{\pinv{(\matB\matT{\matV}_{\tilde{\matB}_k})}}\FNorm{\matB}$ where $\FNorm{\matB},\FNorm{\matT},\FNorm{\pinv{(\matB\matT{\matV}_{\tilde{\matB}_k})}}$ are $poly(nms/\varepsilon)$ (Lemma~\ref{lem:bdinv} implies the bound of $\FNorm{\pinv{(\matB\matT{\matV}_{\tilde{\matB}_k})}}$) and we can make $|\matE|_F$ be $1/n^p$ for an arbitrarily large $p$, and $\|\matB-\matB_k\|_{\mathrm{F}} \geq 1/\poly(mns/\varepsilon)$.

Overall, after rescaling $\varepsilon,$ we have
\begin{equation}\label{eqn:last}
\FNorm{\matA -  \matU \matU\transp \matA}  \le (1+\varepsilon)\|\matB - \matB_k\|_{\mathrm{F}}.
\end{equation}
Finally, we need to relate $\FNorm{\matB - \matB_k}$ to $\FNorm{\matA-\matA_k},$ which we do in the following lemma.

\begin{lemma}\label{lem:LASTFUCKINGLEMMA}
Let $\matA \in \R^{m \times n}$ be any matrix,
$\matN_{m,n}$ be the random matrix of Lemma~\ref{lem:rectangular}, and $\matB = \matA + \matN_{m,n}$. Let $k < \rank(\matA), \eps > 0$. Then, for arbitrarily large constant probability,
$$
\FNorm{\matB - \matB_k} \le (1+\eps)\FNorm{\matA - \matA_k}.
$$
\end{lemma}
\begin{proof}
Since $\matB = \matA + \matN_{m,n}$, the proof idea is to relate the singular values of $\matB$ to the singular values of $\matA$ using Weyl's inequality. Specifically, for $i:1:m$ (recall we assume
$m \le n$) we have
$$
| \sigma_i(\matB) -\sigma_i(\matA)| \le \TNorm{\matN_{m,n}} \le 1/n^D.
$$ Rearranging terms in this inequality, and taking squares of the resulting relation we obtain~(for
$i=1:m$),
\begin{equation}\label{eqn:tmptmp}
\sigma_i^2(\matB) \le \left( \sigma_i(\matA) + 1/n^D \right)^2.
\end{equation}
We manipulate the term $\FNorm{\matB-\matB_k}$ as follows:
\eqan{
\FNorm{\matB - \matB_k}
&=& \sqrt{ \sum_{i=k+1}^m \sigma_i^2(\matB) } \\
&\le& \sqrt{ \sum_{i=k+1}^m \left(  \sigma_i(\matA) + 1/n^D  \right)^2 } \\
&=& \sqrt{ \sum_{i=k+1}^m \left(  \sigma_i^2(\matA) + 1/n^{2 D} + 2 \sigma_i(\matA)/n^D  \right) } \\
&\le& \sqrt{ \sum_{i=k+1}^m \sigma_i^2(\matA) } + \sqrt{1/n^{2 D} } + \sqrt{ 2 \sigma_i(\matA)/n^D } \\
&\le& \FNorm{\matA-\matA_k} + \varepsilon \cdot \FNorm{\matA - \matA_k} + \varepsilon \cdot \FNorm{\matA - \matA_k} \\
&\le& (1 + 2\varepsilon)\FNorm{\matA - \matA_k}.
}
The first inequality uses Eqn.~(\ref{eqn:tmptmp}). In the third inequality, we used
$$ \sqrt{1/n^{2D}} \le \varepsilon \FNorm{\matA - \matA_k},$$
which follows from Corollary~\ref{cor:plusTwo} for a sufficiently large constant $D$. Also, we used
$$
\sqrt{ 2 \sigma_i(\matA)/n^D } \le \sqrt{ 2 \sigma_1(\matA)/n^D } \le \sqrt{ \poly(nms/\eps) / n^D } \le \varepsilon \FNorm{\matA - \matA_k},
$$
where the second inequality follows because $\TNorm{\matA} \le \poly(nms/\eps)$ and the last inequality uses again
Corollary~\ref{cor:plusTwo} for a sufficiently large constant $D$.
\end{proof}

\paragraph{Completing the proof.} Using Lemma~\ref{lem:LASTFUCKINGLEMMA} in Eqn~(\ref{eqn:last}), taking squares in the resulting inequality, and rescaling $\eps$ concludes the proof.

\section{Streaming Principal Component Analysis}\label{sec:streaming}
In this section, we are interested in computing a PCA of a matrix in turnstile streaming model. 
Specifically, there is a \emph{stream} of update operations that the $q^{th}$ operation has form $(i_q,j_q,x_q)$ which indicates that the $(i_q,j_q)^{th}$ entry of $\matA \in \mathbb{R}^{m\times n}$ should be incremented by $x_q$ where $i_q\in\{1,...,m\},j_q\in \{1,...,n\},x_q\in\mathbb{R}$. Initially, $\matA$ is zero.
In the streaming
setting of computation, we are allowed only one pass over the update operations, i.e., the algorithm ``sees'' each update operation one by one and only once.
Upon termination, the algorithm should return a matrix $\matU$ with $k$ orthonormal columns which is a ``good'' basis for $span(\matA)$~.
In Section~\ref{sec:outputu} below, we describe an algorithm which gives a space-optimal streaming algorithm. Further more, we provide a variation of this algorithm in Section~\ref{sec:outputa} which can output $\matA^*_k$ satisfying $\FNormS{\matA-\matA^*_k}\leq (1+\varepsilon)\cdot \FNormS{\matA-\matA_k}$. It meets the space lower bound shown in~\cite{CW09}.
In Section~\ref{sec:twopass}, we relax the problem and we describe a two-pass streaming
algorithm which is a version of the algorithm of Section~\ref{sec:algbatch}; unfortunately, the
space complexity of this algorithm can be bounded only in terms of ``real numbers''.
We fix this in Section~\ref{sec:twopass2} where we describe an optimal two-pass streaming algorithm.

Inputs to the algorithms in Section~\ref{sec:onepass}, Section~\ref{sec:twopass} and Section~\ref{sec:twopass2} are a stream of update operations $(i_1,j_1,x_1),$ $(i_2,j_2,x_2),$ $...,$ $(i_l,j_l,x_l)$,
a rank parameter $k < \rank(\matA),$ and an accuracy parameter $0 < \varepsilon < 1$.

\subsection{One-pass streaming PCA}\label{sec:onepass}

\subsubsection{The algorithm which outputs $\matU$} \label{sec:outputu}
 In the following algorithm, the output should be $\matU$ with orthogonal unit columns and satisfying
 $$\FNormS{\matA-\matU\matU\transp\matA}\leq (1+\varepsilon)\cdot \FNormS{\matA-\matA_k}$$
Our algorithm uses the following property of random sign matrices:

\begin{lemma}[Sketch for regression - Theorem 3.1 in~\cite{CW09}]\label{lem:regre}
Suppose both of $\matA$ and $\matB$ have $m$ rows and $rank(\matA)\leq k$. Further more, if each entry of $\matS\in\mathbb{R}^{\xi \times m}$ is $O(k)$-wise independently chosen from $\{-1,+1\}$ where $\xi=O(k/\varepsilon)$ and
$$\tilde{\matX}=\arg \min_{\matX}\FNormS{\matS\matA\matX-\matS\matB}$$,
with probability at least $0.99$,
$$ \FNormS{\matA\tilde{\matX}-\matB} \leq (1+\varepsilon)\cdot\min_{\matX}\FNormS{\matA\matX-\matB} $$
\end{lemma}

\vspace{0.1in}
\begin{small}
{\bf Algorithm}

\begin{enumerate}

\item Construct random sign sketching matrices $\matS\in \mathbb{R}^{\xi_1 \times m}$ with $\xi_1=O(k\varepsilon^{-1})$ and $\matR\in \mathbb{R}^{n \times \xi_2}$ with $\xi_2=O(k\varepsilon^{-1})$~(see Lemma~\ref{lem:regre})

\item Construct affine embedding matrices $\matT_{left} \in \R^{\xi_3 \times m}$ and $\matT_{right} \in \R^{n \times \xi_4}$ with
$\xi_3 = O(k\varepsilon^{-3}),$ $\xi_4 = O(k/\varepsilon^{-3})$~(see Definition~\ref{def:srht}).

\item Initialize all-zeros matrices: $\matM \in \R^{\xi_3 \times \xi_4}$, $\matL \in \R^{\xi_1\times \xi_4}$, $\matN \in \R^{\xi_3 \times \xi_2}$, $\matD \in \R^{m \times \xi_2}$. We will maintain $\matM,\matL,\matN,\matD$ such that $\matM=\matT_{left}\matA\matT_{right},$ $\matL = \matS \matA \matT_{right},$ $\matN=\matT_{left}\matA\matR$ and $\matD=\matA\matR$.

\item For $(i_q,j_q,x_q):= (i_1,j_1,x_1), ..., (i_l,j_l,x_l)$ (one pass over the stream of update operations)

\begin{enumerate}
\item For all $i=1,...,\xi_3,~j=1,...,\xi_4,$ let $\matM_{i,j}=\matM_{i,j}+(\matT_{left})_{i,i_q}\cdot x_q\cdot (\matT_{right})_{j_q,j}$.
\item For all $i=1,...,\xi_1,~j=1,...,\xi_4,$ let $\matL_{i,j}=\matL_{i,j}+\matS_{i,i_q} \cdot x_q \cdot (\matT_{right})_{j_q,j}$.
\item For all $i=1,...,\xi_3,~j=1,...,\xi_2,$ let $\matN_{i,j}=\matN_{i,j}+(\matT_{left})_{i,i_q}\cdot x_q \cdot \matR_{j_q,j}$.
\item For all $j=1,...,\xi_2,$ let $\matD_{i_q,j}=\matM_{i_q,j}+x_q\cdot\matR_{j_q,j}$
\end{enumerate}

\item end

\item Construct $\matX_{*} = \argmin_{\rank(\matX) \le k} \FNormS{ \matN \cdot \matX \cdot \matL -\matM  }.$
(Notice that $\matX_{*} = \argmin_{\rank(\matX) \le k} \FNormS{ \matT_{left} \left(
\matA \matR \matX \matS \matA  - \matA \right) \matT_{right}  }$.)

\item Compute the SVD of  $\matX_{*} = \matU_{ \matX_{*} } \matSig_{\matX_{*} } \matV_{\matX_{*} }\transp$ ($\matU_{ \matX_{*} } \in \R^{\xi \times k}$, $\matSig_{ \matX_{*} } \in \R^{k \times k}$
$\matV_{ \matX_{*} } \in \R^{\xi \times k}$); then, compute
$$\matT = \matD \matU_{ \matX_{*} }.$$

\item Compute an orthonormal basis $\matU \in \R^{m \times k}$ for $span(\matT)$.

\end{enumerate}

\end{small}

Theorem~\ref{thmonepass} later in this section analyzes the approximation error, the space complexity,
and the running time of the previous algorithm. First, we prove a few intermediate results.

\begin{lemma} \label{twosideaffine}
For all matrices $\matX\in\R^{\xi_2 \times \xi_1}$ and with probability at least $0.98$:
$$(1-\varepsilon)^2\FNormS{  \matA \matR \matX \matS \matA  - \matA   }
\leq \FNormS{\matT_{left} \left(\matA \matR \matX \matS \matA  - \matA \right) \matT_{right}}
\leq (1+\varepsilon)^2\FNormS{\matA \matR \matX \matS \matA  - \matA }$$
\end{lemma}

\begin{proof}
Notice that $rank(\matA\matR)\leq \xi_2=O(k\varepsilon^{-1})$. Then from Lemma~\ref{lem:affine}~($\matG:=\matA\matR$ and $\matH:=\matA$), with probability at least $0.99$ for all $\matY\in \R^{\xi_2 \times n}$:
$$(1-\varepsilon)\FNormS{\matA\matR\matY-\matA}\leq \FNormS{\matT_{left}(\matA\matR\matY-\matA)}\leq (1+\varepsilon)\FNormS{\matA\matR\matY-\matA}$$
Replacing $\matY=\matX\matS\matA$, for an arbitrary $\matX\in\R^{\xi_2\times\xi_1}$, we obtain that with probability at least $0.99$ and for all $\matX\in\R^{\xi_2\times\xi_1}$ simultaneously:
$$(1-\varepsilon)\FNormS{  \matA \matR \matX \matS \matA  - \matA   }
\leq \FNormS{\matT_{left} \left(\matA \matR \matX \matS \matA  - \matA \right) }
\leq (1+\varepsilon)\FNormS{\matA \matR \matX \matS \matA  - \matA }$$
Now notice that $rank(\matS\matA)\leq \xi_1=O(k\varepsilon^{-1})$. Then from Lemma~\ref{lem:affine}~($\matG:=\matA\transp\matS\transp$ and $\matH:=\matA\transp\matT_{left}\transp$), with probability at least $0.99$ for all $\matZ\in \R^{m \times \xi_2}$:
$$(1-\varepsilon)\FNormS{\matA\transp\matS\transp\matZ\transp-\matA\transp\matT_{left}\transp}
\leq \FNormS{\matT_{right}\transp(\matA\transp\matS\transp\matZ\transp-\matA\transp\matT_{left}\transp)}
\leq (1+\varepsilon)\FNormS{\matA\transp\matS\transp\matZ\transp-\matA\transp\matT_{left}\transp}$$
Replacing $\matZ=\matX\transp\matR\transp\matA\transp\matT_{left}\transp$ for an arbitrary $\matX \in \R^{\xi_2\times\xi_1}$, we obtain that with probability at least $0.99$ and for all $\matX$ simultaneously:
$$(1-\varepsilon)^2\FNormS{  (\matA \matR \matX \matS \matA  - \matA)\transp   }
\leq \FNormS{(\matT_{left} \left(\matA \matR \matX \matS \matA  - \matA \right) \matT_{right})\transp}
\leq (1+\varepsilon)^2\FNormS{(\matA \matR \matX \matS \matA  - \matA)\transp }$$
\end{proof}

\begin{lemma}\label{lembatch1}
Let $ \matX_{opt} =  \arg \min_{\rank(\matX) \le k} \FNormS{\matA\matR\matX\matS\matA-\matA}$ with $\matX_{opt} \in \R^{\xi_2 \times \xi_1}$. Then, with probability at least $0.98$
$$
\min_{\rank(\matX) \le k} \FNormS{ \matT_{left} \left(
   \matA \matR \matX \matS \matA - \matA \right) \matT_{right}  } \le
\FNormS{ \matA \matR \matX_{*} \matS \matA - \matA} \le
\frac{(1+\varepsilon)^2}{(1-\varepsilon)^2} \FNormS{  \matA \matR \matX_{opt} \matS \matA - \matA }
$$
\end{lemma}
\begin{proof}
From Lemma~\ref{twosideaffine}, we have that for all matrices $\matX \in \R^{\xi_2 \times \xi_1}$ and with probability at least $0.98$:
$$
(1-\varepsilon)^2 \cdot  \FNormS{ \matA\matR\matX\matS\matA - \matA}   \le
\FNormS{ \matT_{left} \left(  \matA\matR\matX\matS\matA - \matA \right) \matT_{right}
}
\le (1+\varepsilon)^2 \cdot  \FNormS{\matA\matR\matX\matS\matA - \matA}
$$
Applying this for $\matX : = \matX_{opt},$ we obtain:
$$
(1-\varepsilon)^2 \cdot  \FNormS{ \matA\matR \matX_{opt}\matS \matA - \matA}   \le
\FNormS{ \matT_{left} \left(  \matA\matR \matX_{opt}\matS \matA - \matA \right) \matT_{right}
}
\buildrel(e)\over\le (1+\varepsilon)^2 \cdot  \FNormS{ \matA\matR \matX_{opt}\matS \matA - \matA}
$$
Applying this for $\matX : = \matX_{*},$ we obtain:
$$
(1-\varepsilon)^2 \cdot  \FNormS{ \matA\matR\matX_{*}\matS\matA - \matA} \buildrel(f)\over\le
\FNormS{ \matT_{left} \left(   \matA\matR \matX_{*}\matS \matA - \matA \right) \matT_{right}
}
\le (1+\varepsilon)^2 \cdot  \FNormS{ \matA\matR \matX_{opt}\matS \matA - \matA}
$$
Combining $(f), (e)$ along with the optimality of $\matX_{*},$ formally using the relation:
$$\FNormS{ \matT_{left} \left(  \matA\matR \matX_{*}\matS \matA - \matA \right) \matT_{right}}
\le
\FNormS{ \matT_{left} \left(  \matA\matR \matX_{opt}\matS \matA - \matA \right) \matT_{right}},
$$
shows the claim.
\end{proof}
In words, the lemma indicates that in order to $(1+\varepsilon)$-approximate the optimization problem
$\min_{\rank(\matX) \le k} \FNormS{ \matA\matR\matX\matS\matA-\matA},$
it suffices to ``sketch'' the matrix $\matA\matR\matX\matS\matA-\matA$ from left and right with the matrices $\matT_{left}$ and $\matT_{right}$.
Recall that $\matX_{*} = \argmin_{\rank(\matX) \le k} \FNormS{ \matT_{left} \left(
\matA \matR \matX \matS \matA  - \matA \right) \matT_{right}  }$.

\begin{lemma}\label{lembatch2}
Let $ \matX_{opt} =  \arg \min_{\rank(\matX) \le k} \FNormS{\matA\matR\matX\matS\matA - \matA}$. Then, with probability at least $0.98$
$$
\FNormS{ \matA\matR \matX_{opt}\matS\matA - \matA}
\le (1+\varepsilon) \cdot \FNormS{\matA - \matA_k},
$$
\end{lemma}
\begin{proof}
Suppose the SVD of $\matA=\matU_{\matA}\matSig_{\matA}\matV_{\matA}\transp$ and $\matA_k=\matU_{\matA_k}\matSig_{\matA_k}\matV_{\matA_k}\transp$. Consider about the following regression problem:
$$\min_{rank(\matX)\le k} \FNormS{\matU_{\matA_k}\matX-\matA}$$
Since we can choose $\matX$ to be $\matU_{\matA_k}\transp\matA$, we have
$$\min_{rank(\matX)\le k} \FNormS{\matU_{\matA_k}\matX-\matA}\le\FNormS{\matA-\matA_k}$$
Let $\tilde{\matX}=\arg \min_{rank(\matX)\le k}\FNormS{\matS\matU_{\matA_k}\matX-\matS\matA}$. According to Lemma~\ref{lem:regre}, with probability $0.99$
$$\FNormS{\matU_{\matA_k}\tilde{\matX}-\matA}\leq (1+\varepsilon) \min_{rank(\matX)\le k} \FNormS{\matU_{\matA_k}\matX-\matA} \leq (1+\varepsilon)\FNormS{\matA-\matA_k}$$
Since $\tilde{\matX}$ minimizes $\FNormS{\matS\matU_{\matA_k}\matX-\matS\matA}$, the row space of $\tilde{\matX}$ is in the row space of $\matS\matA$ (Otherwise, we can project $\matS\matU_{\matA_k}\matX$ into the row space of $\matS\matA$ to get a better solution). We denote $\tilde{\matX}=\matW\matS\matA$ where $rank(\matW)\leq k$. Now, consider about another regression problem:
$$\min_{rank(\matX)\le k}\FNormS{(\matW\matS\matA)\transp\matX\transp-\matA\transp}$$
If we choose $\matX$ to be $\matU_{\matA_k}$, we have
$$\min_{rank(\matX)\le k}\FNormS{(\matW\matS\matA)\transp\matX\transp-\matA\transp}\leq \FNormS{\matU_{\matA_k}\tilde{\matX}-\matA}\leq (1+\varepsilon)\FNormS{\matA-\matA_k} $$
Let $\hat{\matX}=\arg \min_{rank(\matX)\le k}\FNormS{\matR\transp(\matW\matS\matA)\transp\matX\transp-\matR\transp\matA\transp}$. According to Lemma~\ref{lem:regre}, with probability $0.99$
$$\FNormS{\matR\transp(\matW\matS\matA)\transp\hat{\matX}\transp-\matR\transp\matA\transp}\leq (1+\varepsilon)\min_{rank(\matX)\le k}\FNormS{(\matW\matS\matA)\transp\matX\transp-\matA\transp}\leq (1+\varepsilon)^2\FNormS{\matA-\matA_k}$$
Since $\hat{\matX}$ minimizes $\FNormS{\matX\matW\matS\matA\matR-\matA\matR}$, the column space of $\hat{\matX}$ is in the column space of $\matA\matR$. We denote $\hat{\matX}=\matA\matR\matG$ where $rank(\matG)\le k$. Thus, we have
$$\FNormS{ \matA\matR \matX_{opt}\matS\matA - \matA}\le\FNormS{\matA\matR\matG\matW\matS\matA-\matA}\leq (1+\varepsilon)^2\FNormS{\matA-\matA_k}$$
We scale the $\varepsilon$ with a constant factor. The statement is shown by applying union bound.
\end{proof}

\begin{theorem}\label{thmonepass}
The matrix $\matU \in \R^{m \times k}$ with $k$ orthonormal columns satisfies w.p.
$0.96$:
$$
 \FNormS{\matA -  \matU \matU\transp \matA} \le
\left(1 + O\left(\varepsilon\right) \right) \cdot  \FNormS{\matA - \matA_k}.
$$
The space complexity of the algorithm is
$O\left(m k / \varepsilon + \poly(k\varepsilon^{-1}) \right)$ words, and the running time for each update operation is $O(poly(k\varepsilon^{-1}))$ and the total running time is of the order $ O\left( l\cdot \poly(k\varepsilon^{-1}) + m k^2 \varepsilon^{-1} \right) $ where $l$ is total number of updates.
\end{theorem}

\begin{proof}
Because $\matU\matU\transp$ projects the columns of $A$ into the column space of $\matT_{left}\matA\matR \matX_{*}$,
$$\FNormS{\matA-\matU\matU\transp\matA}\leq \FNormS{ \matT_{left} \left(
   \matA \matR \matX \matS \matA - \matA \right) \matT_{right}  }$$
According to Lemma~\ref{lembatch1} and Lemma~\ref{lembatch2},
$$
 \FNormS{ \matA-\matU \matU\transp \matA  } \le
\left(1 + O\left(\varepsilon\right) \right) \cdot  \FNormS{\matA - \matA_k}
$$

To see the space complexity of the algorithm,
observe that it suffices to maintain the matrices in the fourth step of the algorithm. Furthermore, observe that by the end of the stream:
$\matM=\matT_{left}\matA\matT_{right}\in \R^{\xi_3\times \xi_4},$
$\matL = \matS \matA \matT_{right}\in \R^{\xi_1\times \xi_4},$
$\matN=\matT_{left}\matA\matR\in \R^{\xi_3\times \xi_2}$
and $\matD=\matA\matR\in \R^{m \times \xi_2}$.
Those matrices form the so called ``sketch'' of the algorithm.

Since all of $\xi_1,\xi_2,\xi_3,\xi_4$ are $O(k\varepsilon^{-1})$, the running time for each update operation in the fourth step is $O(\poly(k\varepsilon^{-1}))$. We can do the sixth step by using the result of Section~\ref{sec:Uopt}. The computation takes $\poly(k\varepsilon^{-1})$ running time. In the seventh step, computing SVD needs $\poly(k\varepsilon^{-1})$, and computing $\matT$ needs $O(mk^2\varepsilon^{-1})$. We can use $O(mk^2)$ to compute $\matU$ in the final step, e.g. QR decomposition.
\end{proof}

\subsubsection{A variation which outputs $\matA^*_k$} \label{sec:outputa}
We just slightly modify the previous algorithm in Section~\ref{sec:outputu} to get the following one which can output a matrix $\matA^*_k\in \R^{m\times n}$ satisfying
$$\FNormS{\matA-\matA^*_k}\leq (1+\varepsilon)\cdot \FNormS{\matA-\matA_k}$$

\vspace{0.1in}
\begin{small}
{\bf Algorithm}

\begin{enumerate}

\item Construct random sign sketching matrices $\matS\in \mathbb{R}^{\xi_1 \times m}$ with $\xi_1=O(k\varepsilon^{-1})$ and $\matR\in \mathbb{R}^{n \times \xi_2}$ with $\xi_2=O(k\varepsilon^{-1})$~(see Lemma~\ref{lem:regre})

\item Construct affine embedding matrices $\matT_{left} \in \R^{\xi_3 \times m}$ and $\matT_{right} \in \R^{n \times \xi_4}$ with
$\xi_3 = O(k\varepsilon^{-3}),$ $\xi_4 = O(k/\varepsilon^{-3})$~(see Definition~\ref{def:srht}).

\item Initialize all-zeros matrices: $\matM \in \R^{\xi_3 \times \xi_4},$ $\matL \in \R^{\xi_1\times \xi_4},$ $\matN \in \R^{\xi_3 \times \xi_2},$ $\matD \in \R^{m \times \xi_2}$ and $\matC \in \R^{\xi_1 \times n}$. We will maintain $\matM,\matL,\matN,\matD,\matC$ such that $\matM=\matT_{left}\matA\matT_{right},$ $\matL = \matS \matA \matT_{right},$ $\matN=\matT_{left}\matA\matR,$ $\matD=\matA\matR$ and $\matC=\matS\matA$.

\item For $(i_q,j_q,x_q):= (i_1,j_1,x_1), ..., (i_l,j_l,x_l)$ (one pass over the stream of update operations)

\begin{enumerate}
\item For all $i=1,...,\xi_3,~j=1,...,\xi_4,$ let $\matM_{i,j}=\matM_{i,j}+(\matT_{left})_{i,i_q}\cdot x_q\cdot (\matT_{right})_{j_q,j}$.
\item For all $i=1,...,\xi_1,~j=1,...,\xi_4,$ let $\matL_{i,j}=\matL_{i,j}+\matS_{i,i_q} \cdot x_q \cdot (\matT_{right})_{j_q,j}$.
\item For all $i=1,...,\xi_3,~j=1,...,\xi_2,$ let $\matN_{i,j}=\matN_{i,j}+(\matT_{left})_{i,i_q}\cdot x_q \cdot \matR_{j_q,j}$.
\item For all $j=1,...,\xi_2,$ let $\matD_{i_q,j}=\matM_{i_q,j}+x_q\cdot\matR_{j_q,j}$
\item For all $i=1,...,\xi_1,$ let $\matC_{i,j_q}=\matC_{i,j_q}+\matS_{i,i_q}\cdot x_q$
\end{enumerate}

\item end

\item Construct $\matX_{*} = \argmin_{\rank(\matX) \le k} \FNormS{ \matN \cdot \matX \cdot \matL -\matM  }.$
(Notice that $\matX_{*} = \argmin_{\rank(\matX) \le k} \FNormS{ \matT_{left} \left(
\matA \matR \matX \matS \matA  - \matA \right) \matT_{right}  }$.)

\item Compute the SVD of  $\matX_{*} = \matU_{ \matX_{*} } \matSig_{\matX_{*} } \matV_{\matX_{*} }\transp$ ($\matU_{ \matX_{*} } \in \R^{\xi \times k}$, $\matSig_{ \matX_{*} } \in \R^{k \times k}$
$\matV_{ \matX_{*} } \in \R^{\xi \times k}$); then, compute
$$\matT = \matD \matU_{ \matX_{*} }$$
$$\matK = \matV_{\matX_{*}}\transp\matC$$

\item Output $\matA^*_k=\matT\matSig_{\matX_{*}}\matK$

\end{enumerate}

\end{small}

\begin{theorem}\label{thmonepassvar}
With probability at least $0.96$:
$$
 \FNormS{\matA -  \matA^*_k} \le
\left(1 + O\left(\varepsilon\right) \right) \cdot  \FNormS{\matA - \matA_k}.
$$
The space complexity of the algorithm is
$O\left((m+n) k / \varepsilon + \poly(k\varepsilon^{-1}) \right)$ words, and the running time for each update operation is $O(poly(k\varepsilon^{-1}))$ and the total running time is of the order
$$ O\left( l\cdot \poly(k\varepsilon^{-1}) + (m+n) k^2 \varepsilon^{-1} + mnk\right) $$
where $l$ is total number of updates.
\end{theorem}

\begin{proof}
Notice that $\matA^*_k=\matA\matR\matX_{*}\matS\matA$. According to Lemma~\ref{lembatch1} and Lemma~\ref{lembatch2}, we have
$$\FNormS{\matA -  \matA^*_k} \le
\left(1 + O\left(\varepsilon\right) \right) \cdot  \FNormS{\matA - \matA_k}.$$
The total space of ``sketch'' matrices $\matM,\matL,\matN,\matD$ is the same as the algorithm in Section~\ref{sec:outputu}. The maintenance of $\matC$ needs additional $O(n k \varepsilon^{-1})$ space.

The running time is almost the same as before. Computation of $\matK = \matV_{\matX_{*}}\transp\matC$ needs additional $O(nk^2\varepsilon^{-1})$, and computation of $\matA^*_k=\matT\matSig_{\matX_{*}}\matK$ needs additional $O(mnk)$.
\end{proof}

\subsection{Two-pass streaming PCA with ``real numbers'' space complexity}\label{sec:twopass}
A simple modification of the batch algorithm in Section~\ref{sec:algbatch} leads to a two-pass streaming algorithm with
$O(m k) + \poly(k/\varepsilon)$ ``real numbers'' space complexity:

\vspace{0.2in}
\begin{small}
{\bf Algorithm}

\begin{enumerate}

\item Construct two dense Johnson-Lindenstrauss matrices $\matS\in\mathbb{R}^{\xi_1\times m},\matT\in\mathbb{R}^{n \times \xi_2}$ with $\xi_1=O(k\varepsilon^{-2}),\xi_2=O(k\varepsilon^{-2})$~(see Lemma~\ref{lem:djlm}).

\item Initialize all-zero matrices: $\tilde{\matA}\in\R^{\xi_1\times\xi_2},$ $\matX\in\R^{m\times k}$.

\item For $(i_q,j_q,x_q):= (i_1,j_1,x_1), ..., (i_l,j_l,x_l)$ (\textbf{first pass over the stream})

\item For all $i=1,...,\xi_1,~j=1,...,\xi_2,$ let $\tilde{\matA}_{i,j}=\tilde{\matA}_{i,j}+\matS_{i,i_q}\cdot x_q\cdot \matT_{j_q,j}$.

\item end

\item Compute the SVD of $\tilde{\matA}_k=\matU_{\tilde{\matA}_k}\matSig_{\tilde{\matA}_k}\matV_{\tilde{\matA}_k}\transp$
~($\matU_{ \tilde{\matA}_k } \in \R^{\xi_1 \times k}$, $\matSig_{ \tilde{\matA}_k } \in \R^{k \times k}$, $\matV_{ \tilde{\matA}_k } \in \R^{\xi_2 \times k})$.

\item For $(i_q,j_q,x_q):= (i_1,j_1,x_1), ..., (i_l,j_l,x_l)$ (\textbf{second pass over the stream})

\item For all $i=1,...,\xi_2,~j=1,...,k,$ let $\matX_{i_q,j}=\matX_{i_q,j}+ x_q\cdot \matT_{j_q,i} \cdot (\matV_{\tilde{\matA}_k})_{i,j}$.

\item end

\item Compute an orthonormal basis $\matU \in \mathbb{R}^{m\times k}$ for $span(\matX)$ (notice that $rank(\matX) \le k$).

\end{enumerate}

\end{small}

The theorem below analyzes the approximation error, the space complexity,
and the running time of the previous algorithm. Notice that the space complexity
of this algorithm is only given in terms of ``real numbers'' - we resolve this issue in the following section.

\begin{theorem}\label{thmtwopass}
The matrix $\matU \in \R^{m \times k}$ with $k$ orthonormal columns satisfies w.p.
$0.98$:
$$
 \FNormS{\matA -  \matU \matU\transp \matA} \le
\left(1 + O\left(\varepsilon\right) \right) \cdot  \FNormS{\matA - \matA_k}.
$$
The space complexity of the algorithm is
$$
O\left(m k  + \poly(k \varepsilon^{-1}) \right)
$$
``real numbers'',  the running time of each update operation is $O(poly(k\varepsilon^{-1}))$, and the total running time is of the order
$$
O\left( l\cdot\poly(k\varepsilon^{-1})  + mk^2\right)
$$
where $l$ is the total number of update operations.
\end{theorem}
\begin{proof}

The matrix $\matU$ - up to the randomness in the algorithms - is exactly the same matrix as in the algorithm in Section~\ref{sec:algbatch},
hence Theorem~\ref{thmbatch} shows the claim.

To see the space complexity of the algorithm,
observe that it suffices to maintain the matrices in the fourth and eighth
steps of the algorithm. Furthermore, observe that by the end of the stream:
$$\tilde{\matA} =  \matS\matA\matT,$$
and
$$\matX = \matA\matT\matV_{\tilde{\matA}_k}$$

Since $\xi_1,\xi_2$ are $poly(k\varepsilon^{-1})$, the running time for each update operation is $poly(k\varepsilon^{-1})$. Computing SVD in the sixth step needs $\poly(k\varepsilon^{-1})$. Computing $\matU$ in the final step needs $O(mk^2)$.
\end{proof}

\subsection{Two-pass streaming PCA with bit space complexity}\label{sec:twopass2}
The previous algorithm could suffer from large space complexity in case we need
a lot of machine words to write down the entries of $\matV_{\tilde{\matA}_k}$. To fix this issue
we use the same idea as in the case of the distributed PCA algorithm in Section~\ref{sec:precisionDistributed}.
This leads to a two-pass streaming algorithm for PCA. Again, as in the distributed case, we need to test if the rank of $\matA$ is less than $2k,$ and then we use one approach if $\rank(\matA) < 2k$ and another approach if $\rank(\matA) \ge 2k$. Both of these approaches can be implemented with two passes. In the overall streaming PCA algorithm that we would like to design,
we can not wait for the algorithm that tests the rank to finish and then start running one of the two PCA protocols, because this will lead to a three-pass algorithm~(one pass to test the rank and two passes for the actual PCA protocol). To keep the number of passes to two, we just start running the PCA protocols in parallel with the protocol that tests the rank of the input matrix. In the end of the first pass over the stream of update operations, we already know which protocol to follow, and we just do this, disregarding the other protocol.

We already discussed the streaming version of the algorithm that tests the rank of the matrix in Lemma~\ref{lem:rankTest}. Below, we describe separately the two streaming PCA protocols.

\subsubsection{Streaming PCA protocol when $\rank(\matA) \le2k$}
The idea here is to implement a streaming version of the distributed PCA protocol in
Lemma~\ref{lem:lowRankProtocol}.

\vspace{0.2in}
\begin{small}
{\bf Algorithm}

\begin{enumerate}

\item Construct an $n \times 2k$ matrix $\matH''$ as in Lemma~\ref{lem:lowRankProtocol}.

\item Construct affine embedding matrices
$\matT_{left} \in \R^{\xi_1 \times m}$ and $\matT_{right} \in \R^{n \times \xi_2}$ with
$\xi_1 = O(k/\varepsilon^2),$ $\xi_2 = O(k/\varepsilon^2)$~(see Definition~\ref{def:srht}).
\item Initialize all-zeros matrices:  $\matC \in \R^{m \times 2 k}, \matM \in \R^{\xi_1 \times \xi_2}$, $\matL  \in \R^{\xi  \times \xi_2}$, $\matN \in \R^{\xi_1 \times 2k}$.

\item $(i_q,j_q,x_q):= (i_1,j_1,x_1), ..., (i_l,j_l,x_l)$ (\textbf{first pass over the stream})

\begin{enumerate}
\item  For all $j=1,...,2k,$ let $\matC_{i_q,j} = \matC_{i_q,j} + x_q\cdot \matH''_{j_q,j}$
\end{enumerate}

\item end

\item For $(i_q,j_q,x_q):= (i_1,j_1,x_1), ..., (i_l,j_l,x_l)$ (\textbf{second pass over the stream})

\begin{enumerate}
\item For all $i=1,...,\xi_1,~j=1,...,\xi_2,$ let $\matM_{i,j}=\matM_{i,j}+(\matT_{left})_{i,i_q}\cdot x_q\cdot (\matT_{right})_{j_q,j}$.
\item For all $i=1,...,2k,~j=1,...,\xi_2,$ let $\matL_{i,j}=\matL_{i,j}+{\matC\transp}_{i,i_q} \cdot x_q \cdot (\matT_{right})_{j_q,j}$.
\item For all $i=1,...,\xi_1,~j=1,...,2k,$ let $\matN_{i,j}=\matN_{i,j}+(\matT_{left})_{i,i_q}\cdot x_q \cdot \matH''_{j_q,j}$.
\end{enumerate}

\item end

\item Construct
$$
\matX_{*} = \argmin_{\rank(\matX) \le k} \FNormS{ \matN \cdot \matX \cdot \matL -\matM  } := \argmin_{\rank(\matX)\le k} \|\matT_{left} \matA \matH'' \matX (\matH'')\transp \matA\transp \matA \matT_{right} - \matT_{left} \matA \matT_{right}\|_{\mathrm{F}}^2.
$$

\item Compute an orthonormal basis $\matU \in \R^{m \times k}$ for $span(\matC \matX_{*} )$.

\end{enumerate}

\end{small}

\subsubsection{Streaming PCA protocol when $\rank(\matA) > 2k$}
The idea here is to implement a streaming version of step $3$ of the algorithm in Section~\ref{sec:bit4}.

\vspace{0.2in}
\begin{small}
{\bf Algorithm}

\begin{enumerate}

\item Construct two dense Johnson-Lindenstrauss matrices $\matS\in\mathbb{R}^{\xi_1\times m},\matT\in\mathbb{R}^{n \times \xi_2}$ with $\xi_1=O(k\varepsilon^{-2}),\xi_2=O(k\varepsilon^{-2})$~(see Lemma~\ref{lem:djlm}).

\item Initialize all-zero matrices: $\tilde{\matB}\in\R^{\xi_1\times\xi_2},$ $\hat{\matX}\in\R^{m\times k}$.

\item For $p:=1,...,\xi_1,~q:=1,...,\xi_2,~i:=1,...,m,~j:=1,...,n$

\item Let $r_{i,j}=-1/n^D$ with probability $1/2$ and $r_{i,j}=1/n^D$ with probability $1/2$.

\item $\tilde{\matB}_{p,q}=\tilde{\matB}_{p,q}+\matS_{p,i}\cdot r_{i,j} \cdot \matT_{j,q}$

\item end

\item For $(i_q,j_q,x_q):= (i_1,j_1,x_1), ..., (i_l,j_l,x_l)$ (\textbf{first pass over the stream})

\item For all $i=1,...,\xi_1,~j=1,...,\xi_2,$ let $\tilde{\matB}_{i,j}=\tilde{\matB}_{i,j}+\matS_{i,i_q}\cdot x_q\cdot \matT_{j_q,j}$.

\item end

\item Compute the SVD of $\tilde{\matB}_k=\matU_{\tilde{\matB}_k}\matSig_{\tilde{\matB}_k}\matV_{\tilde{\matB}_k}\transp$
~($\matU_{ \tilde{\matB}_k } \in \R^{\xi_1 \times k}$, $\matSig_{ \tilde{\matB}_k } \in \R^{k \times k}$, $\matV_{ \tilde{\matB}_k } \in \R^{\xi_2 \times k})$. Then, we round each of the entries of $\matV_{\tilde{\matB}_k}$ to the nearest integer multiple of $1/n^{\gamma}$ for a sufficient large constant $\gamma>0$. Let the matrix after the rounding be ${\hat{\matV}}_{\tilde{\matB}_k}$.

\item For $(i_q,j_q,x_q):= (i_1,j_1,x_1), ..., (i_l,j_l,x_l)$ (\textbf{second pass over the stream})

\item For all $i=1,...,\xi_2,~j=1,...,k,$ let $\hat{\matX}_{i_q,j}=\hat{\matX}_{i_q,j}+ x_q\cdot \matT_{j_q,i} \cdot ({\hat{\matV}}_{\tilde{\matB}_k})_{i,j}$.

\item end

\item Compute an orthonormal basis $\matU \in \mathbb{R}^{m\times k}$ for $span(\hat{\matX})$ (notice that $rank(\hat{\matX}) \le k$).

\end{enumerate}

\end{small}

\subsubsection{Main result}
The theorem below analyzes the approximation error, the space complexity,
and the running time of the previous algorithm. Notice that the space complexity
of this algorithm is given in terms of  machine words.

\begin{theorem}\label{thmtwopass2}
The matrix $\matU \in \R^{m \times k}$ with $k$ orthonormal columns satisfies with arbitrarily large constant probability:
$$
 \FNormS{\matA -  \matU \matU\transp \matA} \le
\left(1 + \varepsilon \right) \cdot  \FNormS{\matA - \matA_k}.
$$
The space complexity of the algorithm is
$$
O\left(m k  + \poly(k \varepsilon^{-1}) \right)
$$
machine words,  the running time of each update operation is $O(poly(k\varepsilon^{-1}))$, and the total running time is of the order
$$
O\left( l\cdot\poly(k\varepsilon^{-1})  + mk^2\right)
$$
where $l$ is the total number of update operations.
\end{theorem}
\begin{proof}

The matrix $\matU$ - up to the randomness in the algorithms - is exactly the same matrix as in the algorithm in Theorem~\ref{thmd2}, hence the approximation bound in that theorem shows the claim.

To see the space complexity of the algorithm,
observe that it suffices to maintain all the ``sketch'' matrices in the both two protocols above.

Since $\xi_1,$ $\xi_2$ in both protocols are $\poly(k\varepsilon^{-1})$, the running time for each update operation is $O(\poly(k\varepsilon^{-1}))$. Additional $O(mk^2)$ running time is caused by computing $\matU$ in the final step.
\end{proof}

\section{Distributed PCA for sparse matrices in column-partition model}\label{sec:algorithm}
Recall that in the problem of Definition~\ref{def:model} we are given 1) an $m \times n$ matrix $\matA$ partitioned column-wise as
$
\matA =
\begin{pmatrix}
\matA_1 & \matA_2 & \dots & \matA_s
 \end{pmatrix},
$
with $\matA_i \in \R^{m \times w_i}$ ($\sum w_i =n$);
2) a rank parameter $k < \rank(\matA)$;
3) an accuracy parameter $\varepsilon > 0$.
We would like to design an algorithm that finds an $m \times k$ matrix $\matU$ with $\matU\transp\matU=\matI_k$ and, upon termination, leaves this matrix $\matU$ in each machine of the network.

The high level idea of our algorithm is to find, in a distributed way, a small set - $O(k \varepsilon^{-1})$ - of columns from $\matA$ and then find $\matU$ in the span of those columns. To choose those $O(k \varepsilon^{-1})$ columns of $\matA$ in a distributed way we implement the following three-stage sampling procedure:
\begin{enumerate}
\item Local sampling: Each machine samples $O(k)$ columns using the sampling algorithm from~\cite{cohen2014dimensionality}.
%
\item Global sampling: The server collects the columns from each machine and down-samples them to $O(k)$ columns using the deterministic algorithm from~\cite{BDM11a}.
\item Adaptive sampling: the server sends back to each machine those $O(k)$ columns; then,
an extra of $O(k \varepsilon^{-1})$ columns are selected randomly from the entire matrix $\matA$
using~\cite{DRVW06}.
\end{enumerate}
We argue that if $\tilde{\matC}$ contains the columns selected with this three-stage approach, then, w.p. $0.99,$
$$
\FNormS{\matA -  \tilde\matC  \pinv{\tilde\matC}\matA} \le
\FNormS{\matA-\Pi_{\tilde\matC,k}^{\mathrm{F}}(\matA)}\le
(1 + O(\varepsilon)) \cdot \FNormS{\matA-\matA_k}.
$$
Though we could have used $\Pi_{\tilde\matC,k}^{\mathrm{F}}(\matA)$ to be the rank $k$ matrix that approximates $\matA,$ we are not familiar with any communication-efficient computation of $\Pi_{\tilde\matC,k}^{\mathrm{F}}(\matA)$. To address this issue, using an idea from~\cite{KVW14}, we compute $\matU \in span(\tilde\matC)$ such that
$\FNormS{\matA - \matU\matU\transp\matA} \le (1+O(\varepsilon)) \FNormS{\matA - \Pi_{\tilde\matC,k}^{\mathrm{F}}(\matA)};$
this $\matU$ can be calculated with small communication cost and it is sufficient for our purposes.

Before presenting the new algorithm in detail, we discuss results from previous literature that we employ in the analysis.

\subsection{Background material}

\subsubsection{Column sampling algorithms}\label{sec:sampling}
Our distributed PCA algorithm in Section~\ref{sec:algorithm} samples columns from the input matrix in three stages:
\begin{enumerate}
\item Local sampling: $O(k)$ columns are selected locally in each machine.
\item Global sampling: $O(k)$ columns are selected in the server.
\item Adaptive sampling: an extra $O(k \varepsilon^{-1})$ columns are selected from the entire matrix $\matA$.
\end{enumerate}

In the first stage, we use a sampling algorithm mentioned in~\cite{cohen2014dimensionality}.
Actually, the sampling algorithm is the same as the Batson, Spielman, and Strivastava (BSS) spectral sparsification algorithm~\cite{BSS09}. But Cohen et al. showed that a small set of columns sampled by BSS sampling algorithm can provide a projection-cost preserving sketch.
In the second stage, we use a deterministic algorithm developed in~\cite{BDM11a},
which itself extends the Batson, Spielman, and Strivastava (BSS) spectral sparsification algorithm~\cite{BSS09}.
For the actual algorithm in the first stage we defer the reader to Theorem 15 in~\cite{cohen2014dimensionality}. Lemma~\ref{lem:bssproject} states the result.
And for the actual algorithm in the second stage we defer the reader to Lemma 3.6 in~\cite{BDM11a}.  Lemma~\ref{lem:dualset} and
Lemma~\ref{thm:optimalFdet} below present the relevant results.
In the third sampling stage, we use an adaptive sampling algorithm from~\cite{DRVW06}.
%

\begin{lemma}[BSS Sampling for Projection-Cost Preserving Sketch. Theorem 15 in~\cite{cohen2014dimensionality}]
\label{lem:bssproject}
Let $\matA \in \R^{m\times w}$ be an arbitrary matrix. For any integer $0<k<m$, real numbers $0<\eps<1,0<\delta$, there is a randomized algorithm that runs in $O(\nnz(A)\log(m/\delta)+m\cdot \poly(k,\eps,\log(1/\delta)))$ time, and constructs a ``sampling/rescaling'' $w\times O(k/\eps^2)$ matrix $\matS$ with probability at least $1-\delta$ such that $\matS$ is an $(\varepsilon,k)$-projection-cost preserving sketching matrix of $\matA$. We denote this sampling procedure as
$
\matS = BssSampling\1(\matA, k, \eps, \delta).
$
\end{lemma}
In words, there exists a randomized algorithm that runs in near input-sparsity running time can construct a projection-cost preserving sketch of a given matrix $\matA$ by sampling and rescaling a small subset of columns of $\matA$.

\begin{lemma}[Dual Set Spectral-Frobenius Sparsification. Lemma 3.6 in~\cite{BDM11a}]
\label{lem:dualset}
Let $\matV \in \R^{w \times k}$ be a matrix with $w > k$ and $\matV\transp \matV = \matI_{k}$.
Let $\matE \in \R^{m \times w}$ be an arbitrary matrix.  Then,
given an integer \math{\ell} such that \math{k < \ell \le w}, there exists a deterministic algorithm that runs in
$O\left(\ell w k^2 + m w \right)$ time, and constructs a ``sampling/rescaling'' $w \times \ell$ matrix $\matS$ such
that
\eqan{
\sigma_{k}^2\left(\matV\transp \matS \right)
\ge
\left(1 - \sqrt{{k}/{\ell}}\right)^2,
\qquad
\FNormS{\matE \matS}
\le \FNormS{\matE}.
}
Specifically,  $rank(\matV\transp \matS) = k$.
We denote this sampling procedure as
$
\matS = BssSampling\2(\matV, \matE, \ell).
$
\end{lemma}
In words, given $\matV$ and $\matE,$ there exists an algorithm to select (and rescale) a small number of columns from $\matE$ and rows from $\matV$ such that:
1) the Frobenius norm squared of the sampled $\matE$ is less than the Frobenius norm squared of $\matE$; 2) the sampled $\matV$ has rank equal to the rank of $\matV$; and 3) the smallest singular value squared of the sampled $\matV$ is bounded from below by $\left(1 - \sqrt{{k}/{\ell}}\right)^2$.

\begin{lemma}[Constant factor column-based matrix reconstruction; Theorem 5 in~\cite{BDM11a}]
\label{thm:optimalFdet}
Given matrix \math{\matG\in\R^{m \times \alpha}} of rank $\rho$ and a target rank $k$~\footnote{The original Theorem 5 in~\cite{BDM11a} has the assumption that $k < \rho,$ but this assumption can be dropped having the result unchanged. The only reason the assumption $k < \rho$ exists is because otherwise column subset selection is trivial.},
there exists a deterministic algorithm that runs in $O\left(\alpha m \min\{ \alpha, m\}+ \alpha c k^2 \right)$  time and selects $c > k$
columns of \math{\matG} to form a matrix
$\matC\in\R^{m \times c}$
with
$$
\FNormS{\matG - \matC \pinv{\matC}\matG } \leq \left( 1 + \left( 1 - \sqrt{k/c}\right)^{-2} \right) \cdot \sum_{i=k+1}^{\rank(\matG)}\sigma_i^2(\matG).
$$
The algorithm in this theorem finds $\matC$ as $\matC =  \matG \matS,$ where
$\matS = BssSampling\2(\matV,  \matG-\matG\matV \matV\transp, c)$ and $\matV \in \R^{\alpha \times k}$ contains the top $k$ right singular vectors of $\matG$.
We denote this sampling procedure as $\matC = DeterministicCssFrobenius(\matG, k, c).$
\end{lemma}
In words, there exists a deterministic algorithm, running in polynomial time, to select any number of $c > k$ columns from the given matrix $\matG,$ such that the residual error, in Frobenius norm, from projecting $\matG$ onto the span of the sampled columns is bounded from above with respect to the residual error of the best rank $k$ matrix from the SVD of $\matG$. Notice that
$$\sum_{i=k+1}^{\rank(\matG)}\sigma_i^2(\matG) = \FNormS{\matG - \matG_k},$$ where
$\matG_k \in \R^{m \times \alpha}$ has rank at most $k$ and is computed via the SVD of $\matG$.

Before presenting the adaptive sampling theorem from~\cite{DRVW06} we introduce some useful notation used in the lemma. Let $\matA \in \mathbb{R}^{m \times n}$, $k < n$, and $\matV \in \mathbb{R}^{m \times c}$ with $k < c< n $.
$\Pi_{\matV,k}^\mathrm{F}(\matA) \in \mathbb{R}^{m \times n}$ is the best rank \math{k} approximation to \math{\matA} - wrt the Frobenius norm - in the column span of \math{\matV}.
Equivalently, $\Pi_{\matV, k}^\mathrm{F}(\matA) = \matV \matX_{opt},$ where
$$
\matX_{opt} = \argmin_{\matX \in {\R}^{c \times n}, \rank(\matX)\leq k}\FNormS{\matA-
\matV \matX}.
$$

\begin{lemma}[Adaptive sampling; Theorem 2.1 of~\cite{DRVW06}]
\label{thm:adaptive}
Given $\matA \in \R^{m \times n}$ and $\matV \in \R^{m \times c_1}$ (with $c_1 \le n, m$),
define the residual matrix
$\matPsi = \matA - \matV \matV^{\dagger} \matA \in \R^{m \times n}.$
For $j=1,\ldots,n$,
let $p_j$ be a probability distribution such that
$p_j \ge \beta {\TNormS{\matPsi^{(j)}}}/{\FNormS{\matPsi}},$ for some $1> \beta>0,$
where $\matPsi^{(j)}$ is the $j$-th column of the matrix $\matPsi$. Sample
$c_2$ columns from $\matA$ in \math{c_2} i.i.d. trials, where in each trial the $j$-th column is chosen with probability $p_j$.
Let $\matC_2 \in \R^{m \times c_2}$ contain the $c_2$ sampled columns and let $\matC = [\matV\ \ \matC_2] \in \R^{m \times (c_1+c_2)}$
contain the columns of $\matV$ and $\matC_2$. 
Then, for any integer $k > 0$,
$$
\Expect{ \FNormS{ \matA - \matC \pinv{\matC} \matA } } \le
\Expect{ \FNormS{ \matA - \Pi_{\matC,k}^{\mathrm{F}}(\matA) } } \le  \sum_{i=k+1}^{\rank(\matA)} \sigma_i^2(\matA) + \frac{k}{\beta \cdot c_2} \cdot \FNormS{ \matA - \matV \matV^{\dagger} \matA}.
$$
%
Given $\matA$ and $\matC,$ this method requires $O(c_1 m n)$ time to compute $\matPsi,$
another $O(m n)$ time to compute the probabilities $p_j$'s and another $O(n + c_2)$ time
for the sampling step - using the method in~\cite{Vos91}. In total, the method requires
$O(c_1 m n + c_2)$ time to compute $\matC_2$.
We denote this sampling procedure as
$\matC_2 = AdaptiveCols(\matA, \matV, c_2, \beta).$
\end{lemma}
In words, given the matrix $\matA$ and the subspace $\matV,$ there exists a randomized algorithm to sample an additional of $c_2$ columns from $\matA,$ based on probabilities from the residual error matrix $\matA - \matV\pinv{\matV}\matA,$ such that residual error, in Frobenius norm, from projecting $\matA$ onto the span of the union of the columns of $\matV$ and the sampled columns is bounded from above with respect to the best rank $k$ approximation to $\matA,$ the residual $\matA - \matV\pinv{\matV}\matA,$ and the number of sampled columns $c_2$.

\subsubsection{Distributed adaptive sampling}
In our distributed PCA algorithm below, we also need to use a distributed version of the previous adaptive sampling procedure.
We provide some preliminary results for that task in this section.
\begin{lemma}\label{lem:adaptiveSampling1}
Suppose that the columns of an $m \times n$ matrix $\matA$ are partitioned arbitrarily across the machines
into matrices $\matA_1, \ldots, \matA_s$. Let $\matC$ be an arbitrary $m \times r$ matrix.
Consider the distribution $p$ on $n$ columns in which
$$p_j = \frac{\|\a_j-\matC\matC^{\dagger}\a_j\|_{\mathrm{F}}^2}{\|\matA-\matC\matC^{\dagger}\matA\|_{\mathrm{F}}^2},$$
where $\a_j$ is the $j$-th column of $\matA$~($j=1:n$ here).

For each $i \in [s]$, let some value $\beta_i$ satisfies
$$\|\matA_i - \matC \matC^{\dagger} \matA_i\|_{\mathrm{F}}^2 \leq \beta^i \leq 2\|\matA_i - \matC \matC^{\dagger} \matA_i\|_{\mathrm{F}}^2.$$
For each $j \in [n]$, if column $\a_j$ is held on the $i$-th server (denoted $\a_j^i$), then let
$$q_j = \frac{\beta^i}{\sum_{i' = 1}^s \beta_{i'}} \cdot
\frac{\|\a_j^i - \matC \matC^{\dagger}\a_j^i\|_2^2}{\|\matA_i-\matC\matC^{\dagger}\matA_i\|_{\mathrm{F}}^2}.$$
Then for each $j \in [n]$, $$p_j/2 \leq q_j \leq 2p_j.$$
\end{lemma}
\begin{proof}
By definition of the $\beta^i$, we have that
$$
\frac{\|\matA_i - \matC \matC^{\dagger} \matA_i\|_{\mathrm{F}}^2}{2\|\matA-\matC\matC^{\dagger}\matA\|_{\mathrm{F}}^2}
\leq
\frac{\beta^i}{\sum_{i' = 1}^s \beta_{i'}}
\leq
\frac{2\|\matA_i - \matC \matC^{\dagger} \matA_i\|_{\mathrm{F}}^2}{\|\matA-\matC\matC^{\dagger}\matA\|_{\mathrm{F}}^2}.$$
Hence,
$$p_j/2 \leq q_j \leq 2p_j.$$
\end{proof}


\begin{lemma}\label{lem:adaptiveSampling2} 
Suppose the coordinator has an $m \times r$ matrix $\matC$ of columns of an $m \times n$ matrix $\matA$,
where $r = O(k)$. Suppose the entries of $\matA$ are integers bounded by $\poly(mns/\varepsilon)$ in magnitude, and
let the columns of $\matA$ be partitioned arbitrarily across the servers
into matrices $\matA_1, \ldots, \matA_s$.

There is a protocol using $O(skm)$ machine words of communication for the coordinator to learn values $\beta^i$ so that for
all $i \in [s]$,
$$\|\matA_i - \matC \matC^{\dagger} \matA_i\|_{\mathrm{F}}^2 \leq \beta^i \leq 2\|\matA_i - \matC \matC^{\dagger} \matA_i\|_{\mathrm{F}}^2.$$
\end{lemma}
\begin{proof}
The coordinator first sends $\matC$ to all machines.
The $i$-th server locally computes $\|\matA_i - \matC \matC^{\dagger}\matA_i\|_{\mathrm{F}}^2$.
If this number is $0$, he/she sends $0$ to the coordinator, and in this case $\beta^i$ satisfies the requirement in the statement of the lemma.

Otherwise, consider the matrix $\matB_i$ which is formed by concatenating the columns of $\matA_i$ with those of $\matC$.
Then
\begin{eqnarray}\label{eqn:trick}
\|\matB_i - \matC \matC^{\dagger} \matB_i\|_{\mathrm{F}}^2 = \|\matA_i - \matC \matC^{\dagger} \matA_i\|_{\mathrm{F}}^2 > 0,
\end{eqnarray}
since the columns of $\matB_i$ in $\matC$ contribute a cost of $0$. However, since $\matB_i$ contains the columns of $\matC$
and its cost is non-zero, it implies the rank of $\matB_i$ is at least $r+1$. By Corollary \ref{cor:plusOne}, this implies
$$\|\matB_i - \matC \matC^{\dagger}\matB_i\|_{\mathrm{F}}^2 > (mns/\varepsilon)^{-O(k)},$$
which by (\ref{eqn:trick}) gives the same
lower bound on $\|\matA_i - \matC \matC^{\dagger} \matA_i\|_{\mathrm{F}}^2$.
Note also that
$$\|\matA_i - \matC \matC^{\dagger}\matA_i\|_{\mathrm{F}}^2
\leq \poly(mns/\varepsilon),$$ since $\|\matA\|_{\mathrm{F}}^2 \leq \poly(mns/\varepsilon).$ This implies if the $i$-th machine sends
$\beta^i$ to the coordinator, where $\beta^i$ is the nearest power of $2$ to $\|\matA^i - \matC \matC^{\dagger} \matA^i\|_{\mathrm{F}}^2$,
there will only be $O(k \log (mns/\varepsilon))$ possible values of $\beta^i$, and hence each can be specified using
$O(\log k + \log \log(mns/\varepsilon))$ bits, i.e., a single machine word. This completes the proof.
\end{proof}

\subsubsection{Low-rank matrix approximations within a subspace}\label{sec:bestF}
The final stage of our distributed PCA algorithm below finds $\matU \in span(\tilde{\matC})$ such that the error of the residual matrix
$\matA - \matU \matU\transp\matA$ is ``small''. To implement this step, we employ an algorithm developed in~\cite{KVW14}.
\begin{lemma}\label{lem:KVW}
Let $\matA \in \R^{m \times n}$ be the input matrix and $\matV \in \R^{m \times c}$ be the input subspace.
We further assume that for some rank parameter $k < c$ and accuracy parameter $0 < \varepsilon <1:$
$$ \FNormS{\matA-\Pi_{\matV,k}^{\mathrm{F}}(\matA)} \le (1+O(\epsilon)) \FNormS{\matA - \matA_k}.$$
Let $\matV = \matY \matPsi$ be a $qr$ decomposition of $\matV$
with $\matY \in \R^{m \times c}$ and $\matPsi \in \R^{c \times c}$.
Let $\matXi = \matY\transp \matA \matW\transp \in \R^{c \times \xi},$
where $\matW\transp \in \R^{n \times \xi}$ with $\xi = O(c/\varepsilon^2),$
each element of which is chosen i.i.d. to be $\{+1/\sqrt{n},-1/\sqrt{n}\}$ with probability $1/2$.
Let $\matDelta \in \R^{c \times k}$ contain the top $k$ left singular vectors of $\matXi$.
Then, with probability at least $1-e^{-c}$,
$$
\FNormS{ \matA - \matY \matDelta \matDelta\transp \matY\transp \matA } \le (1+\varepsilon) \FNormS{\matA - \matA_k}.
$$
$\matY$, and $\matDelta$ can be computed in $O(m n \xi )$ time.
We denote this procedure as
$$
[\matY, \matDelta] = ApproxSubspaceSVD(\matA, \matV, k, \varepsilon).
$$
\end{lemma}
\begin{proof}
This result was essentially proven inside the proof of Theorem 1.1 in~\cite{KVW13}. Specifically,
the error bound proven in~\cite{KVW13} is for the transpose of $\matA$ (also $\matY,\matDelta$ are
denoted with $U, V$ in~\cite{KVW13}).
As for the running time, given $\matA, \matV,$ and $k$,
one can compute
(i) $\matY$ in $O(mc^2)$ time;
(ii) $\matXi$ in $O( m n \xi + mc\xi)$ time; and
(iii) $\matDelta$ in $O(c^2 \xi)$ time.
\end{proof}
In words, the lemma indicates that given the matrix $\matA$ and the subspace $\matV$ such that
$$ \FNormS{\matA-\Pi_{\matV,k}^{\mathrm{F}}(\matA)} \le (1+O(\epsilon)) \FNormS{\matA - \matA_k},$$
for some small $\varepsilon > 0,$ there exists a randomized algorithm to compute matrices
$\matY \in span(\matV)$ and $\matDelta$ such that the residual error, in Frobenius norm, from projecting $\matA$ onto the span of $\matY \matDelta$ is also bounded by $(1+O(\epsilon)) \FNormS{\matA - \matA_k}$.
Notice that $\Delta$ contains the top $k$ left singular vectors of a ``sketched'' version of $\matY\transp\matA,$ a property which makes the computation particularly effective in terms of running time.

\subsection{Detailed description of the algorithm}\label{sec:algorithm1}
\begin{small}
{\bf Input:}
\begin{enumerate}
\item $\matA \in \R^{m \times n}$ partitioned column-wise
$
\matA =
\begin{pmatrix}
\matA_1 & \matA_2 & \dots & \matA_s
 \end{pmatrix};
$
for $i=1:s,$ $\matA_i \in \R^{m \times w_i}$; $\sum_i w_i=n$.
\item rank parameter $k < \rank(\matA)$
\item accuracy parameter $\varepsilon > 0$
\end{enumerate}
{\bf Algorithm}
\begin{enumerate}
\item {\bf Local Column Sampling}
\begin{enumerate}
\item For each sub-matrix $\matA_i \in \R^{m \times w_i},$ compute $\matC_i \in \R^{m \times \ell}$ containing
$\ell = O(k)$ columns from $\matA_i$ as follows: $\matC_i =  \matA_i \matS_i$. Here, $\matS_i$ has dimensions $w_i \times \ell$ and is constructed as follows: $ \matS_i = BssSampling\1(\matA_i, 4k, 1/2, 0.01/s)$~(see Lemma~\ref{lem:bssproject}).

\item Machine $i$ sends $\matC_i$ to the server.

\end{enumerate}

\item {\bf Global Column Sampling}

\begin{enumerate}
\item Server constructs $m \times (s \cdot \ell)$ matrix $\matG$ containing $(s \cdot \ell)$ actual columns from $\matA$ as follows:
$
\matG =
\begin{pmatrix}
   \matC_1  &
    \matC_2 &
    \dots &
    \matC_s
 \end{pmatrix}.
$
Then, server constructs $\matC \in \R^{m \times c_1}$ via choosing $c_1 = 4k$ columns from $\matG$ as follows:
$\matC = DeterministicCssFrobenius(\matG, k, c_1)$~(see Lemma~\ref{thm:optimalFdet}).

\item Server sends $\matC$ to all the machines.

\end{enumerate}

\item {\bf Adaptive Column Sampling}

\begin{enumerate}
\item Machine $i$ computes $\matPsi_i = \matA_i - \matC \pinv{\matC} \matA_i \in \R^{m \times w_i}$
and then computes $\beta_i$ as it was described in Lemma~\ref{lem:adaptiveSampling2}.
Machine $i$ sends $\beta_i$ to server.

\item Server computes probability distribution  $g_i = \frac{\beta_i}{\sum_i \beta_i}.$
Server samples i.i.d. with replacement $\ceil{50k/\varepsilon}$ samples (machines) from $g_i$.
Then, server determines numbers $t_i$ ($i=1,2,\dots,s$), where $t_i$ is
the number of times the $i$th machine was sampled. It sends the $t_i$'s to the machines.

\item Machine $i$ computes probabilities $ q_{j}^{i} = \TNormS{ \x  } / \FNormS{\matPsi_i}$ ($j=1:w_i$),
where $\x$ is the $j$th column of $\matPsi_i$. And now machine $i$ samples $t_i$ samples
from it's local probability distribution and sends the corresponding columns to the server.
Let $c_2 = \sum_i t_i = \ceil{50k/\varepsilon}$.

\item Server collects the columns and assigns them to $\hat\matC \in \R^{m \times c_2}$.
Server constructs $\tilde\matC$ to be the $m \times (c_1 + c_2) $ matrix:
$
\tilde\matC =
\begin{pmatrix}
   \matC; & \hat\matC
 \end{pmatrix}.
$
Let $c = c_1 + c_2 = 4k + \ceil{50k/\varepsilon}$.
\end{enumerate}

\item {\bf Rank-$k$ matrix in the span of $\tilde{\matC}$}
\begin{enumerate}

\item Server sends $\tilde\matC$  to all the machines and each machine
computes (the same) qr factorization of $\tilde\matC$: $\tilde\matC = \matY \matR$ where
$\matY \in \R^{m \times c}$ has orthonormal columns and $\matR \in \R^{c \times c}$ is upper triangular.

\item Machine $i$ generates $\tilde\matW_i \in \R^{\xi \times w_i}$ with $\xi = O(c /\varepsilon^2)$ to be i.i.d.
$\{+1,-1\}$ w.p. $1/2$. \emph{Implicitly} all machines together generate
$
\tilde\matW =
\begin{pmatrix}
  \tilde\matW_1  &
   \tilde\matW_2 &
    \dots &
    \tilde\matW_s &
 \end{pmatrix},
$
with $\tilde\matW \in \R^{\xi \times n}$.
 Machine $i$ computes $\matH_i = \tilde\matC\transp (\matA_i \matW_i\transp) \in \R^{c \times \xi}$.
Machine $i$ sends $\matH_i$ to the server.

\item Server computes $\matXi = \sum_{i=1}^{s} \matH_i \in \R^{c \times \xi}$ and sends this back to all the machines. Now machines compute
$\matXi := \matR^{-1} \cdot \matXi \cdot 1/\sqrt{n} (:= \matY\transp \matA \matW\transp)$, where $
\matW =
\begin{pmatrix}
  \matW_1  &
   \matW_2 &
    \dots &
    \matW_s &
 \end{pmatrix} \in
 \R^{\xi \times n}$ is random matrix
each element of which is $\{+1/\sqrt{n},-1/\sqrt{n}\}$ w.p. $1/2$,
and then they compute  $\matDelta \in \R^{c \times k}$ to be the top $k$ left singular vectors of $\matXi.$ Each machine computes $\matU = \matY \cdot \matDelta \in \R^{m \times k}$.
\end{enumerate}
\end{enumerate}
\end{small}

\paragraph{Discussion.} A few remarks are necessary for the last two stages of the algorithm. The third stage (adaptive column sampling), implements the adaptive sampling method of Lemma~\ref{thm:adaptive}.
To see this, note that each column in $\matA$ (specifically the $j$th column in $\matA_i$) is sampled with probability
$$
g_i  \cdot q_j^{i}
\ge \ \frac{1}{2} \cdot \frac{ ||\x_j^i||_2^2 }{ \FNormS{\matPsi}},
$$
where $\matPsi = \matA-\matC\pinv{\matC}\matA$. This follows from Lemma~\ref{lem:adaptiveSampling1}.
Overall, this stage effectively constructs $\hat\matC$ such that (see Lemma~\ref{thm:adaptive})
$$\hat\matC = AdaptiveCols(\matA, \matC, c_2,1/2).$$


The last stage in the algorithm implements the algorithm in Lemma~\ref{lem:KVW}. To see this, note that $\matW$ satisfies the properties in the lemma. Hence,
$$ [\matY, \matDelta] = ApproxSubspaceSVD(\matA, \tilde\matC, k, \varepsilon). $$

\subsection{Main Result}
The theorem below analyzes the approximation error, the communication complexity,
and the running time of the previous algorithm.
%
\begin{theorem}\label{thm1}
The matrix $\tilde\matC \in \R^{m \times c}$
with $c=O(k/\varepsilon)$ columns of $\matA$ satisfies w.p. at least $0.98,$
\begin{equation}\label{eqnthm1}
 \FNormS{\matA -  \tilde\matC  \pinv{\tilde\matC}\matA} \le
\FNormS{\matA-\Pi_{\tilde\matC,k}^{\mathrm{F}}(\matA)}\le
(1 + O(\varepsilon)) \cdot \FNormS{\matA - \matU_k\matU_k\transp\matA}.
\end{equation}
The matrix $\matU \in \R^{m \times k}$  with $k$ orthonormal columns satisfies with probability at least $0.97,$
\begin{equation}\label{eqnthm2}
 \FNormS{\matA -  \matU \matU\transp \matA} \le
(1 + O(\varepsilon)) \cdot\FNormS{\matA - \matU_k\matU_k\transp\matA}.
\end{equation}
Let each column of $\matA$ has at most $\phi$ non-zero elements.
Then, the communication cost of the algorithm is
$O\left(s k \phi \varepsilon^{-1} + s k^2 \varepsilon^{-4}  \right)$
and the running time is
$O\left( m n s \cdot poly(k, 1/\varepsilon) \right).$
\end{theorem}
\subsection{Proof of Theorem~\ref{thm1}}\label{sec:proof}


\subsubsection{Proof of Eqn.~\ref{eqnthm1}}
From Lemma~\ref{thm:adaptive}:
\begin{equation}\label{eqn:adp}
\Expect{ \FNormS{ \matA - \tilde\matC \pinv{\tilde\matC}\matA} }
\le
\Expect{\FNormS{\matA-\Pi_{\tilde\matC,k}^{\mathrm{F}}(\matA)}}
\le \sum_{i=k+1}^{\rank(\matA)} \sigma_i^2(\matA)  + \frac{2 \varepsilon}{50} \cdot \FNormS{\matA - \matC \pinv{\matC}\matA}.
\end{equation}
Therefore, we need to bound $\FNormS{\matA - \matC \pinv{\matC}\matA} $ now.
\begin{lemma}\label{lem:projcomb}
With probability at least $0.99,$ matrix $\matG$ is a $(4k,1/2)$-projection-cost preserving sketch of $\matA$.
\end{lemma}
\begin{proof}
For each $i$, $\matC_i=\matA_i\matS_i$ is a $(4k,1/2)$-projection-cost preserving sketch of $\matA_i$ with probability at least $1-0.01/s$. Due to union bound, with probability at least $1-0.01,\forall i, \matC_i$ is a $(4k,1/2)$-projection-cost preserving sketch of $\matA_i$. Since $\matC_i$ is a $(4k,1/2)$-projection-cost preserving sketch of $\matA_i$, we can assume $c_i$ is a constant which is independent from projection matrix $\matP$ which has rank at most $4k$ such that
$$\frac12\FNormS{\matA_i-\matP\matA_i}\leq \FNormS{\matC_i-\matP\matC_i}+c_i\leq \frac32\FNormS{\matA_i-\matP\matA_i}$$
Then, for any projection matrix $\matP$ which has rank at most $4k$ we have:
\begin{align*}
\frac12\FNormS{\matA-\matP\matA} &=\frac12\sum_{i=1}^s\FNormS{\matA_i-\matP\matA_i}\\
&\leq \sum_{i=1}^s\left(\FNormS{\matC_i-\matP\matC_i}+c_i\right)\\
&\leq \frac32\sum_{i=1}^s\FNormS{\matA_i-\matP\matA_i}\\
&\leq \frac32\FNormS{\matA-\matP\matA}\\
\end{align*}
Consider the right hand side of the first inequality, it is equal to
$$\FNormS{\matG-\matP\matG}+\sum_{i=1}^s c_i$$
Let constant $c=\sum_{i=1}^s c_i$, we have:
$$\frac12\FNormS{\matA-\matP\matA}\leq \FNormS{\matG-\matP\matG}+c \leq \frac32\FNormS{\matA-\matP\matA}$$
Therefore, $\matG$ is a $(4k,1/2)$-projection-cost preserving sketch of $\matA$.
\end{proof}
For convenience, we use $\matP_{\matG}$ to denote a rank-$k$ projection matrix which provides the best rank-$k$ approximation $\matG_k=\matP_{\matG}\matG$ to matrix $\matG$. We also use $\matP_{\matA}$ to denote a rank-$k$ projection matrix which provides the best rank-$k$ approximation $\matA_k=\matP_{\matA}\matA$ to matrix $\matA$. Thus,
\begin{align*}
\frac12\FNormS{\matA-\matC\matC^\dagger\matA} &\leq \FNormS{\matG-\matC\matC^\dagger\matG}+c\\
&\leq 5\FNormS{\matG-\matP_{\matG}\matG}+c\\
&\leq 5\FNormS{\matG-\matP_{\matA}\matG}+c\\
&\leq 5\left(\FNormS{\matG-\matP_{\matA}\matG}+c\right)\\
&\leq \frac{15}2\FNormS{\matA-\matP_{\matA}\matA}\\
\end{align*}
The first inequality is due to Lemma~\ref{lem:projcomb}. $\matG$ is a $(4k,1/2)$-projection cost preserving sketch of $\matA$ and $\matC$ has rank at most $4k$. The second inequality is based on Lemma~\ref{thm:optimalFdet}. The third inequality is because $\matP_{\matG}$ provides the best rank-$k$ approximation to $\matG$. The fourth inequality used the fact $c$ is non-negative. The fifth inequality held because $\matG$ is a $(4k,1/2)$-projection cost preserving sketch of $\matA$. Thus, $$\FNormS{\matA-\matC\matC^\dagger\matA}\leq15\FNormS{\matA-\matA_k}$$

Combining with equation~\ref{eqn:adp}, we have:
\begin{equation}\label{eqn:adp2}
\Expect{ \FNormS{ \matA - \tilde\matC \pinv{\tilde\matC}\matA} }
\le
\Expect{\FNormS{\matA-\Pi_{\tilde\matC,k}^{\mathrm{F}}(\matA)}}
\le (1+\frac{30 }{50}\varepsilon)\sum_{i=k+1}^{\rank(\matA)} \sigma_i^2(\matA)
=(1+O(\eps))\FNormS{\matA-\matA_k}
\end{equation}

\paragraph{Concluding the proof.}
Consider about
\begin{equation}\label{eqn3}
\Expect{ \FNormS{ \matA - \tilde\matC \pinv{\tilde\matC} \matA } } \le
\Expect{\FNormS{\matA-\Pi_{\tilde\matC,k}^{\mathrm{F}}(\matA)}} \le
\sum_{i=k+1}^{\rank(\matA)} \sigma_i^2(\matA)  + O( \varepsilon) \sum_{i=k+1}^{\rank(\matA)} \sigma_i^2(\matA).
\end{equation}
The expectation is taken w.r.t. the randomness in constructing $\tilde\matC;$
hence, the term $\sum_{i=k+1}^{\rank(\matA)} \sigma_i^2(\matA)$ is a constant w.r.t. the
expectation operation, which means that Eqn.~\ref{eqn3} implies the bound:
$$
\Expect{\FNormS{\matA-\Pi_{\tilde\matC,k}^{\mathrm{F}}(\matA)} - \sum_{i=k+1}^{\rank(\matA)} \sigma_i^2(\matA)} \le O( \varepsilon ) \cdot \sum_{i=k+1}^{\rank(\matA)} \sigma_i^2(\matA).
$$
Let $Y: = \FNormS{\matA-\Pi_{\tilde\matC,k}^{\mathrm{F}}(\matA)} - \sum_{i=k+1}^{\rank(\matA)} \sigma_i^2(\matA)$. $Y$ is a random variable with $Y \ge 0,$ because $\rank(\Pi_{\tilde\matC,k}^{\mathrm{F}}(\matA)) \le k$. Markov's inequality on $Y$ implies that with arbitrary constant probability,
$$
\FNormS{\matA-\Pi_{\tilde\matC,k}^{\mathrm{F}}(\matA)} - \sum_{i=k+1}^{\rank(\matA)} \sigma_i^2(\matA) \le O(
\varepsilon) \cdot \sum_{i=k+1}^{\rank(\matA)} \sigma_i^2(\matA),
$$
equivalently,
$$
\FNormS{\matA-\Pi_{\tilde\matC,k}^{\mathrm{F}}(\matA)} \le \sum_{i=k+1}^{\rank(\matA)} \sigma_i^2(\matA)  +O(
\varepsilon) \cdot \sum_{i=k+1}^{\rank(\matA)} \sigma_i^2(\matA).
$$
Using $\FNormS{ \matA - \tilde\matC \pinv{\tilde\matC} \matA } \le
\FNormS{\matA-\Pi_{\tilde\matC,k}^{\mathrm{F}}(\matA)} $ concludes the proof.


\subsubsection{Proof of Eqn.~\ref{eqnthm2}}
We would like to apply Lemma~\ref{lem:KVW} for the matrix $\matU \in \R^{m \times k}$ in the algorithm.
Note that we have already proved that the matrix $\tilde\matC$ in the algorithm satisfies with probability at least $0.99$:
$ \FNormS{\matA-\Pi_{\tilde\matC,k}^{\mathrm{F}}(\matA)} \le (1+O(\varepsilon)) \FNormS{\matA - \matA_k}.$
Lemma~\ref{lem:KVW} and a simple union bound conclude the proof.

\subsubsection{Running time}
Next, we give an analysis of the arithmetic operations required in various steps in the algorithm.
\begin{enumerate}
\item {\bf Local Column Sampling}: $O(\nnz(A)\log(ms)+m\cdot \poly(k,\eps,\log(s)))$ arithmetic operations in total.
\begin{enumerate}
\item $O(\nnz(A_i)\log(ms)+m\cdot \poly(k,\eps,\log(s)))$ for each $\matS_i$
\item -
\end{enumerate}
\item {\bf Global Column Sampling}: $O(s k^3 + s k m)$ arithmetic operations in total.
\begin{enumerate}
\item $O(s k^3 + s k m)$ to construct $\matC$.
\item -
\end{enumerate}
\item {\bf Adaptive Sampling} $O(sk^2 m + m n k + k/\varepsilon)$ arithmetic operations in total.
\begin{enumerate}
\item  $O(k^2 m)$ for each $\pinv{\matC}$ locally;
$O(w_i m k)$ for each $\matPsi_i$ and $O(m w_i)$ for each  $\psi_i$.
\item $O(s)$ in total.
\item $O(k/\varepsilon)$ in total using the method of~\cite{Vos91} locally.
\item $O(k/\varepsilon)$.
\end{enumerate}
\item {\bf  Rank-$k$ matrix in $span(\tilde{\matC})$}:
$O(m n k /\varepsilon^3 + s m k^2 / \varepsilon^4 +
s  k^3/ \varepsilon^5)$ arithmetic operations in total.
\begin{enumerate}
\item $O(s m k^2/\varepsilon^2)$ in total
\item $O( m w_i k /\varepsilon^3 )$ for each $\matA_i \matW_i\transp$
and another $O( m k^2 \varepsilon^{-4} )$ for each $\tilde\matC\transp
(\matA_i \matW_i\transp)$
\item $O(s  k^2 \varepsilon^{-3} )$ to find $\matXi$ in the server;
then another $O(sk^3\varepsilon^{-3}  + s k^3 \varepsilon^{-5} ))$
to update $\matXi$ locally in each machine and another
$O(s k^3 \varepsilon^{-5} )$ to
find $\matDelta$ locally in each machine; and $O(s m k^2/\varepsilon)$ to find
$\matU$ locally in each machine.
\end{enumerate}
\end{enumerate}

\subsubsection{Communication Complexity}
Assume that  we can represent each element in the input matrix $\matA$ with at most
$b$ bits.  Also, we assume that one word has length $b$ bits and $b = O(\log( m n s / \varepsilon ))$.

\begin{enumerate}
\item {\bf Local Column Sampling}: $O(s k \phi)$ words.
\begin{enumerate}
\item -
\item $O(s k \phi)$ elements of $\matA$.
\end{enumerate}

\item {\bf Global Column Sampling}:$O(s k \phi)$ words.
\begin{enumerate}
\item -
\item $O(s k \phi)$ elements of $\matA$.
\end{enumerate}
\item {\bf Adaptive Sampling}: $O(s+ \phi k /\varepsilon)$ words.
\begin{enumerate}
\item $s$ integers each of which is representable with $O(\log k + \log \log(mns/\varepsilon))$ (from Lemma~\ref{lem:adaptiveSampling2}).

\item $s$ integers~(the $t_i$'s) each with magnitude at most $n$, hence representable with $b$ bits.
\item $O(\phi k/\varepsilon)$ elements of $\matA$.
\item -
\end{enumerate}

\item {\bf  Best rank-$k$ matrix in the span of $\tilde{\matC}$}:
$O( s\phi k/\varepsilon +s k^2  \varepsilon^{-4} )$ words.
\begin{enumerate}
\item $O(s k \phi \varepsilon^{-1})$ elements of $\matA$.
\item $O( s k^2  \varepsilon^{-4})$ numbers each of which can be represented with $b$ bits.
\item $O( s k^2  \varepsilon^{-4})$ numbers each of which can be represented with $b$ bits.
\end{enumerate}
\end{enumerate}
In total the communication complexity is $O( s k \phi / \varepsilon + sk^2/\varepsilon^4)$ words.

\section{Faster Distributed PCA for sparse matrices}\label{sec:faster}
We now explain how to modify certain steps of the distributed PCA algorithm in Section~\ref{sec:algorithm} in order to improve the total running time spent to compute the matrix $\matU$; this, at the cost of increasing the communication cost \emph{slightly}. Specifically, we replace the parts
2-$(a),$ 3-$(a),$ and 4-$(b)$ with similar procedures, which are almost as accurate as the original procedures but run in time proportional to the number of the non-zero elements of the underlying matrices. Before presenting the new algorithm in detail, we discuss results from previous literature that we employ in the analysis. We remark that in the analysis below we have not attempted to optimize the constants.

\subsection{Background material}

\subsubsection{Constant probability sparse subspace embeddings and sparse SVD}\label{sec:CWT}
First, we present the so-called sparse subspace embedding matrices of Clarskon and Woodruff~\cite{CW13arxiv}.

\begin{definition}\label{def:sse} [Sparse Subspace Embedding~\cite{CW13arxiv}]
We call $\matW \in \R^{\xi \times n}$ a sparse subspace embedding of dimension $\xi < n$
if it is constructed as follows,
$ \matW = \matPsi \cdot \matY, $
where
\begin{itemize}
\item $h: [n] \rightarrow [\xi]$ is a random map so that for each $i \in [n], h(i) = \xi',$ for $\xi' \in [\xi]$ w.p. $1/\xi$.
\item $\matPsi \in \R^{\xi \times n}$ is a binary
matrix with $\matPsi_{h(i),i} = 1,$ and all remaining entries $0$.
\item $\matY \in \R^{n \times n}$ is a random diagonal matrix, with each diagonal entry independently chosen to be $+1$ or $-1$, with equal probability.
\end{itemize}
\end{definition}
Such sparse subspace embedding matrices have favorable properties which we summarize below:
\begin{lemma}[\cite{CW13arxiv}] \label{lem:subspacesparse1}
Let $\matA\transp \in \R^{n \times m}$ have rank $\rho$ and let $\matW \in \R^{\xi \times n}$
be a randomly chosen sparse subspace embedding with dimension $\xi  = \Omega(\rho^2 \varepsilon^{-2})$, for some $0 < \varepsilon < 1.$
Then, 1) computing $\matA \matW\transp$ requires $O( \nnz(\matA) )$ time;
2) $\nnz(\matA \matW\transp) \le \nnz(\matA)$; and
3) with probability at least $0.99,$ and for all vectors $\y \in \R^m$ simultaneously,
$$ (1-\varepsilon) \TNormS{\matA\transp \y} \le   \TNormS{ \matW\matA\transp \y} \le (1+\varepsilon) \TNormS{\matA\transp \y}.$$
\end{lemma}
Notice that the third property in the lemma fails with constant probability.
Based on the above sparse subspace embeddings, Clarkson and Woodruff~\cite{CW13arxiv} described a
low-rank matrix approximation algorithm for $\matA$ that runs
in time $\nnz(\matA)$ plus low-order terms.

\begin{lemma}[Theorem 47 in~\cite{CW13arxiv}; for the exact statement see Lemma 3.4 in~\cite{BW14}]\label{lem:approxSVDsparse}
Given \math{\matA\in\R^{m\times n}} of rank $\rho$, a target rank $1 \leq k < \rho$, and $0 < \epsilon \le 1$, there exists a randomized algorithm that  computes
$\matZ \in \R^{n \times k}$ with $\matZ\transp\matZ = \matI_k$  and with probability at least $0.99$,
$$\FNormS{\matA - \matA \matZ \matZ\transp} \leq \left(1+{\epsilon}\right)\FNormS{\matA - \matA_k}.$$
The proposed algorithm requires
$O\left(  \nnz(\matA) + (m+n) \cdot poly(k,\varepsilon^{-1}) \right) $
arithmetic operations. We denote this procedure as
$\matZ = SparseSVD(\matA, k, \varepsilon).$

\end{lemma}
Notice that the algorithm in the lemma fails with constant probability.

\subsubsection{High probability sparse  SVD}\label{sec:CWT2}

Lemma~\ref{lem:approxSVDsparse} above presents a randomized algorithm to compute quickly an SVD of a sparse matrix; however, this lemma fails with constant probability. Since in our fast distributed PCA algorithm we will use the lemma multiple times (specifically, $s$ times, since we need such a fast SVD locally in each machine), the overall failure probability would be of the order $O(s)$ (by a union bound). To boost this failure probability, we develop a simple scheme to obtain a high probability result.

Before we proceed, we quickly review the Johnson Lindestrauss transform, which we need in the analysis.
\begin{lemma}\label{lem:jlt} [Theorem 1 in~\cite{Ach03} for fixed $\varepsilon = 1/2$]
Let $\matB \in \R^{m \times n}$. Given $\beta > 0,$ let
$r = \frac{4 + 2 \beta}{(1/2)^2 - (1/2)^3 } \log n.$
Construct a matrix $\matS \in \R^{r \times m},$ each element of which is a random variable which takes values $\pm 1/\sqrt{r}$ with
equal probability. Let $\tilde{\matB} = \matS \matB$. Then, if $\b_i$ and $\tilde{\b}_i$ denote the $i$th column of $\matB$ and $\tilde{\matB}$,
respectively, with probability at least $1 - n^{-\beta},$ and for all $i=1,...n,$
$ (1-\frac12) \TNormS{\b_i} \le \TNormS{ \tilde{\b}_i } \le (1+\frac{1}{2}) \TNormS{\b_i}.$
Given $\matB,$ it takes $O(\nnz(\matB) \log n )$ arithmetic operations to construct $\tilde{\matB}$. We will denote this procedure as
$ \tilde{\matB} = JLT(\matB, \beta). $
\end{lemma}
Now, we are fully equipped to design the ``high-probability'' analog of Lemma~\ref{lem:approxSVDsparse}:
\begin{lemma}
\label{lem:approxSVDsparseBoosting}
Given \math{\matA\in\R^{m\times n}} of rank $\rho$, a target rank $1 \leq k < \rho$, and $0 < \delta, \epsilon \le 1$, there exists a randomized algorithm that  computes
$\matZ \in \R^{n \times k}$ with $\matZ\transp\matZ = \matI_k$  and with probability at least $1-\delta-1/n$,
$\FNormS{\matA - \matA \matZ \matZ\transp} \leq 3\left(1+{\epsilon}\right)\FNormS{\matA - \matA_k}.$
The proposed algorithm requires
$O\left(  \nnz(\matA) \log^2(\frac{n}{\delta}) \right) + n \cdot poly(k,\varepsilon,\log(\frac{n}{\delta}))$
arithmetic operations. We denote this procedure as
$\matZ = SparseSVDBoosting(\matA, k, \varepsilon, \delta).$
\end{lemma}
\begin{proof}
The idea is to use the algorithm in Lemma~\ref{lem:approxSVDsparse} and generate i.i.d.
$r = O(\log(\frac{1}{\delta}))$ matrices $\matZ_1, \matZ_2, \dots, \matZ_r,$ where $\matZ_ i = SparseSVD(\matA, k, \varepsilon),$ and  $\FNormS{\matA - \matA \matZ_i \matZ_i\transp} \le (1+\varepsilon)\FNormS{\matA - \matA_k},$ with probability at least $0.99$ for a single $\matZ_i$.

Now, we could have chosen
$\matZ:=\matZ_i$  which minimizes $ \FNormS{\matA - \matA \matZ_i \matZ_i\transp},$ and this
would have implied that
$\FNormS{\matA - \matA \matZ  \matZ \transp} \le (1+\varepsilon)\FNormS{\matA - \matA_k},$ with probability $1-\delta$ (via a simple chernoff argument), but evaluating the errors $\FNormS{\matA - \matA \matZ_i \matZ_i\transp}$ is computationally intensive for the overall running time that we would like to obtain.

\paragraph{Selection of $\matZ$.}
Instead, we use the Johnson Lindestrauss transform  to speedup this computation (evaluating the errors). Hence, for every $\matZ_i$, we (implicitly) set $\matB_i := \matA - \matA \matZ_i \matZ_i\transp$ and
then we compute $\tilde\matB_i = JLT(\matB_i, 1)$ (see Lemma~\ref{lem:jlt}). Now, we choose the $\matZ_i$ that minimizes $\FNormS{ \tilde\matB_i}$.

\paragraph{Correctness.} For a fixed $\matZ_i,$
from Lemma~\ref{lem:jlt} it follows that with probability $1-1/n$ it is:
$ \frac{1}{2} \FNormS{\matA - \matA \matZ_i \matZ_i\transp} \le
\FNormS{ \tilde\matB_i } \le \frac{3}{2} \FNormS{\matA - \matA \matZ_i \matZ_i\transp},$
where the failure probability is over the randomness of $\matS$ in Lemma~\ref{lem:jlt}.
 Now, for some $\matZ_i$ with probability at least $0.99$ it is $ \FNormS{\matA - \matA \matZ_i \matZ_i\transp} \le (1+\varepsilon)\FNormS{\matA-\matA_k}$, where the failure probability is over the randomness in constructing $\matZ_i$.
Overall, given that $n > 1,$ with probability at least $0.49$ (over the randomness of both $\matS$ and $\matZ_i$; the failure probability follows via a union bound)
for each single $\matZ_i$ it is:
$ \FNormS{ \tilde\matB_i } \le \frac{2}{3}   (1+\varepsilon)\FNormS{\matA-\matA_k}$.
Hence, since $r = O(\log\frac{1}{\delta})$ and we have picked $\matZ$ as the $\matZ_i$ that minimizes
$\FNormS{ \tilde\matB_i },$ it follows that with probability at least $1-\delta,$
$\FNormS{ \tilde\matB_i } \le \frac{2}{3}  (1+\varepsilon)\FNormS{\matA-\matA_k}$.
Combining with $ \frac{1}{2} \FNormS{\matA - \matA \matZ_i \matZ_i\transp} \le
\FNormS{ \tilde\matB_i }$, and using a union bound, it follows
that with probability $1 - 1/n - \delta$:
$
\FNormS{\matA - \matA \matZ_i \matZ_i\transp} \le \frac{6}{2}(1+\varepsilon)\FNormS{\matA-\matA_k}.
$

\paragraph{Running time.} The following costs occur $O(\log(\frac{1}{\delta}))$ times:
(i) $O\left(  \nnz(\matA) + (m+n) \cdot poly(k,\varepsilon^{-1}) \right)$ to find
$\matZ_i$ (via Lemma~\ref{lem:approxSVDsparse});
(ii) $O( \nnz(\matA) \log n + n k \log (n))$ to evaluate each cost, i.e., find $\FNormS{ \tilde\matB_i }$.

\end{proof}

\subsubsection{Fast column sampling techniques}\label{sec:samplingsparse}
Section~\ref{sec:sampling} presented the column sampling algorithms that we used in our distributed PCA algorithm in Section~\ref{sec:algorithm}. Next, we present the ``fast'' analogs of those algorithms. To design such fast analogs
we employ the sparse subspace embedding matrices in the previous section. All the results in this
section developed in~\cite{BW14}.
There is a small difference in Lemma~\ref{lem:dualsets} below, hence we present a short proof for completeness.
Specifically, the original result in~\cite{BW14} has a constant failure probability; here, via using standard arguments, we extended this result to a high probability bound.
\begin{lemma}[Input-Sparsity-Time Dual-Set Spectral-Frobenius Sparsification~\cite{BW14}.]
\label{lem:dualsets}
Let $\matV \in \R^{w \times k}$ be a matrix with $w > k$ and $\matV\transp \matV = \matI_{k}$.
Let $\matE \in \R^{m \times w}$ be an arbitrary matrix.
Let $\delta$ be a failure probability parameter.

For $i=1,2,\dots, r = O(\log\frac{1}{\delta})$,
let $\matB_i = \matE \matW_i\transp \in \R^{m \times \xi}$,
where $ \matW_i \in \R^{\xi \times w}$ is a randomly chosen
sparse subspace embedding with $\xi = O( k^2 / \varepsilon^2 ) < w$, for some $0 < \varepsilon < 1$
(see Definition~\ref{def:sse} and Lemma~\ref{lem:subspacesparse1}),
and run the algorithm of Lemma~\ref{lem:dualset} with
$\matV,$ $\matB_i,$ and some $\ell$ with \math{k < \ell \le w}, to construct $\matS_i \in \R^{w \times \ell}$.
Choose some $\matS$ from the $\matS_i$'s - see the proof for the details.

Then, with probability at least $1-\delta,$ $\sigma_{k}^2\left(\matV\transp \matS \right) \ge
\left(1 - \sqrt{{k}/{\ell}}\right)^2$ and
$\FNormS{\matE \matS}
\le \left( \frac{1 + \varepsilon }{1 - \varepsilon} \right)^2 \cdot  \FNormS{\matE}$.
The matrix $\matS$ can be computed in $O\left( \log\left(\frac{1}{\delta}\right) \cdot \left( \nnz(\matA) + \ell w k^2+ m k^2/\varepsilon^2 + m \ell \right) \right)$ time.
We denote this procedure as $\matS = BssSamplingSparse(\matV, \matA, r, \varepsilon, \delta).$
\end{lemma}
\begin{proof}
The algorithm constructs the $\matS_i$'s as follows,
$ \matS_i = BssSampling(\matV, \matB_i, r)$ - see Lemma~\ref{lem:dualset}.
From that lemma, we also have that $\sigma_{k}^2\left(\matV\transp \matS_i \right) \ge
\left(1 - \sqrt{{k}/{\ell}}\right)^2$ and
$
\FNormS{ \matB_i\transp \matS } \le \FNormS{\matB_i\transp},
$
i.e.,
$
\FNormS{ \matW_i \matA\transp \matS_i } \le \FNormS{\matW_i \matA\transp}.
$
Since $\matW_i$ is a subspace embedding, from Lemma~\ref{lem:subspacesparse1} we have that with probability at least $0.99$ and for \emph{all} vectors $\y \in \R^{n}$ simultaneously,
$
\left( 1 - \varepsilon \right) \TNormS{\mat\matA\transp \y} \le \TNormS{\matW_i \matA\transp \y}.
$
Apply this $r$ times for $\y \in \R^n$ being columns from $\matS_i \in \R^{n \times r}$ and take a sum on the resulting inequalities:
$\left( 1 - \varepsilon \right) \FNormS{\mat\matA\transp \matS_i} \le \FNormS{\matW \matA\transp \matS_i};$ also, apply this $n$ times for the basis vectors in $\R^n$ and take a sum on the resulting inequalities:
$ \FNormS{\matW \matA\transp} \le  \left( 1 + \varepsilon \right)   \FNormS{\matA\transp}. $
Combining all these inequalities together, we conclude that with probability at least $0.99,$
$ \FNormS{\mat\matA\transp \matS_i}  \le \frac{1 + \varepsilon }{1 - \varepsilon} \cdot \FNormS{\matA\transp}.$

To summarize, for each  $\matS_i$ and with probability at least $0.99,$
$\sigma_{k}^2\left(\matV\transp \matS_i \right) \ge
\left(1 - \sqrt{{k}/{\ell}}\right)^2$ and
$\FNormS{\matE \matS_i}
\le \left( \frac{1 + \varepsilon }{1 - \varepsilon} \right)^2 \cdot  \FNormS{\matE}$.

\paragraph{Selection of $\matS$.}
Now that we have $r$ matrices $\matS_i$ each of which satisfying the above bounds with  probability $0.99$
we select one of them as follows: 1) we compute all $\sigma_{k}^2\left(\matV\transp \matS_i \right)$ and sort them from the largest to the smallest value; 2) we compute all $\FNormS{\mat\matA\transp \matS_i}$ and sort them from the smallest to the largest value; 3) we choose an $\matS_i$ which occurs in the top $2/3$ fraction of each of the lists.

\paragraph{Correctness.} Let $X_i$ be an indicator random variable: $X_i = 1$ if the $i$th matrix
$\matS_i$ satisfies $\sigma_{k}^2\left(\matV\transp \matS_i \right) \ge
\left(1 - \sqrt{{k}/{\ell}}\right)^2$, and $X_i=0,$ otherwise. Then, $\Pr[X_i] \ge 0.99$.
Let $X = \sum_{i}^r X_i$. The expected value of $X$ is $\mu_{X} = \Expect{X} = 0.99 r$.
A standard Chernoff bound applied on $X$ gives: $\Pr[X \le (1 - \alpha) \mu_X] \le e^{ - \frac{\mu_X \alpha}{ 2}}$.
Let $(1 - \alpha) 0.99 r = 0.75 r$, i.e., $\alpha = 0.14 / 0.99$. Applying this to the Chernoff bound, we get:
 $\Pr[X \le 0.75 r] \le e^{ - O(r)}$. Replacing $r = O(\log(1/\delta))$:  $\Pr[X \le 0.75 r] \le \delta/2$. This means precisely that $3/4$ of the fraction of the $\matS_i$'s satisfy with probability $1-\delta/2:$
 $\sigma_{k}^2\left(\matV\transp \matS_i \right) \ge \left(1 - \sqrt{{k}/{\ell}}\right)^2$.

Similarly,  let $Y_i$ be an indicator random variable: $Y_i = 1$ if the $i$th matrix
$\matS_i$ satisfies $\FNormS{\matE \matS_i}
\le \left( \frac{1 + \varepsilon }{1 - \varepsilon} \right)^2 \cdot  \FNormS{\matE}$, and $Y_i=0,$ otherwise. Then, $\Pr[Y_i] \ge 0.99$.
Let $Y = \sum_{i}^r X_i$. The expected value of $Y$ is $\mu_{Y} = \Expect{Y} = 0.99 r$.
A standard Chernoff bound applied on $Y$ gives: $\Pr[Y \le (1 - \alpha) \mu_Y] \le e^{ - \frac{\mu_Y \alpha}{ 2}}$.
Let $(1 - \alpha) 0.99 r = 0.75 r$, i.e., $\alpha = 0.14 / 0.99$. Applying this to the Chernoff bound, we get:
 $\Pr[Y \le 0.75 r] \le e^{ - O(r)}$. Replacing $r = O(\log(1/\delta))$:  $\Pr[Y \le 0.75 r] \le \delta/2$. This means precisely that $3/4$ of the fraction of the $\matS_i$'s satisfy w.p. $1-\delta/2:$
$\FNormS{\matE \matS_i} \le \left( \frac{1 + \varepsilon }{1 - \varepsilon} \right)^2 \cdot  \FNormS{\matE}$.

Now a simple union bound indicates that with probability at least $1-\delta,$  the following two events happen simultaneously: $X > 0.75r$ and $Y > 0.75 r$. Since we select an element $\matS_i$ on top of the two sorted lists, it follows that this element should have $X_i = Y_i = 1$.

\paragraph{Running time.} The following costs occur $O(\log(\frac{1}{\delta}))$ times:
(i) $O(\nnz(\matA))$ to find $\matB_i$.
(i) $ O(\ell w k^2+ m \xi ) $ to find $\matS_i$ (via Lemma~\ref{lem:dualset}); and
(iii) $O( k^2 \ell + m \ell )$ to evaluate each cost and, i.e., find $\sigma_{k}^2\left(\matV\transp \matS_i \right)$ and $\FNormS{\matE \matS_i}$. Additionally,
$O(\log(\frac{1}{\delta}) \log \log(\frac{1}{\delta}))$ time in total is spent in sorting.
\end{proof}

The following Lemma is the ``fast'' analog of Lemma~\ref{thm:optimalFdet}. The failure probability in the lemma is constant, but this is sufficient for our purposes since we apply this lemma only once in the global column sampling step of the distributed PCA algorithm in Section~\ref{sec:algorithm2}.
\begin{lemma}[Input-Sparsity-Time constant factor column-based matrix reconstruction; Lemma 6.3 in~\cite{BW14}]
\label{thm:optimalFdets}
Given matrix \math{\matG\in\R^{m \times \alpha}} of rank $\rho$ and a target rank $k$~\footnote{The original Lemma 6.3 in~\cite{BW14} has the assumption that $k < \rho,$ but this assumption can be dropped having the result unchanged. The only reason the assumption $k < \rho$ exists is because otherwise column subset selection is trivial.},
there exists a randomized algorithm that runs in $O\left( \nnz(\matA) \cdot \log \alpha + m \cdot poly( \log \alpha, k, \varepsilon^{-1})  \right)$  time and
selects $c = 4k$
columns of \math{\matG} to form a matrix
$\matC\in\R^{m \times c}$
such that with probability at least $0.69$:
$$
\FNormS{\matG - \matC \pinv{\matC}\matG } \leq 4820 \cdot \sum_{i=k+1}^{\rank(\matG)}\sigma_i^2(\matG).
$$
We denote this procedure as $\matC = DeterministicCssFrobeniusSparse(\matG, k, c).$
\end{lemma}

Finally, the lemma below presents the ``fast'' analog of Lemma~\ref{thm:adaptive}. The failure probability in the lemma is, again, constant; we have not attempted to obtain a high probability bound since we employ this lemma only once in  the algorithm in Section~\ref{sec:algorithm2}.
\begin{lemma}[Input-sparsity-time Adaptive Sampling~\cite{BW14}]
\label{lem:adaptivecolumnss}
Given $\matA \in \R^{m \times n}$ and $\matV \in \R^{m \times c_1}$ (with $c_1 \le n, m$),
define the residual matrix
$\matPsi = \matA - \matV \matV^{\dagger} \matA \in \R^{m \times n}.$ Let $\tilde{\matPsi} = JLT(\matPsi, 1)$.
For $j=1,\ldots,n$,
let $p_j$ be a probability distribution such that, for some positive constant $\beta \le 1,$
$p_j \ge \beta  {\TNormS{\tilde\matPsi^{(j)}}}/{\FNormS{\tilde\matPsi}},$
where $\tilde\matPsi^{(j)}$ is the $j$-th column of the matrix $\tilde\matPsi$. Sample
$c_2$ columns from $\matA$ in \math{c_2} i.i.d. trials, where in each trial the $j$-th column is chosen with probability $p_i$.
Let $\matC_2 \in \R^{m \times c_2}$ contain the $c_2$ sampled columns and let $\matC = [\matV\ \ \matC_2] \in \R^{m \times (c_1+c_2)}$
contain the columns of $\matV$ and $\matC_2$.
Then, for any integer $k > 0$, and with probability $0.9-\frac{1}{n}$
$$
\FNormS{ \matA - \Pi_{\matC,k}^{\mathrm{F}}(\matA) }  \le \FNormS{ \matA - \matA_k } +  \frac{30 \cdot k}{\beta \cdot c_2} \FNormS{ \matA - \matV \matV^{\dagger} \matA}.
$$
%
Given $\matA$ and $\matV,$ the  algorithm takes $O( \nnz(\matA) \log n +  m c_1 \log n +  m c_1^2)$ time to find $\matC_2$.
We denote this sampling procedure as $\matC_2 = AdaptiveColsSparse(\matA, \matV, c_2, \beta).$
\end{lemma}

\subsubsection{Fast Low-rank matrix approximations within a subspace}\label{sec:bestFsparse}
Finally, we present the fast analog of the result in Section~\ref{sec:bestF}. The failure probability in the lemma is, again, constant; we have not attempted to obtain a high probability bound since we employ this lemma only once in  the algorithm in Section~\ref{sec:algorithm2}.

\begin{lemma}\label{lem:KVW2}
Let $\matA \in \R^{m \times n}$ be the input matrix and $\matV \in \R^{m \times c}$ be the input subspace.
We further assume that for some rank parameter $k < c$ and accuracy parameter $0 < \varepsilon <1:$
$$
\FNormS{\matA-\Pi_{\matV,k}^{\mathrm{F}}(\matA)} \le (1+O(\varepsilon)) \FNormS{\matA - \matA_k}.
$$
Let $\matV = \matY \matPsi$ be a $qr$ decomposition of $\matV$
with $\matY \in \R^{m \times c}$ and $\matPsi \in \R^{c \times c}$.
Let $\matXi = \matY\transp \matA \matW\transp \in \R^{c \times \xi},$
where $\matW\transp \in \R^{n \times \xi}$ with $\xi = O(c/\varepsilon^2),$
is a sparse subspace
embedding matrix~(see Definition~\ref{def:sse} and Lemma~\ref{lem:subspacesparse1}).
Let $\matDelta \in \R^{c \times k}$ contain the top $k$ left singular vectors of $\matXi$.
Then, with probability at least $0.99$,
$$
\FNormS{ \matA - \matY \matDelta \matDelta\transp \matY\transp \matA } \le (1+\varepsilon) \FNormS{\matA - \matA_k}.
$$
$\matY$ and $\matDelta$ can be computed in $O(\nnz(\matA) + m c \xi )$ time.
We denote this procedure as
$$
[\matY, \matDelta] = ApproxSubspaceSVDSparse(\matA, \matV, k, \varepsilon).
$$
\end{lemma}
\begin{proof}
This result was proven inside the proof of Theorem 1.5 in~\cite{KVW13}. Specifically,
the error bound proven in~\cite{KVW13} is for the transpose of $\matA$ (also $\matY,\matDelta$ are
denoted with $U, V$ in~\cite{KVW13}). The only
requirement for the embedding matrix $\matW$ (denoted with $P$ in the proof of Theorem 1.5 in~\cite{KVW13})
is to be a subspace embedding for $\matY\transp\matA$,
in the sense that $\matW$ is a subspace embedding for $\matA$ in Lemma~\ref{lem:subspacesparse1}.
Since our choice of $\matW$ with $\xi = O(c^2/\varepsilon^2)$ satisfies this requirement we omit the details of the proof. The running time
is $O(\nnz(\matA)+ m c \xi )$: one can compute (i) $\matY$ in $O(mc^2)$;
(ii) $\matXi$ in $O(\nnz(\matA) + mc\xi)$; and (iii) $\matDelta$ in $O(c \xi \min\{c,\xi\})$.
The failure probability $0.01$ from Lemma~\ref{lem:subspacesparse1}.
\end{proof}

\subsection{Detailed description of the algorithm}\label{sec:algorithm2}
This algorithm is very similar to the algorithm in Section~\ref{sec:algorithm1};
we only replace the parts 2-$(a),$ 3-$(a),$ and 4-$(b)$ with faster procedures. \\
\begin{small}
{\bf Input:}
\begin{enumerate}
\item $\matA \in \R^{m \times n}$ partitioned column-wise
$
\matA =
\begin{pmatrix}
\matA_1 & \matA_2 & \dots & \matA_s
 \end{pmatrix};
$
for $i=1:s,$ $\matA_i \in \R^{m \times w_i}$; $\sum_i w_i=n$.
\item rank parameter $k < \rank(\matA)$
\item accuracy parameter $\varepsilon > 0$
\item failure probability $\delta$
\end{enumerate}
{\bf Algorithm}
\begin{enumerate}
\item {\bf Local Column Sampling}
\begin{enumerate}



\item For each sub-matrix $\matA_i \in \R^{m \times w_i},$ compute $\matC_i \in \R^{m \times \ell}$ containing
$\ell = O(k)$ columns from $\matA_i$ as follows: $\matC_i =  \matA_i \matS_i$. Here, $\matS_i$ has dimensions $w_i \times \ell$ and is constructed as follows: $ \matS_i = BssSampling\1(\matA_i, 4k, 1/2, 0.01/s)$~(see Lemma~\ref{lem:bssproject}).

\item Machine $i$ sends $\matC_i$ to the server.

\end{enumerate}

\item {\bf Global Column Sampling}

\begin{enumerate}
\item Server constructs $m \times (s \cdot \ell)$ matrix $\matG$ containing $(s \cdot \ell)$ actual columns from $\matA$ as follows:
$
\matG =
\begin{pmatrix}
   \matC_1  &
    \matC_2 &
    \dots &
    \matC_s
 \end{pmatrix}.
$
Then, server constructs $\matC \in \R^{m \times c_1}$ via choosing $c_1 = 4k$ columns from $\matG$ as follows:
$\matC = DeterministicCssFrobeniusSparse(\matG, k, c_1)$~(see Lemma~\ref{thm:optimalFdets}).

\item Server sends $\matC$ to all the machines.

\end{enumerate}

\item {\bf Adaptive Column Sampling}

\begin{enumerate}
\item Server initializes the random seed and communicates it to all the machines.
Each machine constructs the same
$\matS \in \R^{r \times m},$ with $r = \frac{4 + 2}{(1/2)^2 - (1/2)^3 } \log n,$ each element of which is a random variable which takes values $\pm 1/\sqrt{r}$ with
equal probability.
Machine $i$ finds $\tilde\matPsi_i = \matS \matA_i -\matS \matC \pinv{\matC} \matA_i \in \R^{m \times w_i}$
and then computes $\beta_i$ as it was described in Lemma~\ref{lem:adaptiveSampling2}.
Machine $i$ sends $\beta_i$ to server.

\item Server computes probability distribution  $g_i = \frac{\beta_i}{\sum_i \beta_i}.$
Server samples i.i.d. with replacement $\ceil{50k/\varepsilon}$ samples (machines) from $g_i$.
Then, server determines numbers $t_i$ ($i=1,2,\dots,s$), where $t_i$ is
the number of times the $i$th machine was sampled. It sends the $t_i$'s to the machines.

\item Machine $i$ computes probabilities $ q_{j}^{i} = \TNormS{ \x  } / \FNormS{\tilde\matPsi_i}$ ($j=1:w_i$),
where $\x$ is the $j$th column of $\tilde\matPsi_i$. And now machine $i$ samples $t_i$ samples
from it's local probability distribution and sends the corresponding columns to the server.
Let $c_2 = \sum_i t_i = \ceil{50k/\varepsilon}$.

\item Server collects the columns and assigns them to $\hat\matC \in \R^{m \times c_2}$.
Server constructs $\tilde\matC$ to be the $m \times (c_1 + c_2) $ matrix:
$
\tilde\matC =
\begin{pmatrix}
   \matC; & \hat\matC
 \end{pmatrix}.
$
Let $c = c_1 + c_2 = 4k + \ceil{50k/\varepsilon}$.
\end{enumerate}

\item {\bf Rank-$k$ matrix in the span of $\tilde{\matC}$}

\begin{enumerate}

\item Server sends $\tilde\matC$  to all the machines and each machine
computes (the same) qr factorization of $\tilde\matC$: $\tilde\matC = \matY \matR$ where
$\matY \in \R^{m \times c}$ has orthonormal columns and $\matR \in \R^{c \times c}$ is upper triangular.

\item  Server initializes the random seed and sends this to each machine, such that each machine generates the same
matrix $\matPsi \in \R^{\xi \times n},$ for $\xi = O(c^2 / \varepsilon^2),$ (see Definition~\ref{def:sse}).
Machine $i$ generates $\matD_i \in \R^{n \times w_i}$ such that implicitly
$
\matD =
\begin{pmatrix}
   \matD_1  &
    \matD_2 &
    \dots &
    \matD_s
 \end{pmatrix}
$
is an $n \times n$ matrix each element of which is $\pm 1,$ with probability $1/2$.
Each machine constructs implicitly $\matW_ i = \matPsi \cdot \matD_i \in \R^{\xi \times w_i}$.
 \emph{Implicitly} all machines together generate
$
\matW =
\begin{pmatrix}
   \matW_1  &
    \matW_2 &
    \dots &
    \matW_s
 \end{pmatrix},
$
with $\matW \in \R^{\xi \times n}$ and $\matW$ is a subspace embedding as in Definition~\ref{def:sse}.
Machine $i$ computes $\matH_i = \tilde\matC \transp (\matA_i \matW_i\transp) \in \R^{c \times \xi}$.
Machine $i$ sends $\matH_i$ to the server.

\item Server computes $\matXi = \sum_{i=1}^{s} \matH_i \in \R^{c \times \xi}$ and sends this back to all the machines. Now machines compute
$\matXi := \matR^{-1} \cdot \matXi  (:= \matY\transp \matA_i \matW_i\transp)$,
and then they compute  $\matDelta \in \R^{c \times k}$ to be the top $k$ left singular vectors of $\matXi.$ Each machine computes $\matU = \matY \cdot \matDelta \in \R^{m \times k}$.
\end{enumerate}

\end{enumerate}
\end{small}

\paragraph{Discussion.} A few remarks are necessary for the last two stages of the algorithm. The third stage (adaptive column sampling), implements the adaptive sampling method of Lemma~\ref{lem:adaptivecolumnss}. To see this note that each column in $\matA$ is sampled with probability
$$
q_j^{i} \cdot g_i \ge \frac{1}{2} \cdot ||\x||_2^2 / \FNormS{\tilde\matPsi},
$$
where $\x$ is this column in $\tilde\matPsi$.
This follows from Lemma~\ref{lem:adaptiveSampling1}.
Overall, this stage constructs $\hat\matC$ such that $\hat\matC = AdaptiveColsSparse(\matA, \matC, c_2,1/2).$

The last stage in the algorithm implements the algorithm in Lemma~\ref{lem:KVW2}. To see this, note that $\matW$ satisfies the properties in the lemma. Hence,
$$ [\matY, \matDelta] = ApproxSubspaceSVDSparse(\matA, \tilde\matC, k, \varepsilon). $$

\subsection{Main result}
The theorem below analyzes the approximation error, the communication complexity,
and the running time of the previous algorithm.
\begin{theorem}\label{thm1s}
The matrix $\tilde\matC \in \R^{m \times c}$  with $c=O(k/\varepsilon)$ columns of $\matA$ satisfies w.p. $0.59 - \frac{s+1}{n} - 2 \delta$:
\begin{equation}\label{eqnthm1s}
 \FNormS{\matA -  \tilde\matC  \pinv{\tilde\matC}\matA} \le
\FNormS{\matA-\Pi_{\tilde\matC,k}^{\mathrm{F}}(\matA)}\le
(1 + O(\varepsilon)) \cdot \left(\sum_{i=k+1}^{\rank(\matA)}\sigma_i^2(\matA) \right).
\end{equation}
The matrix $\matU \in \R^{m \times k}$ with $k$ orthonormal columns satisfies w.p. $0.58 - \frac{s+1}{n} - 2 \delta$:
\begin{equation}\label{eqnthm2s}
 \FNormS{\matA -  \matU \matU\transp \matA} \le
(1 + O(\varepsilon)) \cdot \left(\sum_{i=k+1}^{\rank(\matA)}\sigma_i^2(\matA) \right).
\end{equation}
Let each column of $\matA$ has at most $\phi$ non-zero elements. Then, the communication cost of the algorithm is
$O\left(s k \phi  \varepsilon^{-1} + s k^3 \varepsilon^{-5}  \right)$
and the  running time is
$$O\left(\nnz(\matA) \cdot  \log^2\left(\frac{n s}{\delta} \right)+
(m+n) s \cdot poly(k,\varepsilon^{-1},\log\left(\frac{n s}{\delta} \right) \right).
$$
\end{theorem}

\subsection{Proof of Theorem~\ref{thm1s}}\label{sec:proof2}

%
\subsubsection{Proof of Eqn.~\ref{eqnthm1s}}
From Lemma~\ref{lem:adaptivecolumnss}, we have that with probability at least $0.9 - \frac{1}{n}$:
\begin{equation}\label{eqn:boundeqn}
\FNormS{ \matA - \tilde\matC \pinv{\tilde\matC}\matA}
\le
\FNormS{\matA-\Pi_{\tilde\matC,k}^{\mathrm{F}}(\matA)}
\le \sum_{i=k+1}^{\rank(\matA)} \sigma_i^2(\matA)  + \frac{60 \varepsilon}{50} \cdot \FNormS{\matA - \matC \pinv{\matC}\matA}.
\end{equation}
According to Lemma~\ref{lem:projcomb}, $\matG$ is a $(4k,1/2)$-projection-cost preserving sketch of $\matA$.
It means that there exists a constant $c$ which is independent from projection matrix $\matP$ and we have the following:
$$\frac12\FNormS{\matA-\matP\matA}\leq \FNormS{\matG-\matP\matG}+c \leq \frac32\FNormS{\matA-\matP\matA}$$
For convenience, we use $\matP_{\matG}$ to denote a rank-$k$ projection matrix which provides the best rank-$k$ approximation $\matG_k=\matP_{\matG}\matG$ to matrix $\matG$. We also use $\matP_{\matA}$ to denote a rank-$k$ projection matrix which provides the best rank-$k$ approximation $\matA_k=\matP_{\matA}\matA$ to matrix $\matA$. With probability at least $0.68$, the following will be held,
\begin{align*}
\frac12\FNormS{\matA-\matC\matC^\dagger\matA} &\leq \FNormS{\matG-\matC\matC^\dagger\matG}+c\\
&\leq 4820\FNormS{\matG-\matP_{\matG}\matG}+c\\
&\leq 4820\FNormS{\matG-\matP_{\matA}\matG}+c\\
&\leq 4820\left(\FNormS{\matG-\matP_{\matA}\matG}+c\right)\\
&\leq \frac{4820}2\FNormS{\matA-\matP_{\matA}\matA}\\
\end{align*}
The first inequality is due to Lemma~\ref{lem:projcomb}. $\matG$ is a $(4k,1/2)$-projection cost preserving sketch of $\matA$ and $\matC$ has rank at most $4k$. The second inequality is based on Lemma~\ref{thm:optimalFdets}. The third inequality is because $\matP_{\matG}$ provides the best rank-$k$ approximation to $\matG$. The fourth inequality used the fact $c$ is non-negative. The fifth inequality held because $\matG$ is a $(4k,1/2)$-projection cost preserving sketch of $\matA$. Thus, $$\FNormS{\matA-\matC\matC^\dagger\matA}\leq4820\FNormS{\matA-\matA_k}$$
\paragraph{Concluding the proof.}
Combining with Equation~\ref{eqn:boundeqn}, we have:
\begin{equation}\label{eqn3s}
\FNormS{ \matA - \tilde\matC \pinv{\tilde\matC} \matA }  \le
\FNormS{\matA-\Pi_{\tilde\matC,k}^{\mathrm{F}}(\matA)} \le
\sum_{i=k+1}^{\rank(\matA)} \sigma_i^2(\matA)  + O(\varepsilon) \cdot \sum_{i=k+1}^{\rank(\matA)} \sigma_i^2(\matA).
\end{equation}

\subsubsection{Proof of Eqn.~\ref{eqnthm2s}}
We would like to apply Lemma~\ref{lem:KVW2} for the matrix $\matU \in \R^{m \times k}$ in the algorithm.
Note that we already proved that the matrix $\tilde\matC$ in the algorithm satisfies  with probability $0.59 - \frac{s+1}{n} - 2 \delta$:
$$ \FNormS{\matA-\Pi_{\tilde\matC,k}^{\mathrm{F}}(\matA)} \le (1+O(\varepsilon)) \FNormS{\matA - \matA_k}.$$
Lemma~\ref{lem:KVW2} and a simple union bound conclude the proof.

\subsubsection{Running time}
\begin{enumerate}
\item {\bf Local Column Sampling}: $O(\nnz(A)\log(ms)+m\cdot \poly(k,\eps,\log(s)))$ arithmetic operations in total.
\begin{enumerate}
\item $O(\nnz(A_i)\log(ms)+m\cdot \poly(k,\eps,\log(s)))$ for each $\matS_i$
\item -
\end{enumerate}
\item {\bf Global Column Sampling}: $O\left( \nnz(\matA) \cdot \log k + m \cdot \poly(k, \varepsilon^{-1})  \right)$.
\begin{enumerate}
\item $O\left( \nnz(\matA) \cdot \log k + m \cdot poly(k, \varepsilon^{-1})  \right)$ to construct $\matC$ - from Lemma~\ref{thm:optimalFdets}.
\item -
\end{enumerate}
\item {\bf Adaptive Sampling}: $O(  \nnz(\matA) \cdot \log n + m \cdot s \cdot \poly(k,\varepsilon^{-1}, \log n) )$.
\begin{enumerate}
\item  First, we analyze the cost to compute the matrix $\tilde{\matPsi}_i$.
$O(k^2 m)$ for each $\pinv{\matC}$ locally; then, we compute
$(\matS \matA_i)  -  ( ( (\matS \matC_i) \pinv{\matC}_i) \matA_i )$. The costs are:
$\matS \matA_i : = \matD$ takes $O(\nnz(\matA_i))$ time,
$\matS \matC_i : = \matG$ takes $O(\nnz(\matA_i) \log(n))$ time,
$\matG \pinv{\matC}_i : = \matH$ takes $O(m \log(n) k/ \varepsilon^2)$ time,
$\matH \matA := \matL$ takes $O(\nnz(\matA_i) \log(n))$ time, and
$\matD - \matL$ takes $O(n \log n)$ time. So, the total time to compute one $\tilde{\matPsi}_i$ is
$O( \nnz(\matA_i) \log n + m \cdot poly(k,\varepsilon^{-1}, \log n) )$.

Also, there is a cost of computing $\psi_i,$ which is $O(\log n \cdot w_i)$ arithmetic operations.

\item $O(s)$ in total.
\item $O(k/\varepsilon)$ in total using the method of~\cite{Vos91} locally.
\item $O(k/\varepsilon)$.
\end{enumerate}
\item {\bf  Rank-$k$ matrix in $span(\tilde{\matC})$}:
$O( \nnz(\matA) + m \cdot s \cdot poly( k, \varepsilon^{-1})  )$ arithmetic operations in total.
\begin{enumerate}
\item $O(s m k^2/\varepsilon^2)$ in total.
\item $O(s n)$ in total to generate $s$ times the same $\matPsi$. Then another $O(n)$ in total
to generate the $\matD_i$'s. For all  $\matA_i \matW_i\transp$ we need
$O( \nnz(\matA) )$ operations and then
 another $O( m k^3 \varepsilon^{-7} )$ for each $\matY\transp \cdot (\matA_i \matW_i\transp)$
\item $O(s  k^3 \varepsilon^{-5} )$ to find $\matXi$; then another $O(s k^3\varepsilon^{-3}  + s k^4 \varepsilon^{-6} ))$
to update $\matXi$ locally in each machine and
another $O(s k^4 \varepsilon^{-6} )$ to
find $\matDelta$; and $O(smk^2/\varepsilon)$ to find $\matU$ ($s$ times).
\end{enumerate}
\end{enumerate}

\subsubsection{Communication Complexity}
Assume that  we can represent each element in the input matrix $\matA$ with at most
$b$ bits.  Also, we assume that one word has length $b$ bits and $b = O(\log( m n s / \varepsilon ))$.

\begin{enumerate}
\item {\bf Local Column Sampling}: $O(s k \phi)$ words.
\begin{enumerate}
\item -
\item $O(s k \phi)$ elements of $\matA$.
\end{enumerate}

\item {\bf Global Column Sampling}:$O(s k \phi)$ words.
\begin{enumerate}
\item -
\item $O(s k \phi)$ elements of $\matA$.
\end{enumerate}
\item {\bf Adaptive Sampling}: $O(s+ \phi k /\varepsilon)$ words.
\begin{enumerate}
\item $s$ integers each of which is representable with $O(\log k + \log \log(mns/\varepsilon))$ (from Lemma~\ref{lem:adaptiveSampling2}).

\item $s$ integers~(the $t_i$'s) each with magnitude at most $n$, hence representable with $b$ bits.
\item $O(\phi k/\varepsilon)$ elements of $\matA$.
\item -
\end{enumerate}

\item {\bf  Best rank-$k$ matrix in the span of $\tilde{\matC}$}:
$O( s\phi k/\varepsilon +s k^3  \varepsilon^{-5} )$ words.
\begin{enumerate}
\item $O(s k \phi \varepsilon^{-1})$ elements of $\matA$.
\item $O( s k^3  \varepsilon^{-5})$ numbers each of which can be represented with $b$ bits.
\item $O( s k^3  \varepsilon^{-5})$ numbers each of which can be represented with $b$ bits.
\end{enumerate}
\end{enumerate}
In total the communication complexity is $O( s k \phi / \varepsilon + sk^3/\varepsilon^5)$ words.

%

\section{Lower bounds}\label{sec:lower}

Only in Section~\ref{sec:low1}, we discuss the result in arbitrary partition model. 
In Section~\ref{sec:low2}, Section~\ref{sec:low3} and Section~\ref{sec:low4}, we discuss results in column partition model which provides stronger lower bounds than in arbitrary partition model.

\subsection{Dependence on $\Omega(\varepsilon^{-2})$ bits for distributed PCA in arbitrary partition model} \label{sec:low1}

\begin{lemma}[Theorem 8 in \cite{WZ16}]\label{lem:epslower}
There are two machines, and each of them hold a matrix $A_i\in \R^{m\times n}$. Let $\matA=\matA_1+\matA_2$. Suppose $\varepsilon^2 \leq 2n$, the communication of computing a rank-$k$ matrix $\matU$ with constant probability such that 
$$\FNormS{\matA-\matU\matU\transp\matA}\leq (1+\varepsilon)\cdot\FNormS{\matA-\matA_k}$$
needs $\Omega(1/\varepsilon^2)$ bits.
\end{lemma}

\begin{proof}[Proof. (from \cite{WZ16})]
The reduction is from GHD problem.

\begin{lemma}[Theorem 1.1 in~\cite{chakrabarti2012optimal}] \label{lem:GHD}
Each of two machines has an $n$-bit vector. There is a promise: either the inputs are at Hamming distance less than $n/2-c\sqrt{n}$ or greater than $n/2+c\sqrt{n}$. If they want to distinguish these two cases with constant probability, the communication required $\Omega(n)$ bits.
\end{lemma}

Without loss of generality, we assume $1/\varepsilon^{2}$ is an integer, and machine $1$ and machine $2$ have $\x,\y\in\{-1,1\}^{1/\varepsilon^2}$ respectively. There is a promise that either $\x\transp\y<-2/\varepsilon$ or $\x\transp\y>2/\varepsilon$ holds.
Consider about the following protocol:

\begin{enumerate}
\item Machine $1$ constructs $\matA_1\in\R^{(k+1)\times (1/\varepsilon^2+k)}$, and machine $2$ constructs $\matA_2\in\R^{(k+1)\times (1/\varepsilon^2+k)}$:
    
    \begin{align*}
\matA_1=\left(\begin{array}{ccccc}
\x\transp\varepsilon&0&0&...&0\\
{\bf 0}&\sqrt{2}&0&...&0\\
{\bf 0}&0&\frac{\sqrt{2(1+\varepsilon)}}{\varepsilon}&...&0\\
...&...&...&...&...\\
{\bf 0}&0&0&...&\frac{\sqrt{2(1+\varepsilon)}}{\varepsilon}
\end{array}\right)
& & 
\matA_2=\left(\begin{array}{ccccc}
\y\transp\varepsilon&0&0&...&0\\
{\bf 0}&0&0&...&0\\
{\bf 0}&0&0&0& 0\\
...&...&...&...&...\\
{\bf 0}&0&0&...&0
\end{array}\right)
\end{align*}

\item The server computes $\matU$ such that
$$\FNormS{\matA-\matU\matU\transp\matA}\leq (1+\varepsilon)\cdot\FNormS{\matA-\matA_k}$$
and sends $\matU$ to both of two machines, where $\matA=\matA_1+\matA_2$.

\item Let $\matP=\matI_{k+1}-\matU\matU\transp$. Each machine construct $\v=\matP^{(1)}/\|\matP^{(1)}\|_2$

\item Each machine checks whether $\v_1^2<\frac12(1+\varepsilon)$. If yes, return the case $\x\transp\y>2/\varepsilon$. Otherwise return the case $\x\transp\y<2/\varepsilon$
    
\end{enumerate}

Suppose the server successfully computes $\matU$. Since the rank of $\matA$ is at most $k+1$:
\eqan{
\FNormS{\matA-\matU\matU\transp\matA}
& = & \FNormS{\matP\matA} \\
& = & \|\v\matA\|_2^2 \\
& = & \v\matA\matA\transp\v\transp
}
An observation is that
$$\matA\matA\transp=\left(\begin{array}{ccccc}\|\x+\y\|_2^2\varepsilon^2&0&0&...&0\\0&2&0&...&0\\0&0&\frac{2(1+\varepsilon)}{\varepsilon^2}&...&0\\...&...&...&...&...\\0&0&0&...&\frac{2(1+\varepsilon)}{\varepsilon^2}\end{array}\right)$$
We have $\FNormS{\matA-\matA_k}\leq 2$. Then, 
\eqan{
\frac{2(1+\varepsilon)}{\varepsilon^2}\sum_{i=3}^{k+1}\v_i^2 
& = & \frac{2(1+\varepsilon)}{\varepsilon^2}(1-\v_1^2-\v_2^2)\\
& < &2(1+\varepsilon)
}
We have $\v_1^2+\v_2^2>1-\varepsilon^2$.
When $\x\transp\y>2/\varepsilon$, $\|\x+\y\|_2^2\varepsilon^2\geq2+4\varepsilon$
\eqan{
(1+\varepsilon)\FNormS{\matA-\matA_k}
& = & 2(1+\varepsilon) \\
& > & \v_1^2(2+4\varepsilon)+2\v_2^2 \\
& > & 2-2\varepsilon^2+4\varepsilon \v_1^2
}
So, $\v_1^2<\frac12(1+\varepsilon)$.
When $\x\transp\y<-2/\varepsilon$, $\|\x+\y\|_2^2\varepsilon^2\leq2-4\varepsilon$
\eqan{
(1+\varepsilon)\FNormS{\matA-\matA_k}
& = & \|\x+\y\|_2^2\varepsilon^2(1+\varepsilon) \\
& > & \v_1^2\|\x+\y\|_2^2\varepsilon^2+2\v_2^2\\
& = & 2(\v_1^2+\v_2^2)-(2-\|\x+\y\|_2^2\varepsilon^2)\v_1^2 \\
& > & 2(1-\varepsilon^2)-(2-\|\x+\y\|_2^2\varepsilon^2)\v_1^2
}
So, 
$$\v_1^2>\frac{2(1-\varepsilon^2)-\|\x+\y\|_2^2\varepsilon^2(1+\varepsilon)}{2-\|\x+\y\|_2^2\varepsilon^2}\geq\frac12(1+\varepsilon)$$
Therefore, machines can distinguish these two cases. The only communication is computing $\matU$. This cost shoulde be the same as $\Omega(1/\varepsilon^2)$ bits lower bound of the gap hamming distance problem.

\end{proof}

\subsection{Lower bounds for distributed PCA on dense matrices} \label{sec:low2}
This section provides a communication cost lower bound for the Distributed PCA problem of Definition~\ref{def:dpca}. Specifically, we  describe the construction of an $m \times n$ matrix $\matA$ (hard instance), and formally argue that for this $\matA,$ any $k \le 0.99 m,$ and any error parameter $C$ with $1 < C < poly(s k m),$
if there exists some algorithm to construct an $m \times k$ matrix $\matU$ such that, with constant probability,
$\FNormS{\matA - \matU \matU\transp \matA} \le C \cdot \FNormS{\matA - \matA_k},$ then this
algorithm has communication cost $\Omega(s k m)$ words.

\subsubsection{Preliminaries}
We use the notation ${\bf G}_{k,m}$ to denote the set of $k$-dimensional
subspaces of $\mathbb{R}^m$, which we identify with corresponding projector matrices
$\matQ \in \R^{m \times m}$, with rank$(\matQ) = k \le m$, onto
the subspaces, i.e.,
$${\bf G}_{k,m}= \{\matQ \in \R^{m \times m}: \matQ^2 = \matQ, \rank(\matQ) = k \le m. \}$$
We also need a description of a subset from ${\bf G}_{k,m}$:
$$C_{\delta}^k(\matP) = \{\matQ \in {\bf G}_{k,m}: \|\matP-\matQ\|_2 \leq \delta\}.$$
This is the set of those projectors $\matQ \in {\bf G}_{k,m}$ which differ in operator norm by at most $\delta$
from another fixed  projector $\matP \in {\bf G}_{k,m}$.
In our analysis below we also need a result from~\cite{kt13}:
\begin{theorem}(Net Bound - Corollary 5.1 of \cite{kt13})\label{thm:net}
For any $m,k,\delta > 0,$ there is a family $\mathcal{N} = \{\matP^1, \ldots, \matP^N\}$,
$N = 2^{\Omega(k(m-k) \log(1/\delta))}$, where $\matP^i \in {\bf G}_{k,m}$ and
$C_{\delta}^k(\matP^i) \cap C_{\delta}^k(\matP^j) = \emptyset,$ for all $i \neq j$.
\end{theorem}
Theorem \ref{thm:net} proves the existence of a large, but finite, set $\mathcal{N}$ of projection matrices, such that if one considers the ball of matrices of operator norm at most $\delta$ centered at each matrix in $\mathcal{N}$, then these balls are disjoint. Our hard instance matrix $\matA$ is constructed from some member in $\mathcal{N}$.

\subsubsection{Hard instance construction}\label{sec:hard}

Recall that in the column-partition model of  Definition~\ref{def:model} there are $s$ servers holding matrices $\matA_1, \ldots, \matA_s$, respectively, where $\matA_i$ has $m$ rows and some subset
of $w_i$ columns of some $m \times n$ matrix $\matA$. Notice that $\sum w_i = n$.
First of all, for arbitrary $m,s$ we assume\footnote{As, otherwise we can choose a
value $s' < s$ so that $n-m \leq s' m \leq n$ and apply the argument of Theorem~\ref{thm:main} with $s$ replaced with $s'$. Then, the lower bound in Theorem~\ref{thm:main} will then be $\Omega(s' m k)$, which assuming $n \geq 2d$, is an $\Omega(n)$ communication lower bound}
 that $n \geq s m$. Below, we describe a specific construction for a matrix $\matA$.

First of all, we set $\delta = 1/(skm)$, the parameter to be used in Theorem \ref{thm:net}. Next,
fix a family $\mathcal{N}$ of subspaces of ${\bf G}_{k,m}$ with the guarantees of Theorem \ref{thm:net}.
Let $\matP_{\matR} = \matR \matR\transp$ be a uniformly random member of $\mathcal{N}$,
where $\matR \in \mathbb{R}^{m \times k}$ has orthonormal columns (any projector can be expressed as
such a matrix, where the columns of $\matR$ span the subspace that $\matP$ projects onto).
The entries of $\matA_1$ are equal to the entries of $\matR$, each rounded to the nearest integer multiple of $1/B$,  where $B = \textrm{poly}(skm)$ is a large enough parameter specified later.
We denote this rounded matrix with $\tilde\matR \in \R^{m \times k}$. So, if $\alpha_{ij}$ is the $(i,j)$th entry in $\matA_1,$ then
$\alpha_{ij} = \frac{k_{min}}{B}$,  with $k_{min} = \argmin_{k \in \mathbb{Z}}  | R_{ij} - \frac{k}{B} |$.
Each $\matA_i$ for $i = 2,3,\dots,s-1$ is $1/B$ times the $m \times m$ identity matrix. Finally, $\matA_s$ is the $m \times t$ matrix of all zeros with $t = n - (s-1)m - k$. I.e.,
$$
\matA =
\begin{pmatrix}
\tilde{\matR} & \frac{1}{B}\matI_m & \frac{1}{B}\matI_m & \dots & \frac{1}{B}\matI_m & {\bf 0}_{m \times t}
 \end{pmatrix}.
$$

\subsubsection{Intermediate results}

First, we give a bound regarding the best rank $k$ approximation for the above matrix $\matA.$
\begin{lemma}\label{lem:bestrankA}
For the matrix $\matA \in \R^{m \times n}$ in Section~\ref{sec:hard}, and any $k \le 0.99 m:$
$\|\matA-\matA_k\|_\mathrm{F}^2 < \frac{sm}{B^2}.$
\end{lemma}
\begin{proof}
We have,
\begin{eqnarray*}
\|\matA-\matA_k\|_\mathrm{F}^2 & \leq & \|\matA-\matA_1 \matA_1^{\dagger} \matA\|_\mathrm{F}^2\\
& = & \sum_{i > 1} \|(\matI-\matA_1 \matA_1^{\dagger})\matA_i\|_\mathrm{F}^2\\
& \leq & \sum_{i > 1} \|\matA_i\|_\mathrm{F}^2\\
& = & (s-2) \cdot \frac{m}{B^2}\\
& < & \frac{sm}{B^2},
\end{eqnarray*}
where the first inequality uses the fact that $\matA_1\matA_1^{\dagger}$ is a matrix of rank at most $k$,
the first equality uses the fact that $(\matI-\matA_1 \matA_1^{\dagger}) \matA_i$ is the all-zeros matrix, the second inequality uses
that a projector cannot increase a unitarily invariant norm, and the third inequality follows by construction.
\end{proof}
Next, we bound, in the operator norm, the difference of $\matA_1$ from the matrix $\matR$.
The lemma follows from the fact that $\matA_1 \in \R^{m \times k}$ is obtained by
$\matR \in \R^{m \times k}$ by rounding
each entry to the nearest integer multiple of $1/B$, and then summing the squared differences
across all $km$ entries.

\begin{lemma}\label{lem:precision}(Precision Lemma)
$\|\matA_1 - \matR\|_2 \leq \frac{\sqrt{km}}{B}.$
\end{lemma}
\begin{proof}
$
\|\matA_1 - \matR\|_2^2 \leq \|\matA_1-\matR\|_\mathrm{F}^2 \leq \frac{km}{B^2}.
$
\end{proof}
Next, we prove a pure linear algebraic result.
The following lemma captures the fact that if some $\matP \in {\bf G}_{k,m}$ is close to $\matQ \matP$
(in Frobenius norm)
for
$\matQ \in {\bf G}_{k,m}$, then $\matP$ is also close to $\matQ$ (in Frobenius norm). We will use this lemma later for $\matP = \matP_{\matR}.$
\begin{lemma}\label{lem:errorMeasure}(Error Measure)
Let $\matP, \matQ  \in {\bf G}_{k,m}$
with $\|\matP-\matQ \matP\|_\mathrm{F}^2 \leq \Delta$.
Then $\|\matP-\matQ\|_\mathrm{F}^2 \leq 2\Delta$.
\end{lemma}
\begin{proof}
By the matrix Pythagorean theorem,
\begin{eqnarray}\label{eqn:pythag}
k = \|\matQ\|_\mathrm{F}^2 = \|\matQ \matP\|_\mathrm{F}^2 + \|\matQ(\matI-\matP)\|_\mathrm{F}^2.
\end{eqnarray}
Also by the matrix Pythagorean theorem,
\begin{eqnarray*}\label{eqn:pythag2}
k = \|\matP\|_\mathrm{F}^2 = \|\matQ\matP\|_\mathrm{F}^2 + \|(\matI-\matQ)\matP\|_\mathrm{F}^2,
\end{eqnarray*}
and so using the premise of the lemma, $\|\matQ \matP\|_\mathrm{F}^2 \geq k-\Delta$.
Combining with (\ref{eqn:pythag}),
\begin{eqnarray}\label{eqn:needed}
\|\matQ(\matI-\matP)\|_\mathrm{F}^2 \leq \Delta.
\end{eqnarray}
Hence,
\begin{eqnarray*}\label{eqn:triangle}
\|\matP-\matQ\|_\mathrm{F}^2 & = & \|(\matP-\matQ)\matP\|_\mathrm{F}^2 + \|(\matP-\matQ)(\matI-\matP)\|_\mathrm{F}^2 \\
& = & \|\matP-\matQ\matP\|_\mathrm{F}^2 + \|\matQ - \matQ\matP\|_\mathrm{F}^2 \\
& \leq & \Delta + \|\matQ(\matI-\matP)\|_\mathrm{F}^2 \\
& \leq & 2\Delta,
\end{eqnarray*}
where the first equality follows by the matrix Pythagorean theorem, the second
equality uses $\matP^2 = \matP$ (since $\matP$ is a projector), the third equality uses the
bound in the premise of the lemma, and the fourth inequality uses (\ref{eqn:needed}).
\end{proof}

\subsubsection{Main Argument}
Before presenting the main theorem, we give an intermediate technical lemma.

\begin{lemma}(Implication of Correctness)\label{lem:impl}
Let $\matA \in \R^{m \times n}$ be the hard instance matrix in Section~\ref{sec:hard} and let $\matA$ be distributed in $s$ machines in the way it was described in Section~\ref{sec:hard}:
$$
\matA =
\begin{pmatrix}
\tilde{\matR} & \frac{1}{B}\matI_m & \frac{1}{B}\matI_m & \dots & \frac{1}{B}\matI_m & {\bf 0}_{m \times t}
 \end{pmatrix}.
$$
Suppose $n = \Omega(sm)$, $k < \rank(\matA)$, and $1 < C < \textrm{poly}(skm)$.
$\matA$ is revealed to the machines by just describing the entries in $\matA$ and without specifying any other detail regarding the construction.
Given this $\matA$ and assuming further that each machine knows some projector matrix $\matQ \in \R^{m \times m}$ with rank at most $k$ such that $\|\matA-\matQ\matA\|_\mathrm{F} \leq C\|\matA-\matA_k\|_\mathrm{F}$, there is a
deterministic algorithm (protocol) which, upon termination, leaves on each machine the matrix
$\matR \in \R^{m \times k}$ which was used to construct $\matA_1$. %
\end{lemma}
\begin{proof}
First, we prove that $\|\matR\matR\transp-\matQ\|_2 < \frac{1}{2skm}$. Then, using this bound, we describe a protocol that deterministically reveals the matrix $\matR$  which was used to construct $\matA_1$.

Using Lemma~\ref{lem:bestrankA} and the premise in the lemma, it follows that:
$
\|\matA-\matQ\matA\|_\mathrm{F}^2 \leq C \frac{sm}{B^2},
$
which in particular implies that
\begin{eqnarray}\label{eqn:boundU}
\|\matA_1-\matQ\matA_1\|_\mathrm{F}^2 \leq C \frac{sm}{B^2}.
\end{eqnarray}
We further manipulate the term $\|\matA_1-\matQ\matA_1\|_\mathrm{F}$ as follows:
\begin{eqnarray*}
\|\matA_1-\matQ\matA_1\|_\mathrm{F} & = & \|(\matI-\matQ)\matA_1\|_\mathrm{F}\\
& = & \|(\matI-\matQ)\matR\matR\transp + (\matI-\matQ)(\matA_1-\matR\matR\transp)\|_\mathrm{F}\\
& \geq & \|(\matI-\matQ)\matR\matR\transp\|_\mathrm{F} - \|(\matI-\matQ)(\matA_1-\matR\matR\transp)\|_\mathrm{F}\\
& \geq & \|(\matI-\matQ)\matR\matR\transp\|_\mathrm{F} - \|\matI-\matQ\|_\mathrm{F} \|(\matA_1-\matR\matR\transp)\|_2 \\
& \geq & \|(\matI-\matQ)\matR\matR\transp\|_\mathrm{F} - \frac{\sqrt{km}}{B}(\sqrt{m} + \sqrt{k}),
\end{eqnarray*}
where the first inequality is the triangle inequality, the second
inequality uses sub-multiplicativity, and the third inequality uses the
triangle inequality, i.e., that $\|\matI-\matQ\|_\mathrm{F} \leq \|\matI\|_\mathrm{F} + \|\matQ\|_\mathrm{F}$,
and Lemma \ref{lem:precision}.
Combining this bound with (\ref{eqn:boundU}), it follows that
$\|(\matI-\matQ)\matR \matR\transp\|_\mathrm{F} \leq \frac{\sqrt{Csm}}{B} + \frac{2m\sqrt{k}}{B}.$
At this point we would like to apply Lemma \ref{lem:errorMeasure} to conclude that
$\|\matR\matR\transp-\matQ\|_\mathrm{F}$ is small. To do so, we need $\matR\matR\transp$ and $\matQ$ to be of rank $k$. While
$\matR$ has rank $k$ by construction, $\matQ$ may not. However,
increasing the rank of $\matQ$ cannot increase the error  $\|\matA-\matQ\matA\|_\mathrm{F}$.
Hence, if rank$(\matQ) < k$, given $\matQ,$ the protocol in the premise of the lemma allows
each server to locally add standard basis vectors to the column
space of $\matQ$ until the column space of $\matQ$ becomes $k$.
This involves no communication and each server ends up with the same new
setting of $\matQ$, which for simplicity we still denote with $\matQ$.
Therefore, $rank(\matQ)=k$. Applying Lemma \ref{lem:errorMeasure},
it now follows that
\vspace{-0.1in}
$$\|\matR\matR\transp-\matQ\|_2 \leq \|\matR\matR\transp-\matQ\|_\mathrm{F} \leq \sqrt{2} \left (\frac{\sqrt{Csm}}{B} + \frac{2m\sqrt{k}}{B} \right ).$$
By setting $B$ to be a sufficiently large poly$(skm)$,
we have $\|\matR\matR\transp-\matQ\|_2 < \frac{1}{2skm}$.

Therefore, given $\matQ$,
by Theorem \ref{thm:net} there is a unique matrix $\matP^i$ in $\mathcal{N}$
for which $\matQ \in C_{\delta}^k(\matP^i)$, with $\delta = \frac{1}{skm}$.
Indeed, since $\|\matR-\matQ\|_2 \leq \frac{1}{skm} < \delta/2$, there cannot
be two matrices $\matP^i$ and $\matP^j$ for $i \neq j$ for which $\matQ \in C_{\delta}^k(\matP^i)$
and $\matQ \in C_{\delta}^k(\matP^j)$ since then by the triangle inequality $\|\matP^i - \matP^j\|_2 < \delta$,
contradicting the construction of $\mathcal{N}$.
Hence, it follows that each machine can deterministically identify this $\matP^i,$ by enumeration over $\mathcal{N}$,
and therefore compute $\matR$, via, for example, an SVD.
\end{proof}
We are now fully equipped to present the main argument about the communication cost lower bound for distributed PCA.
\begin{theorem}(Main)\label{thm:main}
Let $\matA \in \R^{m \times n}$ be the hard instance matrix in Section~\ref{sec:hard} and let $\matA$ be distributed in $s$ machines in the way it was described in Section~\ref{sec:hard}:
$$
\matA =
\begin{pmatrix}
\tilde{\matR} & \frac{1}{B}\matI_m & \frac{1}{B}\matI_m & \dots & \frac{1}{B}\matI_m & {\bf 0}_{m \times t}
 \end{pmatrix}.
$$
Suppose $n = \Omega(sm)$, $k \leq .99m$, and $1 < C < \textrm{poly}(skm)$.
Assume that there is an algorithm (protocol) which succeeds with probability at least $2/3$ in having the
$i$-th server output $\matQ \matA_i$, for all $i$,
where $\matQ \in \R^{m \times m}$ is a projector matrix with rank at most $k$ such that $\|\matA-\matQ\matA \|_\mathrm{F} \leq C\|\matA-\matA_k\|_\mathrm{F}$. Then, this algorithm requires $\Omega(skm \log(skm))$ bits of communication.

Further, the bound holds even if the input matrices $\matA_i$ to the servers have all
entries which are integer multiples of $1/B$, for a value $B = \textrm{poly}(skm)$,
and bounded in magnitude by $1$, and therefore all entries can be specified with $O(\log(skm))$ bits.

Note that assuming a word size of $\Theta(\log(skm))$ bits, we obtain an $\Omega(skm)$ word lower bound.
\end{theorem}
\begin{proof}
If the protocol succeeds, then, since each
server outputs $\matQ \matA_i$, and $\matA_i$ is $1/B$ times the identity
matrix for all $i =2,3,...s-1$, each of those servers can compute $\matQ$. One of those
servers can now send this $\matQ$ to the first and the last server.
It follows by Lemma \ref{lem:impl}, that each server can identify $\matP_{\matR}$ and $\matR$.

Letting $\mathcal{E}$ be the event that the protocol succeeds,
and letting $\Pi^i$ be the ordered sequence
of all incoming and outgoing messages
at the $i$-th server, it follows that
$$I(\Pi^i ; \matQ \mid \mathcal{E}) = \Omega(\log |\mathcal{N}|)
= \Omega(km \log(smk)),$$
where $I(X ; Y \mid \mathcal{F}) = H(X \mid \mathcal{F}) - H(X \mid Y, \mathcal{F})$
is the mutual information between random variables $X$ and $Y$ conditioned on $\mathcal{F}$. To see the first equality in the above, observe that
$I(\Pi^i ; \matQ \mid \mathcal{E}) = H(\matQ \mid \mathcal{E})
- H(\matQ \mid \Pi^i, \mathcal{E})$, and conditioned on $\mathcal{E}$, there
is a uniquely determined matrix $\matP_{\matR}$ each server identifies,
which is uniform
over a set of size $|\mathcal{N}|$, and so
$H(\matQ \mid \mathcal{E}) \geq H(\matP_{\matR}) = \Omega(\log_2 |\mathcal{N}|)$.

Here, $H(X)$
is the Shannon entropy of a random variable $X$, given by $H(X)=\sum_x \Pr[X = x] \log(1/\Pr[X=x])$,
and $H(X \mid Y) = \sum_{y} \Pr[Y = y] H(X \mid Y = y)$ is the conditional Shannon entropy.

Hence, if $Z$
is an indicator random variable which is $1$ if and only if $\mathcal{E}$ occurs, then using that
$I(X ; Y) \geq I(X ; Y \mid W) - H(W)$ for any random variables $X, Y,$ and $W$, and that
$I(X ; Y \mid W) = \sum_{w} \Pr[W = w] I(X ; Y \mid W = w)$,
we have
\begin{eqnarray*}
I(\Pi^i ; \matQ) & \geq & I(\Pi^i ; \matQ \mid Z) - 1\\
& \geq & I(\Pi^i ; \matQ \mid Z = 1)\Pr[Z = 1] - 1\\
& = & \Omega(km \log(smk)),
\end{eqnarray*}
mapping $X$ to $\Pi^i$, $Y$ to $\matQ$, and $W$ to $Z$ in the mutual
information bound above.
It follows that
$$\Omega(km \log(smk)) \leq I(\Pi^i ; \matQ) \leq H(\Pi^i) \leq {\bf E}_{\matP_{\matR}, rand}[|\Pi^i|],$$
where $|\Pi^i|$ denotes the length, in bits, of the sequence $\Pi^i$,
and $rand$ is the concatenation of the private randomness of all $s$ players
Note that the entropy of $\Pi^i$ is a lower bound on the expected encoding length of $|\Pi^i|$, by the Shannon coding theorem.
By linearity of expectation, $\sum_i {\bf E}_{R, rand} \|\Pi^i| = \Omega(skm \log(smk))$, which
implies there exists an $\matR$ and setting of $rand$ for which $\sum_i \|\Pi^i| = \Omega(skm \log(smk))$.
It follows that the total communication is $\Omega(smk \log(smk))$ bits. 
\end{proof}
\subsection{Lower bounds for distributed PCA on sparse matrices} \label{sec:low3}
Applying the previous theorem with $m=\phi$ gives a communication lower bound for sparse matrices.
\begin{corollary}\label{thm:main2}
Let $\matA \in \R^{\phi \times n}$ be the hard instance matrix in Section~\ref{sec:hard} (with $m$ replaced with $\phi$ in noting the number of rows in $\matA$):
$$
\matA =
\begin{pmatrix}
\tilde{\matR} & \frac{1}{B}\matI_{\phi} & \frac{1}{B}\matI_{\phi} & \dots & \frac{1}{B}\matI_{\phi} & {\bf 0}_{\phi \times t}
 \end{pmatrix}.
$$
Suppose $n = \Omega(s \phi)$, $k \leq .99 \phi$, and $1 < C < \textrm{poly}(sk \phi)$. Now, for arbitrary $m$ consider the matrix $\hat\matA \in \R^{m \times n}$ which has $\matA$ in the top part and the all-zeros $(m-\phi) \times n$ matrix in the bottom part. $\hat\matA$ admits the same column partition as $\matA$ after padding with zeros:

$$
\hat\matA =
\begin{pmatrix}
\tilde{\matR}                        & \frac{1}{B}\matI_{\phi}         & \frac{1}{B}\matI_{\phi}         & \dots & \frac{1}{B}\matI_{\phi}           & {\bf 0}_{\phi \times        t} \\
{\bf 0}_{(m-\phi) \times k}  & {\bf 0}_{(m-\phi) \times \phi}  & {\bf 0}_{(m-\phi) \times \phi}  & \dots & {\bf 0}_{(m-\phi) \times \phi}    & {\bf 0}_{(m-\phi) \times t}
 \end{pmatrix}.
$$
We denote the new sub-matrices with $\hat\matA_i \in \R^{m \times w_i}$.

Assume that there is an algorithm (protocol) which succeeds with probability at least $2/3$ in having the
$i$-th server output $\hat\matQ \hat\matA_i$, for all $i$,
where $\hat\matQ \in \R^{m \times m}$ is a projector matrix with rank at most $k$ such that $\|\hat\matA-\hat\matQ\hat\matA \|_\mathrm{F} \leq C\|\hat\matA-\hat\matA_k\|_\mathrm{F}$. Then, this algorithm requires $\Omega(sk\phi \log(sk \phi))$ bits of communication.

Note that assuming a word size of $\Theta(\log(sk \phi))$ bits, we obtain an $\Omega(sk \phi)$ word lower bound.
\end{corollary}
%

%
%

\subsection{Lower bounds for distributed column-based matrix reconstruction}\label{sec:low4}
In this section we develop a communication lower bound for the Distributed Column Subset Selection Problem of Definition~\ref{def:dcssp2}.

\subsubsection{A Net of discretized Deshpande-Vempala hard instances}
We start by describing the construction of a special sparse matrix from the work of Desphande and Vempala \cite{DV06}, which was further studied by
\cite{BDM11a}.
This matrix gives a lower bound (on the number of columns need to be selected in order to achieve a certain accuracy)
for the standard column subset selection problem~(see the following Fact).
Suppose $\phi \geq 2k/\epsilon$.
Consider a $(\phi+1) \times \phi$ matrix $\matB$
for which the $i$-th column is $\e_1 + \e_{i+1}$, where $\e_i$ is the $i$-th standard basis vector in $\R^{\phi+1}$.
Now define an
$((\phi+1)k) \times (\phi k)$ matrix $\matA$ with $k$ blocks along the diagonal,
where each block is equal to the matrix $\matB$. Let $n = (\phi+1)k$.

\begin{fact}(Proposition 4 in \cite{DV06}, extended in Theorem 37 of
\cite{BDM11a})\label{fact:dv}
Let $\matA$ be as above. If $\matP$ is an $n \times n$ projection operator
onto any subset of at most
$k/(2\epsilon)$ columns of $\matA$,
then
$\FNormS{\matA - \matP\matA} > (1+\epsilon)\FNormS{\matA - \matA_k}.$
\end{fact}
%

To prove a communication lower bound for distributed column subset selection,
we will need to transform the above matrix $\matA$ by multiplying by a rounded orthonormal matrix, as follows.
We also need to prove a similar lower bound as in the previous fact for this transformed matrix.
\begin{lemma}\label{lem:transform}
Let $\matA$, $k$, $\phi$, $1/\eps$, and $n$ be as above.
For any $n \times n$ orthonormal matrix $\matL$, let $\tilde{\matL}$ denote the matrix formed
by rounding each of the entries of $\matL$ to the nearest integer multiple of $M$,
where $M \leq 1/((\phi+1)k)^c,$ for a sufficiently
large constant $c > 0$. If $\matP$ is an $n \times n$ projection
operator onto any subset of at most $k/(6\epsilon)$ columns of $\tilde{\matL} \matA$, then
$
\FNormS{\tilde{\matL} \matA - \matP \tilde{\matL} \matA} > (1+\epsilon) \FNormS{\tilde{\matL} \mat A - [\tilde{\matL} \matA]_k}.
$
\end{lemma}
\begin{proof}
We will show that, assuming $1/\eps < n/3$,
if $\matP$ is an $n \times n$ projection operator onto any subset of at most $k/(2\epsilon)$ columns of $\tilde{\matL}\matA$,
then $\FNormS{\tilde{\matL} \matA - \matP \tilde{\matL} \matA} > (1+\epsilon/3) \FNormS{\tilde{\matL} \mat A - [\tilde{\matL} \matA]_k}.$
The lemma will then follow by replacing $\eps$ with $3\eps$.

Suppose $S$ is a subset of $k/(2\epsilon)$ columns of $\tilde{\matL} \matA$
for which
\begin{eqnarray}\label{eqn:yuck}
\FNormS{\tilde{\matL}\matA - \matP \tilde{\matL}\matA} \leq (1+\epsilon/3) \FNormS{\tilde{\matL}\matA - [\tilde{\matL}\matA]_k},
\end{eqnarray}
where $\matP$ is the projection onto the columns in $S$
and $[\tilde{\matL}\matA]_k$ is the best rank-$k$ approximation to $\tilde{\matL}\matA$
%
The proof strategy is to
relate the left hand side of (\ref{eqn:yuck}) to $\FNormS{\matA - \matQ \matA}$ where $\matQ$ is the projection onto
the columns of $\matA$ indexed by the same index set $S$, and simultaneously to relate the right hand side of (\ref{eqn:yuck}) to
$\FNormS{\matA- \matA_k}$.

We start by looking at $\FNorm{\tilde{\matL} \matA - [\tilde{\matL} \matA]_k}$. Notice that
$\tilde{\matL} \matA = \matL \matA + \matE \matA$, where $\matE$ is a matrix whose entries are all bounded in magnitude by $M$.
Therefore, $\tilde{\matL} \matA = \matL \matA + \matF$, where $\matF = \matE \matA$ and
$$\|\matF\|_2 \leq \FNorm{\matE} \FNorm{\matA} \leq \left (n^2M^2 \right)^{1/2} \cdot (2(n-1))^{1/2} \leq \left (2n^3 M^2 \right )^{1/2},$$
where we have used $\FNorm{\matA} \le (2(n-1))^{1/2}$.
By Weyl's inequality (Corollary 7.3.8 of \cite{HJ85}), for all $i \leq \phi k$,
$$|\sigma_i(\matL \matA) - \sigma_i(\matL \matA + \matF)| \leq \|\matF\|_2 \leq \left (2n^3 M^2 \right )^{1/2},$$
where $\sigma_i(\matX)$ denotes the $i$-th singular value of a matrix $\matX$. Since $\matL$ is orthonormal,
$\sigma_i(\matL \matA) = \sigma_i(\matA)$ for all $i$. It follows that
\begin{eqnarray}\label{eqn:bound0}
\FNormS{\tilde{\matL} \matA - [\tilde{\matL} \matA]_k} & = & \sum_{i > k} \sigma_i^2(\tilde{\matL} \matA)\notag \\
& = & \sum_{i > k} \left (\sigma_i(\matL \matA) \pm \left (2n^3M^2 \right )^{1/2} \right )^2 \notag \\
& = & \sum_{i > k} \left (\sigma_i(\matA) \pm \left (2n^3M^2 \right )^{1/2} \right )^2 \notag \\
& = & \FNormS{\matA - \matA_k} \pm O(n \|\matA\|_2 \left (n^3 M^2 \right )^{1/2} + n^4 M^2) \notag \\
& = & \FNormS{\matA - \matA_k} \pm O(1/n^2),
\end{eqnarray}
where the final line follows for a sufficiently small choice of $M \leq 1/n^c$, using that $\|\matA\|_2 \leq \FNorm{\matA} = O(\sqrt{n})$.

We now look at $\FNorm{\tilde{\matL}\matA - \matP \tilde{\matL}\matA}$. By the triangle inequality,
\begin{eqnarray}\label{eqn:hallucinate}
\FNorm{\tilde{\matL}\matA - \matP \tilde{\matL}\matA} & \geq & \FNorm{\matL \matA - \matP \matL \matA} - \FNorm{\matE\matA} - \FNorm{\matP \matE \matA} \notag\\
& \geq & \FNorm{\matL \matA - \matP \matL \matA} - \FNorm{\matE} \FNorm{\matA} - \|\matP\|_2 \FNorm{\matE} \FNorm{\matA} \notag\\
& = & \FNorm{\matA - \matL\transp \matP \matL \matA} - O(1/n^3),
\end{eqnarray}
where the equality uses that multiplying by $\matL\transp$ preserves norms, as well as
that $\FNorm{\matE} \leq 1/\poly(n)$ for an arbitrarily small $\poly(n)$ by making $M \leq 1/n^c$
sufficiently small, whereas $\|\matP\|_2 \leq 1$ and $\FNorm{\matA} = O(\sqrt{n})$.
We claim $\matL\transp \matP \matL = \matQ$, where $\matQ$ is the projection onto the columns in $\matA$
indexed by the set $S$. To see this, let $\matU$ be an orthonormal basis for the columns in $\matA$ spanned by
the index set $S$, so that $\matP = \matL \matU\matU\transp \matL\transp$. Then $\matL\transp \matP \matL = \matU\matU\transp = \matQ$.
Hence, using that $\FNorm{\matA-\matQ\matA} \leq \FNorm{\matA} = O(\sqrt{n})$ and squaring (\ref{eqn:hallucinate}), we have,
\begin{eqnarray}\label{eqn:bound1}
\FNormS{\tilde{\matL}\matA - \matP \tilde{\matL}\matA} \geq \FNormS{\matA-\matQ\matA} - O(1/n^2).
\end{eqnarray}
Combining (\ref{eqn:yuck}), (\ref{eqn:bound0}), and (\ref{eqn:bound1}),
\begin{eqnarray}\label{eqn:combo}
\FNormS{\matA-\matQ\matA} \leq (1+\epsilon/3)\FNormS{\matA-\matA_k} + O(1/n^2).
\end{eqnarray}
Note that $\FNorm{\matA-\matA_k} \geq 1$ for any $k < n$, since if we remove the top row of $\matA$, creating a matrix $\matB$,
then $\FNorm{\matB- \matB_k} \leq \FNorm{\matA-\matA_k}$, yet $\matB$ is the identity matrix, and so $\FNorm{\matB-\matB_k} \geq 1$. It follows
that the additive $O(1/n^2)$ error in (\ref{eqn:combo}) translates into a relative error, provided $k < n$ and for $1/\eps < n/3$.
Hence, $\FNormS{\matA - \matQ\matA} \leq (1+\epsilon)\FNormS{\matA - \matA_k}$, and therefore by Fact \ref{fact:dv},
we have $|S| > k/(2\epsilon)$.
\end{proof}

\subsubsection{Column sharing of discretized orthonormal matrices}
We need the following technical lemma about sampling random orthonormal matrices, concerning the number of
columns they have in common after rounding.
\begin{lemma}\label{lem:randomOrthonormal}
Let $1 \leq r \leq n$, and let $M \leq 1/n^c$ for a sufficiently large absolute constant $c > 0$.
Suppose we choose two independently random orthogonal matrices $\matA \in \mathbb{R}^{n \times n}$ and $\matB \in \mathbb{R}^{n \times n}$ from the Haar measure (so they each have
orthonormal columns and orthonormal rows). We then
round each of the entries in $\matA$ to the nearest integer multiple of $M$, obtaining a matrix $\tilde{\matA}$,
and similarly obtain $\tilde{\matB}$. Let $\mathcal{E}$ denote the event that there exist $r$
distinct indices $i_1, \ldots, i_{r} \in [n]$ and $r$ distinct indices $j_1, \ldots, j_{r} \in [n]$
for which $\tilde{\matA}_{*, i_{\ell}} = \tilde{\matB}_{*, j_{\ell}}$ for all
$\ell = 1, 2, \ldots, r$, where $\tilde{\matA}_{*, i_{\ell}}$ denotes the $i_{\ell}$-th column of $\tilde{\matA}$.
Then,
$\Pr[\mathcal{E}] \leq e^{-\Theta(nr \log M)}.$
\end{lemma}
\begin{proof}
We fix a subset $S$ of $r$ distinct indices $i_1, \ldots, i_r \in [n]$ and a subset $T$ of $r$ distinct
indices $j_1, \ldots, j_r \in [n]$ and show that the probability
$\tilde{\matA}_{*, i_{\ell}} = \tilde{\matB}_{*, j_{\ell}}$
simultaneously for all $\ell \in [r]$ is at most $e^{-\Theta(nr \log M)}$. It then
follows by a union bound that
$$\Pr[\mathcal{E}] \leq \binom{n}{r} \cdot \binom{n}{r} \cdot e^{-\Theta(nr \log M)} \leq \left (\frac{ne}{r} \right )^{2r} e^{-\Theta(nr \log M)}
\leq e^{-\Theta(nr \log M)},$$
as desired.

Since $\matA$ and $\matB$ are independent and their distribution is permutation-invariant,
we can assume $i_{\ell} = j_{\ell} = \ell$ for $\ell \in [r]$.
Let $\mathcal{F}$ be the event that $\tilde{\matA}_{*, i_{\ell}} = \tilde{\matB}_{*, j_{\ell}}$ occurs for all $\ell \in [r]$.
If $\mathcal{F}$ occurs, then we need $\|\matA_{*, i_{\ell}} - \matB_{*,j_{\ell}}\|_{\infty} \leq M$ for all $\ell \in [r]$,
which implies $\|\matA_{*, i_{\ell}} - \matB_{*, j_{\ell}}\|_2 \leq \sqrt{n} M$ for all $\ell \in [r]$.
Let $\matA^r$ be the leftmost
$r$ columns of $\matA$, and $\matB^r$ the leftmost $r$ columns of $\matB$. Then, if $\mathcal{F}$ occurs
\begin{eqnarray*}
\|\matA^r - \matB^r\|_2^2  \leq  \FNormS{\matA^r - \matB^r}
 = \sum_{\ell=1}^r \|\matA_{*, \ell} - \matB_{*, \ell}\|_2^2
 \leq r n M^2.
\end{eqnarray*}
Note that $\matA^r$ and $\matB^r$ have orthonormal columns, and so $\matP_{\matA^r} = (\matA^r)(\matA^r)\transp$ and
$\matP_{\matB^r} = (\matB^r)(\matB^r)\transp$ are projection matrices. Then,
\begin{eqnarray}\label{eqn:oneToTwo}
\|\matA^r-\matB^r\|_2 & \geq & \|(\matA^r)(\matA^r)\transp - \matB^r(\matA^r)\transp\|_2/2
+ \|(\matA^r)(\matB^r)\transp - (\matB^r)(\matB^r)\transp\|_2/2 \notag\\
& = & \|(\matA^r)(\matA^r)\transp - (\matA^r)(\matB^r)\transp\|_2/2 + \|(\matA^r)(\matB^r)\transp - (\matB^r)(\matB^r)\transp\|_2/2 \notag \\
& \geq & \|(\matA^r)(\matA^r)\transp - (\matB^r)(\matB^r)\transp\|_2/2,
\end{eqnarray}
where the first inequality follows since multiplying by a matrix on the right with orthonormal
rows cannot increase the spectral norm, the equality follows by taking transposes, and the
second inequality follows by the triangle inequality.

We need another result from~\cite{kt13}.
\begin{theorem}(Probability Bound - Claim 5.2 of \cite{kt13})\label{thm:probBound}
Suppose we choose two independent random subspaces $Y$ and $S$ of $\mathbb{R}^n$ each of dimension $r$ from the Haar measure, where
we denote the projection operators onto the subspaces by $\matY\matY\transp$ and $\matS\matS\transp$, respectively, where
$\matY$ and
$\matS$ have $r$ orthonormal columns. Then for any $\delta \in (0,1)$,
$\Pr[\|\matY\matY\transp - \matS\matS\transp\|_2 \leq \delta] \leq e^{-\Theta(r(n-r)\log(1/\delta))}.$
\end{theorem}
Combining (\ref{eqn:oneToTwo}) with Theorem \ref{thm:probBound}, we have that for $r \leq n/2$,
$\Pr[\mathcal{F}] \leq e^{-\Theta(rn \log(M))},$ since by taking $M \leq 1/n^c$ sufficiently small, we can
ensure $\|\matA^r - \matB^r\|_2 \leq \sqrt{rn} M \leq M^{1/2}$, and so $2M^{1/2}$ upper bounds
$\|(\matA^r)(\matA^r)\transp - (\matB^r)(\matB^r)\transp\|_2$.

Note that if $r > n/2$, then in particular we still have that
$\tilde{\matA}_{*, i_{\ell}} = \tilde{\matB}_{*, j_{\ell}}$ for all $\ell \in [r/2]$, and therefore in this
case we also have $\Pr[\mathcal{F}] \leq e^{-\Theta(rn \log M)}$. Hence,
$\Pr[\mathcal{E}] \leq e^{-\Theta(rn \log M)}$, which completes the proof.
\end{proof}

We now extend Lemma \ref{lem:randomOrthonormal} to a $k$-fold version.
\begin{lemma}\label{lem:randomExtension}
Let $M \leq 1/n^c$ for a sufficiently large absolute constant $c > 0$.
Suppose we independently choose $k$ pairs of matrices $\matA^{z}, \matB^{z} \in \mathbb{R}^{n \times n}$, $z \in [k]$,
each with orthonormal columns (and hence also orthonormal rows). We then
round each of the entries in each $\matA^{z}$ and $\matB^{z}$ to the nearest integer multiple of $M$, obtaining
matrices $\tilde{\matA^{z}}$, and $\tilde{\matB^{z}}$. Let $r = (r_1, \ldots, r_k) \in (\mathbb{Z}^{\geq 0})^k$
and let $\mathcal{E}$ be the event that for each $z \in [k]$, there exist
distinct indices $i_1, \ldots, i_{r_z} \in [n]$ and $j_1, \ldots, j_{r_z} \in [n]$
for which $\tilde{\matA^{z}}_{*, i_{\ell}} = \tilde{\matB^{z}}_{*, j_{\ell}}$ for all
$\ell = 1, 2, \ldots, r_z$.
Then,
$\Pr[\mathcal{E}] \leq e^{-\Theta(n \sum_{z=1}^k r_z \log M)}.$
\end{lemma}
\begin{proof}
This follows by Lemma \ref{lem:randomOrthonormal}, and the fact that the matrices $\matA^1, \ldots, \matA^k, \matB^1, \ldots, \matB^k$
are jointly independent. Here we also use a union bound.
\end{proof}

\begin{corollary}\label{cor:net}
Suppose $1 \leq k < c_0n$ for an absolute constant $c_0 > 0$.
Let $r = (r_1, \ldots, r_k) \in (\mathbb{Z}^{\geq 0})^k$.
Let $M \leq 1/n^c$ for a sufficiently large constant $c > 0$.
There exists a set $\mathcal{N}$ of $e^{\Theta(nt \log M)}$ $k$-tuples $(\tilde{\matA_1}, \ldots, \tilde{\matA_k})$ of matrices
each of dimensions $n \times n$ with entries that
are integer multiples of $M$ and bounded in magnitude by $1$, such that for every vector $r = (r_1, \ldots, r_k) \in (\mathbb{Z}^{\geq 0})^k$
with $\sum_{z=1}^k r_z \leq t$, it holds that
for all distinct $k$-tuples $(\tilde{\matA_1}, \ldots, \tilde{\matA_k}), (\tilde{\matB_1}, \ldots, \tilde{\matB_k})$,
there exists a $z \in [k]$ for which there are no sets $S = \{i_1, \ldots, i_{r_z}\}, T = \{j_1, \ldots, j_z\} \subset [n]$
for which $\tilde{\matA}_{*, i_{\ell}} = \tilde{\matB}_{*, j_{\ell}}$ for all $\ell = 1, 2, \ldots, r_z$.
\end{corollary}
\begin{proof}
For a fixed vector $r$ and a random choice of $(\tilde{\matA^1}, \ldots, \tilde{\matA^k}), (\tilde{\matB^1}, \ldots, \tilde{\matB^k})$,
by the guarantee of Lemma \ref{lem:randomExtension}, the complement of the
event in the statement of the corollary
happens with probability at most $e^{-\Theta(nt \log M)}$.
The number of vectors
$r$ with $\sum_{z=1}^k r_z \leq t$ is at most $t^k$, and $t^k e^{-\Theta(nt \log M)} \leq e^{-\Theta(nt \log M)}$ provided $k < c_0n$ for an absolute constant $c_0 > 0$, so we can apply a union bound over
all pairs in a random choice of $e^{\Theta(nt \log M)}$ $k$-tuples.
\end{proof}

\subsubsection{Hard instance construction}\label{sec:hard2}
We have $s$ servers holding matrices $\matA_1, \ldots, \matA_s$, respectively, where $\matA_i$ has $m$ rows and a subset
of $w_i$ columns of an $m \times n$ matrix $\matA$, where $\sum w_i = n$. For our lower bound we assume $2/\eps \leq \phi$ and
$k\phi \leq m$. Note that for $\eps$ less than a sufficiently small constant, this implies that $k < c_0 k \phi$, and so the assumption
in Corollary \ref{cor:net} is valid (since the $n$ of that corollary is equal to $(k+1)\phi$
We set $M = 1/(mn)^c$ for a sufficiently
large constant $c > 0$.

We will have $\matA_1$ being an $m \times k(\phi-1)$ matrix for which all but the first $k\phi$ rows are zero.
We assume $2/\eps \leq \phi$.
On the first $k\phi$ rows, $\matA_1$ is block diagonal containing $k$ blocks, each block
being a $\phi \times (\phi-1)$ matrix.

To specify $\matA_1$, let the parameters $t$ and $n$ of Corollary \ref{cor:net} equal $k/(2\eps)$ and $\phi$, respectively.
That corollary gives us a net $\mathcal{N}$ of $e^{\Theta(k\phi(\log(mn))/\eps)}$ $k$-tuples of matrices each of dimensions $\phi \times \phi$ with
entries that are integer multiple of $M$ and bounded in magnitude by $1$.

For a random $k$-tuple $(\matB^1, \ldots, \matB^k)$ in $\mathcal{N}$,
we create an $m \times m$ block diagonal matrix $\matB$, which is $0$ on all but its first $k\phi$ rows and first $k\phi$ columns.
On its first $k\phi$ rows, it is
a block diagonal matrix $\matB$
whose blocks along the diagonal are $\matB^1, \ldots, \matB^k$, respectively. Our matrix $\matA_1$ is
then the matrix product of $\matB$ and an matrix $\matD$, where $\matD$ is $m \times k(\phi-1)$ is formed by taking
the $k\phi \times k(\phi-1)$ matrix of Fact \ref{fact:dv} and padding it with $m-k\phi$ rows which are all zeros.
By Lemma \ref{lem:transform},
no $k/(2\epsilon)$ columns of $\matA_1$ contain a $k$-dimensional
subspace in their span which is a $(1+\epsilon)$-approximation to the best rank-$k$ approximationt to $\matA_1$.

Each $\matA_i$ for $i = 2,3,\dots,s-1$ is the $m \times m$ matrix of all zeros.
Finally, $\matA_s$ is the $m \times t$ matrix of all zeros with $t = n - (s-1)m - k(\phi-1)$, i.e.,
$
\matA =
\begin{pmatrix}
\matA_1 & {\bf 0}_{m \times m} & {\bf 0}_{m \times m} & \dots & {\bf 0}_{m \times m} & {\bf 0}_{m \times t}
 \end{pmatrix}.
$
By construction, each column of $\matA$ has at most $\phi$ non-zero entries, since the only non-zero columns
are those in $\matA_1$, and each column in $\matA_1$ has the form $\matB^z_{*,1} + \matB^z_{*,i}$ for a $z \in [k]$
and an $i \in [\phi]$, so it has at most $\phi$ non-zero entries.

\subsubsection{Main theorem}
In the Distributed Column Subset Selection
Problem in Definition~\ref{def:dcssp1}, a correct protocol should have each of the $s$ machines simultaneously output a matrix $\matC \in \mathbb{R}^{m \times c}$ with $c < n$ columns of $\matA$
such that
$$\FNormS{\matA-\matC \matC^{\dagger}\matA} \leq (1+\eps) \FNormS{\matA - \matA_k}.$$

\begin{theorem}\label{lower:cssp1}
Assume $2/\eps \leq \phi$ and $k\phi \leq \min(m,n)$.
Then, any, possibly randomized, protocol $\Pi$ which succeeds with probability at least $2/3$ in solving the Distributed Column Subset Selection Problem of Definition~\ref{def:dcssp1},
under
the promise that each column of $\matA$ has at most $\phi$ non-zero entries which are integer multiples of $1/(mn)^c$ and bounded in magnitude by $1$,
for a constant $c> 0$,
requires $\Omega(s\phi k (\log(mn))/\eps)$ bits of communication. If the word size is $\Theta(\log(mn))$, this implies the communication cost of the porticol
is $\Omega(s \phi k/\eps)$ words.
\end{theorem}
\begin{proof}
By construction of $\matA_1$, whenever $\Pi$ succeeds,
each of the $s$ machines outputs the same set $\matC$ of at least $t \geq k/(2\epsilon)$ columns of $\matA_1$.

Let $\Pi^i$ denote the transcript (sequence of all messages) between the coordinator and the $i$-th player.
We look at the information $\Pi^i$ reveals about $\matA_1$, given the event $\mathcal{E}$ that $\Pi$ succeeds, that is,
$I(\Pi^i ; \matA_1 \mid \mathcal{E})$.

A critical observation is that any column $\matC_{*, i}$ of $\matC$, is equal to $\matB^z_{*,1} + \matB^z_{*, i}$ for some $z \in [k]$,
where $\matB^z$ is as defined above (i.e.,
$\matA_1 = \matB \cdot \matD$). Hence, given $\matB^z_{*,1}$ and $\matC_{*,i}$, one can reconstruct $\matB^z_{*,i}$. Now, for any random variable $W$,
\begin{eqnarray}\label{eqn:infoNew}
I(\Pi^i ; \matA_1 \mid \mathcal{E}) \geq I(\Pi^i ; \matA_1 \mid \mathcal{E}, W) - H(W).
\end{eqnarray}
We let $W = (\matB^1_{*,1}, \ldots, \matB^k_{*,1})$, and observe that $H(W) \leq \log_2 ((2M)^{k\phi})$ since there are $2M$ choices for each coordinate
of each $\matB^z_{*,1}$. Hence, $H(W) = O(k \phi \log (mn))$. Note that $I(\Pi^i ; \matA_1 \mid \mathcal{E}, W) = H(\matA_1) = \log_2(|\mathcal{N}|)$, since
conditioned on the protocol succeeding, and given $W$, by Corollary \ref{cor:net} we can reconstruct $\matA_1$. By (\ref{eqn:infoNew}),
it follows that
$$(\Pi^i ; \matA_1 \mid \mathcal{E}) \geq \Theta(\phi k(\log(mn))/\epsilon) - O(k \phi \log(mn)) = \Omega(\phi k (\log(mn))/\epsilon).
$$
Note that if $Z$ is an indicator random variable that indicates that $\mathcal{E}$ occurs, we have $H(Z) \leq 1$, and so
$$
I(\Pi^i ; \matA_1) \geq I(\Pi^i; \matA_1 \mid \mathcal{E})\Pr[\mathcal{E}] - 1 = \Omega(\phi k (\log(mn))/\epsilon),
$$
since $\Pr[\mathcal{E}] \geq 2/3$.

It follows that
$$
\Omega(\phi k (\log(mn))/\epsilon) \leq I(\Pi^i ; \matA_1) \leq H(\Pi^i) \leq {\bf E}_{\matA_1, rand}[|\Pi^i|],
$$
where $|\Pi^i|$ denotes the length, in bits, of the sequence $\Pi^i$ of messages,
and $rand$ is the concatenation of the private randomness of all $s$ machines.
Note that the entropy of $\Pi^i$ is a lower bound on the expected encoding length of $|\Pi^i|$, by the Shannon coding theorem.
By linearity of expectation,
$$\sum_i {\bf E}_{\matA_1, rand} |\Pi^i| = \Omega(\phi k (\log(mn))/\epsilon),$$
which
implies there exists an $\matA_1$ and setting of $rand$ for which
$$\sum_i |\Pi^i| = \Omega(s\phi k(\log(mn))/\epsilon).$$
It follows that the total communication is $\Omega(s\phi k (\log(mn))/ \epsilon)$ bits.
\end{proof}



\begin{corollary}\label{lower:cssp2}
Assume $2/\eps \leq \phi$ and $k\phi \leq \min(m,n)$.
Then, any, possibly randomized, protocol $\Pi$ which succeeds with probability at least $2/3$ in solving the Distributed Column Subset Selection Problem - rank $k$ subspace version~(see Definition~\ref{def:dcssp2}),
 promise that each column of $\matA$ has at most $\phi$ non-zero entries which are integer multiples of $1/(mn)^c$ and bounded in magnitude by $1$,
for a constant $c> 0$,
requires $\Omega(s\phi k (\log(mn))/\eps)$ bits of communication. If the word size is $\Theta(\log(mn))$, this implies the total communication across
all $s$ machines is $\Omega(s \phi k/\eps)$ words.
\end{corollary}
\begin{proof}
As any such protocol is also a protocol for the Distributed Column Subset Selection Problem of Definition~\ref{def:dcssp1}~(the problem in Definition~\ref{def:dcssp2} is more general/difficult than the problem of Definition~\ref{def:dcssp1}), the proof follows immediately from Theorem \ref{lower:cssp1}.
\end{proof}

\end{document}